\def\draft{1}
\documentclass[11pt]{article}
\usepackage{amsfonts}
\usepackage{amsmath}
\usepackage{amssymb}
\usepackage{amsthm}
\usepackage{mathabx}
\usepackage{thmtools, thm-restate}
\usepackage{bbm}
\usepackage{bm}
\usepackage{dsfont}
\usepackage{fullpage}
\usepackage{tikz}
\usepackage[most]{tcolorbox}
\usepackage{xcolor}
\usepackage{ytableau}

\usepackage{lineno}

\newcommand*{\myfont}{\fontfamily{bch}\selectfont}
\DeclareTextFontCommand{\textmyfont}{\myfont}

\usepackage[linesnumbered,ruled,vlined]{algorithm2e} 
\SetKwComment{Comment}{/* }{ */}

\SetCommentSty{mycommfont}

\usepackage[hidelinks]{hyperref}
\usepackage[nameinlink]{cleveref}
\usepackage{url}            
\usepackage{footnotebackref}

\newtheorem{theorem}{Theorem}[section]
\newtheorem{corollary}[theorem]{Corollary}
\newtheorem{lemma}[theorem]{Lemma}
\newtheorem{observation}[theorem]{Observation}
\newtheorem{proposition}[theorem]{Proposition}
\newtheorem{definition}[theorem]{Definition}
\newtheorem{claim}[theorem]{Claim}

\newtheorem{remark}[theorem]{Remark}

\newtheorem{example}{Example}[theorem]

\newcommand{\prob}[2]{\mathop{\mathrm{Pr}}_{#1}[#2]}

\newcommand{\poly}{\mathop{\mathrm{poly}}}

\newcommand{\F}{\mathbb{F}}
\newcommand{\R}{\mathbb{R}}
\newcommand{\N}{\mathbb{N}}
\newcommand{\mc}[1]{\mathcal{#1}}

\newcommand{\Boo}{\{0,1 \}}

\newcommand{\sgrid}{\mathcal{S}}
\newcommand{\bigO}{\mathcal{O}}

\newcommand{\veca}{\mathbf{a}}
\newcommand{\vecb}{\mathbf{b}}
\newcommand{\Sym}{\mathrm{Sym}}

\newcommand{\J}{\mathcal{J}}
\newcommand{\A}{\mathcal{A}}
\newcommand{\Q}{\mathcal{Q}}

\newcommand{\LM}{\mathrm{LM}}

\newcommand{\paren}[1]{\left( #1 \right)}
\newcommand{\brac}[1]{\left[ #1 \right]}
\newcommand{\set}[1]{\left\{ #1 \right\}}

\newcommand{\setcond}[2]{\left\{ #1 \;\middle\vert\; #2 \right\}}

\newcommand{\abs}[1]{\left\lvert#1\right\rvert}

\DeclareMathOperator*{\E}{\mathbb{E}}

\newcommand{\x}{\mathbf{x}}
\newcommand{\y}{\mathbf{y}}
\newcommand{\z}{\mathbf{z}}
\newcommand{\ODLSZ}{\mathrm{ODLSZ}}

\newcommand{\by}{\mathbf{y}}

\newcommand{\bern}{\textnormal{Bern}}

\newcommand{\Z}{\mathbb{Z}}

\definecolor{thmcolor}{RGB}{235, 235, 235}
\definecolor{citecolor}{RGB}{1, 210, 56}
\definecolor{lemmacolor}{RGB}{130, 169, 252}

\newtcolorbox{algobox}{colback=lightgray!5!white,colframe=lightgray!75!black}
\newtcolorbox{thmbox}{colback=thmcolor!5!white,colframe=black!75!black}
\newtcolorbox{lemmabox}{colback=lemmacolor!5!white,colframe=blue!75!blue}

\usepackage{setspace}
\usepackage{enumerate}
\usepackage[parfill]{parskip}
\usepackage{changepage}
\usepackage{framed}

\allowdisplaybreaks
\hypersetup{
	colorlinks,
	linkcolor={blue},
	citecolor={citecolor},
	urlcolor={blue}
}

\usepackage[
backend=biber,
style=alphabetic,
sorting=nyt,
backref=true,
maxcitenames = 8,
mincitenames = 5,
maxalphanames = 8,
minalphanames = 5,
maxnames = 10,
minnames = 5
]{biblatex}

\ifnum\draft=1
\newcommand{\anote}[1]{{\color{brown} [Amik: #1]}}
\newcommand{\mnote}[1]{{\color{red} [Madhu: #1]}}
\newcommand{\pnote}[1]{{\color{blue} [Prashanth: #1]}}
\newcommand{\snote}[1]{{\color{teal} [Srikanth: #1]}}

\else  
\newcommand{\anote}[1]{}
\newcommand{\mnote}[1]{}
\newcommand{\pnote}[1]{}
\newcommand{\snote}[1]{}
\fi

\def\anon{0}
\date{June 23, 2025} 

\begin{document} 
\title{Eigenvalue Bounds for Symmetric Markov Chains on Multislices With Applications}

    \if\anon1{}\else{    
    \author{Prashanth Amireddy\thanks{School of Engineering and Applied Sciences, Harvard University, Cambridge, Massachusetts, USA. Supported in part by a Simons Investigator Award and NSF Award CCF 2152413 to Madhu Sudan and a Simons Investigator Award to Salil Vadhan. Email: \texttt{pamireddy@g.harvard.edu}} \and
    Amik Raj Behera\thanks{Department of Computer Science, University of Copenhagen, Denmark. Supported by Srikanth Srinivasan's start-up grant from the University of Copenhagen. Email: \texttt{ambe@di.ku.dk} } \and
     Srikanth Srinivasan \thanks{Department of Computer Science, University of Copenhagen, Denmark. Supported by the European Research Council (ERC) under grant agreement no. 101125652 (ALBA). Email: \texttt{srsr@di.ku.dk} } \and 
     Madhu Sudan\thanks{School of Engineering and Applied Sciences, Harvard University, Cambridge, Massachusetts, USA. Supported in part by a Simons Investigator Award, NSF Award CCF 2152413 and AFOSR award FA9550-25-1-0112. Email: \texttt{madhu@cs.harvard.edu}}}
     }\fi

	\maketitle
        \pagenumbering{arabic}

\begin{abstract}

We consider random walks on ``balanced multislices''
of any ``grid'' that respects the ``symmetries'' of the grid, and show that a broad class of such walks are good spectral expanders. (A grid is a set of points of the form $\mathcal{S}^n$ for finite $\mathcal{S}$, and a balanced multi-slice is the subset that contains an equal number of coordinates taking every value in $\mathcal{S}$. A walk respects symmetries if the probability of going from $u = (u_1,\ldots,u_n)$ to $v = (v_1,\ldots,v_n)$ is invariant under simultaneous permutations of the coordinates of $u$ and $v$.) Our main theorem shows that, under some technical conditions, every such walk where a single step leads to an almost $\mathcal{O}(1)$-wise independent distribution on the next state, conditioned on the previous state, satisfies a non-trivially small singular value bound.\\

We give two applications of our theorem to error-correcting codes: (1) We give an analog of the Ore-DeMillo-Lipton-Schwartz-Zippel lemma for polynomials, and junta-sums, over balanced multislices. (2) We also give a local list-correction algorithm for $d$-junta-sums mapping an arbitrary grid $\mathcal{S}^n$ to an Abelian group, correcting from a near-optimal $(\frac{1}{|\mathcal{S}|^{d}} - \varepsilon)$ fraction of errors for every $\varepsilon > 0$, where a $d$-junta-sum is a sum of (arbitrarily many) $d$-juntas (and a $d$-junta is a function that depends on only $d$ of the $n$ variables).\\ 

Our proofs are obtained by exploring the representation theory of the symmetric group and merging it with some careful spectral analysis.

\end{abstract} 

        \newpage 
        
\tableofcontents

        \newpage

\section{Introduction}\label{sec:intro}

Consider the following natural random walk whose states are the balanced vectors of $\{0,1\}^n$, i.e., the balanced Boolean slice with an equal number of $0$s and $1$s, where a single step of the random walk takes a state $u$ to a state $v$ at Hamming distance exactly $n/2$ from it. One would expect this random walk to mix extremely rapidly, and indeed this is known. The underlying graph here is a special case of a Johnson graph whose entire eigenspectrum is well known~\cite{Delsarte} and, in particular, implies that the second eigenvalue of this graph is $o_n(1)$.

Now consider the following variant of the above random walk: The states now are elements of the `balanced multislice' in $\{-1,0,1\}^n$, i.e. vectors with exactly $1/3$rd fraction of the letters $-1$, $0$ and $1$, and in a single step from a balanced vector $u$ to a random balanced vector $v$ obtained by flipping exactly $1/3$ fraction of each of the letters of $u$ to $-1$, $0$ and $1$. (So for a single coordinate $i$, $v_i$ is uniform in $\{-1,0,1\}$ conditioned on $u_i$.) It is intuitive to believe that such a random walk should also converge to the uniform distribution over all balanced vectors extremely fast, but, as far as we know, it was not even known that the second-eigenvalue of this random walk (or its transition probability matrix) has value $o_n(1)$. 

The gap in the understanding between the Boolean and non-Boolean cases in such problems can be significant for fundamental reasons. For example, for the alternate version of the random walk where the transition is defined by a uniformly random transposition of coordinates, it took a decade after optimal bounds on the mixing time were proved in the Boolean case \cite{Diaconis-Shahshahani} to prove similar results in the non-Boolean setting \cite{Scarabotti}. We refer the reader to the work of Filmus, O'Donnell, and Wu \cite{FilmusOdonnellWu} for a nice overview of the challenges posed by the non-Boolean setting in such problems. Some of these obstructions have to do with associated representations of the symmetric group that play a role in the corresponding proofs; these representations are simpler (`multiplicity-free') in the Boolean setting than in the non-Boolean setting. This also creates difficulties in resolving the questions we consider.

The main contribution of this work is to address some of the challenges alluded to above. In particular, we show that the variant random walk described in the second paragraph also has fast mixing, specifically by giving a $o_n(1)$ bound on its second eigenvalue. Indeed, we study this question in more generality for balanced multislices, with ``nearly balanced moves''. We believe the questions carry intrinsic interest and should find broad applications in the field. We justify this belief partially by describing two applications in coding theory:
\begin{itemize}
    \item The first gives a near-tight distance bound on codes obtained by evaluations of polynomials on balanced multislices.
    \item The second gives a local list-correction algorithm for subclasses of polynomials evaluated on {\em grids}. (Note that the second application does not refer to balanced multislices in the problem definition --- the multislices arise naturally in the design and analysis of the local correction algorithm!)
\end{itemize}
We elaborate on our setting and results, the applications, and the technical challenges below.

\subsection{Multislices and Random Walks}

By a {\em grid}, we refer to sets of the form $\sgrid^n$ for some finite set $\sgrid$ and positive integer $n$. (Usually we think of $s := |\sgrid|$ as a constant and study the growth of relevant parameters as a function of $n$). The {\em balanced multislice} of a grid $\sgrid^n$ is the set
\begin{align*}
     \sgrid^n_\mu := \setcond{\veca \in \sgrid^n }{\forall \sigma \in \sgrid, |\{i\in[n] | a_i = \sigma\}| = \frac{n}{s}}.
\end{align*}
(Note that a multislice is non-empty if and only if $s$ divides $n$. We will drop the term ``balanced'' in the future and simply refer to multislices to keep the term short.) 

The random walks we consider have the multislice of some grid as their state space. Recall that such a random walk can be described by a $\sgrid^n_\mu \times \sgrid^n_\mu$  matrix $W$ with $W(\veca,\vecb)$ denoting the probability of transition from state $\veca$ to $\vecb$. We consider walks where every step of the walk makes a ``nearly balanced move''.
To elaborate, let us define the {\em generalized Hamming distance}\footnote{Sometimes, this is also called as a ``meet table'' in algebraic combinatorics.} $\Delta(\veca,\vecb)$, for vectors $\veca,\vecb \in \sgrid^n$, to be the $\sgrid \times \sgrid$ matrix given by $\Delta_{\sigma,\tau}(\veca,\vecb) = |\{i\in [n] | a_i = \sigma, b_i = \tau\}|$. 
We say that a generalized Hamming distance parameter $\Delta \in \Z^{\mathcal{S} \times \mathcal{S}}$ {\em determines} a random walk matrix, denoted $W_\Delta$, if for each vertex ${\bf a}$, the random step corresponding to $W_\Delta$ is obtained by picking, uniformly at random, a vertex ${\bf b}$ on the multislice such that $\Delta({\bf a},{\bf b})=\Delta$. 

For constant $C < \infty$, we say that a generalized Hamming distance parameter $\Delta \in \Z^{\mathcal{S} \times \mathcal{S}}$ is {\em $C$-balanced} if each entry of $\Delta$ is $\frac{m}{s} \pm \paren{C\sqrt{m\log m}}$ where $m=n/s$. In other words, all the entries of $\Delta$ are equal up to a difference of at most $2C\sqrt{m \log m}$.  Informally, when considering $n \to \infty$ we refer to $\Delta$, as also 
 a random walk matrix $W_\Delta$ determined by  $\Delta$, as ``nearly balanced'' if $\Delta$ is $C$-balanced for some constant $C$. Here, we note that $W_\Delta$ is a well-defined random walk matrix over the multislice only if $\Delta/m$ is a doubly-stochastic matrix (i.e., every row and every column of $\Delta$ sums to $m$).

Note that the mixing time of a random walk matrix $W$ is closely tied to the second largest singular value, which we denote by $\sigma_2(W)$. (In particular, the singular values satisfy $1 = \sigma_1 \geq \sigma_2 \geq \cdots \sigma_N \geq 0$ and we let $\sigma_2(W) = \sigma_2$. If the walk is symmetric, then this captures the second eigenvalue. Specifically, if the eigenvalues are $1 = \lambda_1 \geq \lambda_2 \geq \cdots \lambda_N \geq -1$ where $N = |\sgrid_\mu^n|$, then $\sigma_2(W) = \max\{|\lambda_2|,  |\lambda_N|\}$.) Our main goal is to bound the value of $\sigma_2(W)$ by some function $o_n(1)$ that tends to 0 with growing $n$ for a broad class of random walk matrices $W$ over the multislice $S_\mu^n$. In general, it is desirable to have such singular value bounds, and such random walk matrices are said to have good ``spectral expansion'' or ``fast mixing''.

The following theorem gives such a fast mixing result on the balanced multislice for all nearly balanced walks that ``respect symmetries''. More formally, for a permutation $\pi \in \mathrm{Sym}_n$ and $\mathbf{a} \in \sgrid^{n}$, let $\pi(\veca)$ denote the action of $\pi$ on $\sgrid^n_\mu$, i.e., $\pi (\mathbf{a}) := (a_{\pi^{-1}(1)},\ldots,a_{\pi^{-1}(n)})$. 
For a stochastic matrix $M \in \R^{\sgrid^n_\mu \times \sgrid^n_\mu}$, we say $M$ {\em \underline{respects symmetries}} if for all 
permutations $\pi \in \mathrm{Sym}_{n}$ and for all $\veca,\vecb \in \sgrid^n_\mu$ we have $M(\pi(\veca),\pi(\vecb)) = M(\veca,\vecb)$. Our main theorem shows that walks that respect symmetries and are nearly balanced have fast mixing.

\begin{thmbox}
    \begin{restatable}[{\bf Singular value bound for nearly balanced walks}]{theorem}{genbalcor} \label{cor:gen-bal-cor}
        For every $s\ge 2$ and $C < \infty$, there exists $\tau > 0$ such that for every finite set $\sgrid$ of size $s$ and sufficiently large $n\in \N$, the following holds:\newline
        If $W$ is a stochastic matrix over the multislice $\sgrid^n_\mu$ that {\em respects symmetries}, and satisfies the condition that 
        \begin{align*}
            W({\mathbf{a}},{\mathbf{b}}) > 0 \quad \Rightarrow \quad \Delta({\mathbf{a}},{\mathbf{b}})  \text{ is $C$-balanced}  \quad \forall \; {\mathbf{a}},{\mathbf{b}} \in \sgrid^{n}_{\mu},
        \end{align*}
        then $\sigma_2(W) \le 1/n^\tau$.
    \end{restatable}
\end{thmbox}

The above result implies that the random walk on the balanced multislice mentioned earlier in this section (which corresponds to $s=3$ and $C$-balanced generalized distance parameter with $C=0$) has its second largest eigenvalue polynomially bounded. In fact, \Cref{cor:gen-bal-cor} is more general and covers multislices over any grid $\sgrid$ of constant size (i.e., for every $|\sgrid| = \bigO(1)$). Additionally, it is robust to perturbations of transition probabilities as long as the transition probabilities are nearly balanced.

Indeed \Cref{cor:gen-bal-cor} follows from our main technical theorem, stated as \Cref{thm:more-general-matrix-eigenvalue} below, which abstracts the main properties that suffice to prove the bound on the second largest singular value. Specifically \Cref{thm:more-general-matrix-eigenvalue} shows that, in addition to the symmetries respected by the matrix $W$, the important features that suffice to prove fast mixing are:
\begin{enumerate}
    \item The next state of the random walk is ``almost $\bigO(1)$-wise independent'' conditioned on the current state
    \item The Frobenius norm of $W$ is polynomially bounded in $n$.
\end{enumerate}
We elaborate on these conditions below before stating our main technical result \Cref{thm:more-general-matrix-eigenvalue}.\\

\noindent
For a distribution $D$ supported on $\sgrid^n$ and set $T \subseteq [n]$, we let $D_T$ denote the marginal distribution supported on $\sgrid^T$ induced by projecting a random variable $x \sim D$ to its coordinates in $T$. Recall that a distribution $D$ is $k$-wise independent if for every set $T \subseteq [n]$ with $|T| \leq k$ we have $D_T$ is the uniform distribution on $\sgrid^T$. Recall further that $D$ is $\varepsilon$-almost $k$-wise independent if for every set $T \subseteq [n]$ with $|T| \leq k$ we have $D_T$ is $\varepsilon$-close in total variation distance to the uniform distribution on $\sgrid^T$.\\

\noindent
In the following definition we view the rows of a stochastic matrix $M \in \R^{\sgrid^n_\mu \times \sgrid^n_\mu}$, denoted $M(\veca) := (M(\veca,\vecb)| \vecb \in \sgrid^n_\mu)$ for $\veca\in \sgrid^n_\mu$, as distributions supported on $\sgrid^n$ (which have zero support outside $\sgrid^n_\mu$). \\

\begin{definition}[{\bf $\varepsilon$-almost $k$-wise independent matrix}]\label{defn:eps-close-k-wise-matrix}
For parameter $k \in \N$ and $\epsilon > 0$ we say that a stochastic matrix $M \in \mathbb{R}^{\sgrid^n_\mu \times \sgrid^n_\mu}$ is \emph{$\varepsilon$-almost $k$-wise independent} if for every row $\mathbf{a} \in \sgrid^{n}_{\mu}$, the distribution $M(\mathbf{a})$ is  $\varepsilon$-almost $k$-wise independent.
\end{definition}

Finally we recall that for a matrix $M \in \R^{N \times N}$, its Frobenius norm, denoted $\|M\|_F$, is the quantity $\sqrt{\sum_{(i,j) \in N \times N} M(i,j)^2}$. 

We now state the main theorem of our work.\\

\begin{thmbox}
\begin{restatable}[{\bf Singular Value Bound for Markov Chains on Balanced Multislice}]{theorem}{generaleigenvalue}\label{thm:more-general-matrix-eigenvalue}
{
For every $\kappa > 0$, and $s\in\N$ with $\kappa \geq s$ there exists $c_1.c_2,c_3 < \infty$ such that for every $\varepsilon > 0$  and every sufficiently large $n \in \N$ that is divisible by $s$, the following holds:\newline 
Suppose $\sgrid$ is a set of size $s$, and $M \in \mathbb{R}^{\sgrid^{n}_\mu \times \sgrid^{n}_\mu}$ is a stochastic matrix that satisfies the following three conditions:\\

\begin{enumerate}
    \item The matrix $M$ respects symmetries.
    \item $\| M \|_{F} \; \leq \; c_1 \cdot n^{\kappa}$.
    \item The matrix $M$ is $\varepsilon$-almost $k$-wise independent for  $k = 10s \kappa$.\\
\end{enumerate}

\noindent
Then we have  $\sigma_{2}(M) \, \leq \, \max\set{c_2/n, c_3\cdot \varepsilon}$}.

\end{restatable}
\end{thmbox}

If the Markov chain is symmetric, then the singular values correspond to the eigenvalues, and hence \Cref{thm:more-general-matrix-eigenvalue} yields eigenvalue bounds for symmetric Markov chains satisfying the properties mentioned above. \Cref{thm:more-general-matrix-eigenvalue} is proved in \Cref{sec:eigenvalue}. The proof involves many standard and some new elements of representation theory for the symmetric group. We elaborate on this in \Cref{subsec:proof-overview-spectral}. We also note that \Cref{cor:gen-bal-cor} immediately follows from \Cref{thm:more-general-matrix-eigenvalue}, modulo some calculations that verify that Condition (2) above applies to $C$-balanced matrices. For more details, see \Cref{subsec:eigbound}.\\ 

\noindent
To illustrate the applicability of \Cref{cor:gen-bal-cor}, we give two examples, both related to coding theoretic aspects of polynomials and other polynomial-like functions that we refer to as \emph{junta-sums}. These results extend corresponding works in the Boolean setting \cite{ABPSS25, ABSS25-SZ-Lemma}, obtaining natural generalizations to non-Boolean settings.

\paragraph*{Distance of polynomials and junta-sums on multislices.}
A function $f:\sgrid^n \to G$ is called a $d$-junta if it depends on only $d$ of the $n$ variables, i.e, there exists a set $I \subseteq [n]$, $|I|\leq d$ and a function $g:\sgrid^I \to G$ such that for all $\veca\in\sgrid^n$, $f(\veca) = g(\veca|_I)$ where $\veca|_I$ is the projection of $\veca$ to the coordinates in $I$. Here we could allow $G$ to be any set, though in this work $G$ will denote an Abelian group. 
We say $f$ is a \emph{degree $d$ junta-sum} (or simply a \emph{$d$-junta-sum}) if there exists $d$-junta's $f_1,\ldots,f_k:\sgrid^n\to G$ such that $f = f_1 + \cdots + f_k$.\newline 
When $G=\F$ is a field and $\sgrid \subseteq \F$, then degree $d$ junta-sums are closely related to the notion of degree $d$ polynomials. In particular, every degree-$d$ polynomial is also a degree-$d$ junta-sum, and degree-$d$ junta-sums are polynomials of degree at most $(s-1)d$ where $s = |\sgrid|$. Junta-sums come up naturally when studying questions related to testing direct sums and low-degree polynomials~\cite{dinur2019direct,bogdanov2021direct,AmireddySS}.

A well-studied question about degree-$d$ polynomials is: How often can a non-zero polynomial be zero on a grid? The well-known and oft-discovered Ore-DeMillo-Lipton-Schwartz-Zippel lemma \cite{ore1922hohere,DL78,Zippel79,Schwartz80} (henceforth ODLSZ lemma) asserts that a non-zero degree-$d$ polynomial over a field $\F$ is non-zero with probability at least $s^{-d/(s-1)}$ over the uniform distribution over $\sgrid^n$. When $\mathcal{S} = \F$, the  precise bound is  $\delta(q,d) = (1-\beta/q)q^{-\alpha}$, where $\alpha$ and $\beta$ are the quotient and remainder respectively when $d$ is divided by $q-1$. The former bound immediately implies that a degree-$d$ junta-sum is non-zero with probability at least $s^{-d}$ over $\sgrid^n$ (and the claim even extends to arbitrary $\sgrid$ and Abelian groups $G$).\\ 

A natural related question then becomes --- \textit{how do these bounds change when considering natural subsets $T$ that are not grids (or more generally product sets)?} Recent work has begun to address such questions~\cite{ABSS25-SZ-Lemma, KoppartyKumarSha}. In this work, we consider the case of the balanced multislice i.e., $T = \sgrid_\mu^n$. Despite the simple nature of these questions, the answer does not seem to have been pinned down before, with the exception of the Boolean case that was resolved recently~\cite{ABSS25-SZ-Lemma}. 
 We are able to show a clean connection between $\sgrid^n$ and $\sgrid_\mu^n$ that allows us to show that these probabilities (in the worst case) differ by at most $\lambda_2(W)$ for some nearly balanced walk over the multislice. This allows us to prove the following theorem, which generalizes the work of~\cite{ABSS25-SZ-Lemma} beyond the Boolean case.\\
 
\begin{thmbox}
\begin{restatable}[{\bf Polynomial distance over multislice}]{theorem}{algebraicszlemma}\label{thm:dist-polys-intro}
    For every finite field $\F=\F_q$, if a degree $d$ polynomial $P(\x)$ is such that $P({\bf a}) \ne 0$ for some ${\bf a} \in \F^n_{\mu}$ on the balanced multislice, then $$\Pr_{{\bf b}\sim {\F^n_{\mu}}} [P({\bf b})\ne 0] \ge \delta(q,d) -\frac{1}{n^{\Omega_{q}(1)}},$$ where $\delta(q,d) = (1-\beta/q)q^{-\alpha}$, where $\alpha$ and $\beta$ are the quotient and remainder respectively when $d$ is divided by $q-1$.     
\end{restatable}
\end{thmbox}

We prove this theorem in \Cref{sec:sz}. 

Note that $\delta(q,d)$ is exactly the distance of the space of degree-$d$ polynomials on the field $\F_q$ and hence the above theorem says that the distance of the space of polynomials on the balanced multislice is nearly exactly what it is in the grid $\F_q^n$.\footnote{An important subtlety here is that there are polynomials that are non-zero in the grid $\F_q^n$ but are zero at all points on the multislice. That is the reason this theorem is only stated for polynomials that are non-zero as functions on the multislice. This is analogous to similar restrictions we place on polynomials in the setting of grids (e.g., in the setting of the Boolean cube, we only consider non-zero \emph{multilinear} polynomials).} An analogous statement can also be made for junta-sums, getting a bound that almost matches the bound over grids, i.e., $1/s^d$ (see~\Cref{thm:dist-junta-sums}).

Following the proof idea of~\cite{ABSS25-SZ-Lemma}, both cases are handled by a similar proof technique that first proves a quantitatively weak bound on the probability that $f$ is non-zero\footnote{We prove these bounds by an adaptation of the standard inductive strategy used to prove the standard ODLSZ lemma. Unfortunately, we are unable to use this strategy to prove a tight bound.}(see \Cref{coro:weak-distance-balanced-slice}), and then randomly identifies a small grid inside the multislice. On each such grid, we can apply the ODLSZ lemma as a black-box to assert that if $f$ is non-zero within the randomly identified small grid, then it is non-zero with the `correct' probability (either $\delta(q,d)$ or $s^{-d}$). It suffices, therefore, to prove that $f$ is non-zero on most of the grids, which is where the main technical theorem regarding the expansion of the walk on the balanced multislice comes into play. We use our eigenvalue bounds along with the quantitatively weak bound already obtained to establish that all but an $n^{-\Omega(1)}$-fraction of the grids satisfy this property. See \Cref{subsec:proof-overview-sz-lemma} for the proof overview and \Cref{sec:sz} for a formal proof.

\paragraph*{Local Correction of Junta-Sums.} 
Our next application considers the local correction problem for junta-sums over {\em grids}. Here a (possibly randomized) corrector is given oracle access to a function $f$ that is known to be $\delta$-close (in normalized Hamming distance) to some degree-$d$ junta-sum $P$, and also given a point $\veca\in \sgrid^n$ and needs to output $P(\veca)$ (with high probability) while making few oracle queries to $f$.\\
In the list-correction setting, the amount of error $\delta$ may be too high for $P$ to be defined uniquely by $f$ and $\delta$, but it may be known a priori that the list size is bounded. In the local list-correction problem, the goal for the corrector is to make a few queries to $f$ to produce several ``oracle'' algorithms, such that for every degree $d$ junta-sum $P$ that is $\delta$-close to $f$, there is an algorithm with oracle access to $f$ that computes $P$. We refer the reader to~\Cref{sec:prelims} for more formal definitions.

Local correction algorithms for low-degree polynomials have played a central role in complexity theory, for example \cite{GoldreichL, STV-list-decoding}. While most of the early works like \cite{gkz-list-decoding,BhowmickL} considered the setting where $\sgrid = G = \F$, some recent works have considered the setting of $\sgrid = \{0,1\}$ and general abelian $G$ such as \cite{ABPSS25} (Note that when $|\sgrid|=2$, then degree-$d$ polynomials are the same as degree-$d$ junta-sums.) 

For general $\sgrid$ and Abelian group $G$, even the list-decoding radius was not completely understood till this work. We prove that for $\delta = |\sgrid|^{-d}-\epsilon$ there are most $\bigO_\varepsilon(1)$ degree $d$ junta-sums $P$ that are $\delta$-close to any given function $f$. (This bound is tight in that for $\delta = |\sgrid|^{-d}$ the number of junta-sums grows with $n$.) This motivates the corresponding local list-correction problem, which we solve tightly in this work. We state an informal version below and point to~\Cref{thm:local-list-correction} for the more precise version.\\

\begin{thmbox}
\begin{theorem}[{\bf Local list-correction of junta-sums (Informal)}]\label{thm:local-list-correction-intro}
    For every set $\sgrid$, every Abelian group $G$, every integer $d$ and $\varepsilon > 0$, there exists an $L = L(\varepsilon,d,\sgrid)$ such that the following holds.\\
    
    There exists an algorithm $\mathcal{A}$ that on oracle access to a function $f:\sgrid^n\to G$, outputs $L$ oracle algorithms $\psi_1,\ldots,\psi_L$ such that for every degree $d$ junta-sum $P:\sgrid^n\to G$ that is $(1/s^d-\varepsilon)$-close to $f$, there exists $i \in [L]$ such that $\psi_i^{f}(\cdot)$ computes $P$ (with high probability for every input).\\
    The query complexity of $\mathcal{A}$ and $\psi_1,\ldots,\psi_L$ is $\mathrm{poly}(\log n)$.     
\end{theorem}
\end{thmbox}
This theorem is formalized as \Cref{thm:local-list-correction} with more explicit bounds on the error probability and query complexity, and is proved in \Cref{sec:llc}. 

This theorem generalizes a theorem of Amireddy, Behera, Paraashar, Srinivasan, and Sudan (\cite[Theorem 1.3.4]{ABPSS25}) who solved the corresponding problem over the Boolean cube $\Boo^{n}$. (Note that in the Boolean setting, junta-degree is the same as algebraic degree, and their result is thus expressed in terms of the latter phrase.) 
Our extension follows the same sequence of steps as employed in \cite{ABPSS25}. Their work ultimately ends up using the expansion properties of Boolean multislices, which, as we've noted earlier, is well-understood. Extending their work to general grids requires a number of changes that we elaborate on in  \Cref{subsec:proof-overview-llc}, with the most significant change being the use of \Cref{cor:gen-bal-cor} instead of the expansion results on the Boolean slice.

\subsection{Techniques and Proof Overview}

In this section, we first review known methods for bounding the singular values of walks that respect symmetries and explain where there is a gap in knowledge. We then show how we overcome these challenges by overviewing the proof of our main theorem \Cref{thm:more-general-matrix-eigenvalue} in \Cref{subsec:proof-overview-spectral}. Next, we give an overview of the proof of the ODLSZ theorem for multislices, \Cref{thm:dist-polys-intro}, in \Cref{subsec:proof-overview-sz-lemma}. Finally, we discuss the proof of the local correction theorem for grids, \Cref{thm:local-list-correction-intro} in \Cref{subsec:proof-overview-llc}. 

\subsubsection{Prior Approaches and Obstructions}

We describe some prior cases where random walk matrices respecting symmetries (i.e., the first condition of \Cref{thm:more-general-matrix-eigenvalue}) have been studied and explain the special properties in play there. 

\paragraph*{Boolean Hypercube and Cayley graphs}A broad class of examples bounding eigenvalues of highly symmetric graphs are the bounds on the eigenvalues of Cayley graphs over abelian groups - this captures random walks on the Boolean hypercube and many more general settings. Here it is well known that the random walk matrix is {\em diagonalizable}\footnote{A matrix $M \in \R^{N \times N}$ is said to be diagonalizable if there is a unitary matrix $U \in \R^{N \times N}$ such that $UMU^T$ is a diagonal matrix.} and the eigenvectors of the random walk matrix depend only on the group (and not the set of generators). This makes it possible to determine the entire eigenspectrum for many basic groups using Fourier analysis. We note that Cayley graphs over some non-abelian groups have been studied, but general results are mostly lacking. In these cases, the random walk matrix is typically not diagonalizable, but can be made block diagonal, using the representation theory of the underlying group. This is a complex tool, and many basic questions are unanswered as we elaborate below.

\paragraph*{Boolean slices}One well-studied setting that happens to be the special case of $s = 2$ of our problem is the setting of Boolean slices. Here $\mathcal{S} = \Boo$ and $\sgrid^n_\mu$ is the balanced Boolean slice (all points in $\Boo^{n}$ of Hamming weight exactly $n/2$). 
This setting has particular relevance to the analysis of Boolean functions and combinatorics;  see, e.g. \cite{Delsarte, Filmus2014AnOB, Filmus-JuntaSlices}. The random walk matrices in this setting lie in the \emph{Johnson scheme}, which is an algebra of symmetric matrices that commute with one another. This implies that all such matrices can be diagonalized \emph{simultaneously}, i.e., there exists one unitary matrix $U$ such that for every random walk matrix $M$ on the Boolean slice that respects symmetries, we have that $UMU^T$ is diagonal. This implies that all such matrices $M$ have the same eigenvectors. The works \cite{Filmus2014AnOB,MuraliSrinivasan-SymmetricChain} gave explicit descriptions of the common eigenspaces. This can be quite useful when analyzing the spectrum of such matrices and in a recent example \cite{ABSS25-SZ-Lemma} used this description to show that a particular random walk matrix on the balanced Boolean slice is a good spectral expander (see \cite[Lemma 3.2]{ABSS25-SZ-Lemma}).

\subsubsection{Spectral Expansion of Multislice Walks}\label{subsec:proof-overview-spectral}

Turning to our setting, our matrix $M$ is not {\em diagonalizable} and so the techniques from the analysis of Cayley graphs on abelian groups as well as the random walk on the Boolean slice, do not work in this setting. We have to resort to the use of representation theory, but here as we alluded to earlier, our understanding is not as complete. In what follows, we explain what representation theory implies for our setting and how we build on it. 

\paragraph*{Summary of known facts from representation theory}The fact that our matrix $M$ respects symmetries allows us to invoke results from the representation theory of the symmetric group $\Sym_n$. We cover these results in detail in \Cref{prop:standard-rep-theory} and \Cref{thm:instantiation-rep-thy-sym} (Parts (1) and (2)). Essentially, we can use representation theory to show that our matrix $M$ can be block-diagonalized with relatively few blocks.

Specifically\footnote{This conclusion only requires that $M$ respects symmetries (see \Cref{thm:more-general-matrix-eigenvalue}).} there is an orthonormal matrix $U = U_{\sgrid,n}$ such that for every $M$ respecting symmetries the matrix $UMU^{T}$ is block diagonal with blocks $M_0,M_1,\ldots,M_t$ where $t$ and the ``shape'' of the blocks is known from  standard representation theory. In particular, $M_0 = [1]$ is just a $1\times 1$ matrix that contributes the top singular value (which is $1$), and $M_1,\ldots,M_t$ determine $\sigma_{2}(M) < 1$.\newline

\noindent
Now let us understand the structure of $M_{i}$'s in more detail. Each block $M_i$ is a Kronecker/tensor product of a ``small'' matrix $A_i \in \R^{m(i) \times m(i)}$ with a somewhat larger identity matrix i.e.,
\begin{align*}
    M_i = A_i \otimes \mathrm{Id}_{k(i)}, && \text{where } \; \mathrm{Id}_{k} \text{ is the $k \times k$ identity matrix}.
\end{align*}
See \Cref{fig:block-matrix-M} for an informal pictorial description. Both the quantities $m(i)$ and $k(i)$ are known from the representation theory of $\mathrm{Sym}_{n}$ (and in particular only depend on $\sgrid$ and $n$ and are independent of the particular matrix $M$). However, the small matrices $A_i$'s do depend on $M$, and more importantly, the matrix $U$ is not too well-understood. (In particular, we need to understand the effect of $U$ on $A_i$ and this is not clear.) In particular, if we were to arrange $i$ such that $m(i)$'s are non-decreasing, then $k(i)$'s are also non-decreasing. Intuition from Fourier analysis in the Abelian world would suggest that $\lambda_2$ comes from $M_1 = A_1 \otimes \mathrm{Id}_{k(1)}$ but, as far as we are aware, even this is {\em not known}.

\begin{figure}[!h]
    \centering

\tikzset{every picture/.style={line width=0.75pt}} 

\begin{tikzpicture}[x=0.75pt,y=0.75pt,yscale=-1,xscale=1]

\draw    (158.4,29.73) -- (158.4,219.47) ;
\draw    (397,29.08) -- (397,218.08) ;
\draw  [fill={rgb, 255:red, 155; green, 155; blue, 155 }  ,fill opacity=0.25 ] (177,48.33) -- (216.67,48.33) -- (216.67,88) -- (177,88) -- cycle ;
\draw  [fill={rgb, 255:red, 155; green, 155; blue, 155 }  ,fill opacity=0.25 ] (257.13,109.27) -- (296.8,109.27) -- (296.8,148.93) -- (257.13,148.93) -- cycle ;
\draw  [fill={rgb, 255:red, 155; green, 155; blue, 155 }  ,fill opacity=0.25 ] (296.8,148.93) -- (325.8,148.93) -- (325.8,177.93) -- (296.8,177.93) -- cycle ;
\draw  [fill={rgb, 255:red, 155; green, 155; blue, 155 }  ,fill opacity=0.25 ] (367.75,188.83) -- (397,188.83) -- (397,218.08) -- (367.75,218.08) -- cycle ;
\draw  [fill={rgb, 255:red, 155; green, 155; blue, 155 }  ,fill opacity=0.25 ] (158.4,29.73) -- (177,29.73) -- (177,48.33) -- (158.4,48.33) -- cycle ;
\draw  [color={rgb, 255:red, 155; green, 155; blue, 155 }  ,draw opacity=1 ] (135.98,51.89) .. controls (131.31,51.89) and (128.98,54.22) .. (128.98,58.89) -- (128.98,90.91) .. controls (128.98,97.58) and (126.65,100.91) .. (121.98,100.91) .. controls (126.65,100.91) and (128.98,104.24) .. (128.98,110.91)(128.98,107.91) -- (128.98,142.93) .. controls (128.98,147.6) and (131.31,149.93) .. (135.98,149.93) ;
\draw  [color={rgb, 255:red, 155; green, 155; blue, 155 }  ,draw opacity=1 ] (135.53,151.33) .. controls (130.86,151.33) and (128.53,153.66) .. (128.53,158.33) -- (128.53,175.63) .. controls (128.53,182.3) and (126.2,185.63) .. (121.53,185.63) .. controls (126.2,185.63) and (128.53,188.96) .. (128.53,195.63)(128.53,192.63) -- (128.53,212.93) .. controls (128.53,217.6) and (130.86,219.93) .. (135.53,219.93) ;
\draw [color={rgb, 255:red, 128; green, 128; blue, 128 }  ,draw opacity=1 ] [dash pattern={on 0.84pt off 2.51pt}]  (298.17,127.5) .. controls (337.77,97.8) and (412.98,137.61) .. (453.29,109.3) ;
\draw [shift={(454.5,108.42)}, rotate = 143.13] [color={rgb, 255:red, 128; green, 128; blue, 128 }  ,draw opacity=1 ][line width=0.75]    (10.93,-3.29) .. controls (6.95,-1.4) and (3.31,-0.3) .. (0,0) .. controls (3.31,0.3) and (6.95,1.4) .. (10.93,3.29)   ;
\draw  [color={rgb, 255:red, 144; green, 19; blue, 254 }  ,draw opacity=1 ][fill={rgb, 255:red, 218; green, 188; blue, 246 }  ,fill opacity=0.25 ] (462.13,68.92) -- (523.44,68.92) -- (523.44,128.17) -- (462.13,128.17) -- cycle ;
\draw  [color={rgb, 255:red, 144; green, 19; blue, 254 }  ,draw opacity=1 ][fill={rgb, 255:red, 218; green, 188; blue, 246 }  ,fill opacity=0.25 ] (565.23,87.08) -- (590.33,87.08) -- (590.33,108.58) -- (565.23,108.58) -- cycle ;
\draw  [color={rgb, 255:red, 144; green, 19; blue, 254 }  ,draw opacity=1 ][fill={rgb, 255:red, 218; green, 188; blue, 246 }  ,fill opacity=0.25 ] (632.4,76.93) -- (672.07,76.93) -- (672.07,116.6) -- (632.4,116.6) -- cycle ;
\draw  [color={rgb, 255:red, 144; green, 19; blue, 254 }  ,draw opacity=1 ] (565.67,117.5) .. controls (565.67,121.04) and (567.44,122.81) .. (570.97,122.81) -- (570.97,122.81) .. controls (576.03,122.81) and (578.56,124.58) .. (578.56,128.11) .. controls (578.56,124.58) and (581.09,122.81) .. (586.14,122.81)(583.86,122.81) -- (586.14,122.81) .. controls (589.67,122.81) and (591.44,121.04) .. (591.44,117.5) ;
\draw  [color={rgb, 255:red, 144; green, 19; blue, 254 }  ,draw opacity=1 ] (633.33,127.83) .. controls (633.33,132.5) and (635.66,134.83) .. (640.33,134.83) -- (643.22,134.83) .. controls (649.89,134.83) and (653.22,137.16) .. (653.22,141.83) .. controls (653.22,137.16) and (656.55,134.83) .. (663.22,134.83)(660.22,134.83) -- (666.11,134.83) .. controls (670.78,134.83) and (673.11,132.5) .. (673.11,127.83) ;
\draw    (158.4,219.47) -- (184.13,219.47) ;
\draw    (373.73,29.08) -- (397,29.08) ;
\draw   (159,243) .. controls (159,247.67) and (161.33,250) .. (166,250) -- (266.83,250) .. controls (273.5,250) and (276.83,252.33) .. (276.83,257) .. controls (276.83,252.33) and (280.16,250) .. (286.83,250)(283.83,250) -- (387.67,250) .. controls (392.34,250) and (394.67,247.67) .. (394.67,243) ;

\draw (162.53,34.27) node [anchor=north west][inner sep=0.75pt]  [font=\scriptsize]  {$1$};
\draw (188.8,61.87) node [anchor=north west][inner sep=0.75pt]  [font=\scriptsize]  {$M_{1}$};
\draw (374.8,198.53) node [anchor=north west][inner sep=0.75pt]  [font=\scriptsize]  {$M_{t}$};
\draw (270.4,122.13) node [anchor=north west][inner sep=0.75pt]  [font=\scriptsize]  {$M_{i}$};
\draw (33.13,94.43) node [anchor=north west][inner sep=0.75pt]  [font=\scriptsize,color={rgb, 255:red, 74; green, 144; blue, 226 }  ,opacity=1 ]  {$k( i) \ \gg \ \| M\| _{F} \ $};
\draw (32.13,177.93) node [anchor=north west][inner sep=0.75pt]  [font=\scriptsize,color={rgb, 255:red, 74; green, 144; blue, 226 }  ,opacity=1 ]  {$ \begin{array}{l}
k( i) \ \approx \ \| M\| _{F} \ \\
m( i) \ =\ \mathcal{O}( 1)
\end{array}$};
\draw (485.69,89.87) node [anchor=north west][inner sep=0.75pt]  [font=\scriptsize,color={rgb, 255:red, 144; green, 19; blue, 254 }  ,opacity=1 ]  {$M_{i}$};
\draw (571.47,92.96) node [anchor=north west][inner sep=0.75pt]  [font=\scriptsize,color={rgb, 255:red, 144; green, 19; blue, 254 }  ,opacity=1 ]  {$A_{i}$};
\draw (645.67,89.8) node [anchor=north west][inner sep=0.75pt]  [font=\scriptsize,color={rgb, 255:red, 144; green, 19; blue, 254 }  ,opacity=1 ]  {$\mathrm{Id}$};
\draw (568.78,136.13) node [anchor=north west][inner sep=0.75pt]  [font=\scriptsize,color={rgb, 255:red, 144; green, 19; blue, 254 }  ,opacity=1 ]  {$m( i)$};
\draw (644.78,150.13) node [anchor=north west][inner sep=0.75pt]  [font=\scriptsize,color={rgb, 255:red, 144; green, 19; blue, 254 }  ,opacity=1 ]  {$k( i)$};
\draw (605.44,92.47) node [anchor=north west][inner sep=0.75pt]  [font=\footnotesize,color={rgb, 255:red, 144; green, 19; blue, 254 }  ,opacity=1 ]  {$\otimes $};
\draw (533.78,90.47) node [anchor=north west][inner sep=0.75pt]  [font=\footnotesize]  {$=$};
\draw (245,269.4) node [anchor=north west][inner sep=0.75pt]    {$U\ M\ U^{T}$};
\draw (224.67,92.9) node [anchor=north west][inner sep=0.75pt]    {$\ddots $};
\draw (337.67,179.9) node [anchor=north west][inner sep=0.75pt]  [font=\footnotesize]  {$\ddots $};

\end{tikzpicture}
    
    \caption{An informal visualization for the block diagonalization of our matrix $M$}
    \label{fig:block-matrix-M}
\end{figure}
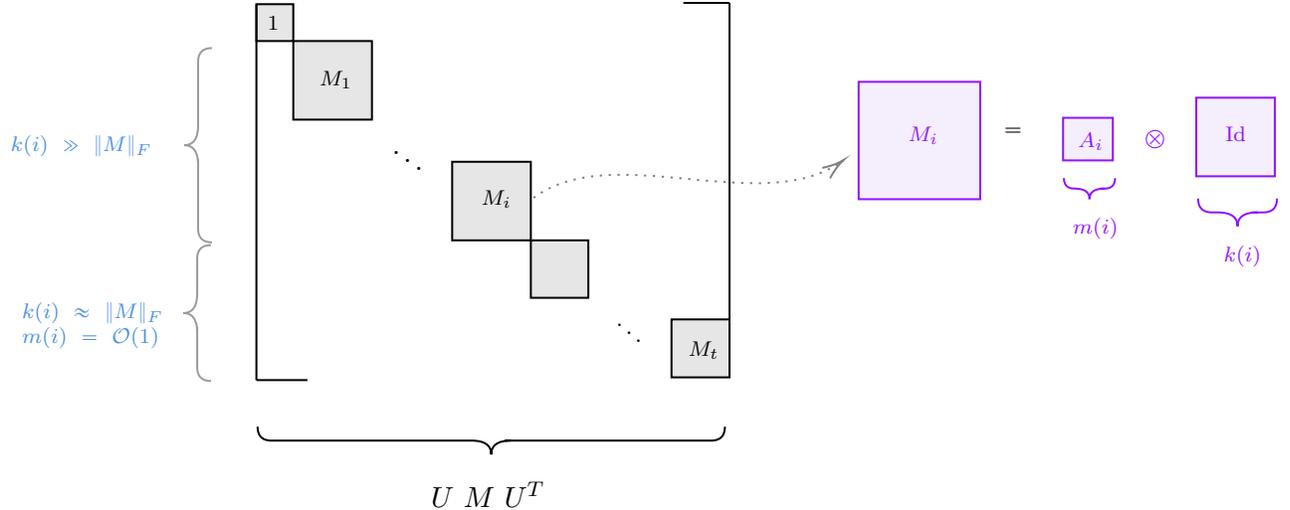

\paragraph*{Our analysis}
Given that the $A_i$'s are not determined by only $\sgrid$ and $n$, and $U$ is not explicitly understood, we need to find some crude ways to bound the singular values of $M$. We give such an analysis in \Cref{sec:eigenvalue} and summarize the essence here. We start with the following informal observation.\\

\begin{observation}\label{obs:informal-proof-overview}
If for some $i \in [t]$, the quantity $k(i)$ is a polynomial larger than the Frobenius norm of $M$, then every singular value corresponding to the block $M_{i}$ must necessarily be small. \\
\end{observation}

The observation follows from the fact that the Frobenius norm is lower bounded by the sum of the squares of the singular values of $M_i$, each of which repeats at least $k(i)$ times, so each singular value must be small.

In the context of our goal, the observation above immediately allows us to eliminate all the large blocks in the block diagonalization of $M$ and turns the focus to the small blocks --- where $k(i)$ is bounded by some polynomial in $n$. Here, standard facts of representations of $\mathrm{Sym}_{n}$ imply that the corresponding matrix $A_i$ is of constant size (dependent on $\sgrid$ and the exponent of $n$ in $k(i)$ but not on $n$). However, this does not immediately translate to bounds on the singular values, since these depend on the actual matrix $A_i$ and its entries. Ideally, we would have liked to get our hands on the singular vectors of $M$ corresponding to singular vectors of the $A_i$'s (after a change of basis according to $U$), but such vectors were not known (and we do not get such vectors either). 
Fortunately, a collection of vectors that span the space corresponding to the singular vectors of $A_i$'s is known. In particular, we use a description given by Dafni, Filmus, Lifshitz, Lindzey, and Vinyals \cite{complexity-symmetric-group} --- see \Cref{defn:chi-T}. Our main contribution adds two observations about this collection of vectors, called ``special vectors'' below.

Our main contribution here does get something almost as good, for our purposes (i.e., to show that each $M_i$  has top singular value  $o_n(1)$):

\begin{quote}
\noindent \textit{We observe that the special vectors given in \Cref{defn:chi-T} are ``weakly orthogonal'' in the sense that they have $\Omega(1)$ volume (in the space they span). We further observe that these functions are junta-like and so shrink significantly when acted on by $\varepsilon$-almost $k$-wise independent matrices for sufficiently large constant $k$. (See~\Cref{thm:instantiation-rep-thy-sym}, Parts 3(c) and 3(d)).}
\end{quote}
More specifically, the work of \cite{complexity-symmetric-group} yields $\dim(A_i) = m(i)$ many vectors (see~\Cref{subsec:alt-special-vectors}) of $M$ that are supported on coordinates of one of the blocks of $M_i$ after transforming the basis according to $U$. We show that while these vectors do not form an orthogonal basis, they are sufficiently divergent to ensure their determinantal volume is large (see~\Cref{lem:special-vecs}). Thus, bounding the length of the vectors obtained by applying the linear map $A_i$ by $o_n(1)$ suffices to bound the spectral norm of $A_i$ and hence also $M_i$. Details of this part may be found in \Cref{subsec:alt-special-vectors}. We give some more insight in the next paragraph.\\

\noindent
To show our singular value bounds, we use the fact that the vectors in the basis correspond to $\bigO(1)$-junta's. Specifically, note that a vector that $M$ acts on can be viewed as a function $f$ from $\sgrid^n_\mu$ to $\mathbb{R}$, which can in turn be viewed as a partial function from $\sgrid^n$ to $\mathbb{R}$. We show that this function depends on only $\bigO(1)$-coordinates of the input vector. (See \Cref{subsec:alt-special-vectors}.) This property now combines nicely with our third condition in \Cref{thm:more-general-matrix-eigenvalue} which asserts that after multiplication by $M$ any vector $f$ looks essentially random when projected on $\bigO(1)$-coordinates and so has little correlation left with $f$ --- this immediately translates into an upper bound on the singular value of $M$ corresponding to coordinates in $M_i$, and yields a proof of \Cref{thm:more-general-matrix-eigenvalue}. 

\paragraph*{Why does our \Cref{cor:gen-bal-cor} hold only for nearly-balanced walks?}The reason is related to \Cref{obs:informal-proof-overview}. We note that the Frobenius norm condition is not completely natural, and indeed the natural matrices in our applications do not satisfy this condition (and we have to find workarounds). The Frobenius norm restriction is satisfied by nearly-balanced walks as considered in \Cref{cor:gen-bal-cor}, and indeed it is one of the reasons why \Cref{cor:gen-bal-cor} is restricted to such nearly balanced walks.

\subsubsection{Distance Lemma over Balanced Multislice}\label{subsec:proof-overview-sz-lemma}
In this subsection, we discuss the proof overview for \Cref{thm:dist-polys-intro}. The strategy is a generalization of the proof for \cite[Lemma 3.2]{ABSS25-SZ-Lemma}. The idea is to find a random copy of $\sgrid^{n/s}$ inside the balanced multislice $\sgrid^{n}_{\mu}$ such that it is a good \emph{sampler} for $\sgrid^{n}_{\mu}$, i.e. if we choose points from this subgrid $\sgrid^{n/s}$ at random, then the corresponding points in $\sgrid^{n}_{\mu}$ behave like random samples. As we explain now, this guarantee essentially allows us to move from balanced multislice to subgrid, where we have a complete understanding of distance. For every non-zero $d$-junta-sum $P: \sgrid^{n}_{\mu} \to G$, we choose a random copy of $\sgrid^{n/s}$ inside $\sgrid^{n}_{\mu}$ and restrict $P$ to this copy. With the sampling guarantee, we can argue that the restricted $d$-junta-sum is also non-zero on the subgrid $\sgrid^{n/s}$, and we get the claimed bound by applying \Cref{clm:dist-junta-sums} on this restricted polynomial. Next, we explain the process of finding a random copy of the $n/s$-dimensional subgrid inside the balanced multislice.\\

The key step in our proof is to show that we can find a random copy of $\sgrid^{n/s}$ inside $\sgrid^{n}_{\mu}$, which is a sampler for $\sgrid^{n}_{\mu}$. We do it by randomly grouping the coordinates $x_{1},\ldots,x_{n}$ into $n/s$ buckets of size $s$ and in each bucket, we randomly assign distinct values to the $s$ coordinates. We prove that for two random points in $\sgrid^{n/s}$, their corresponding points in the balanced multislice $\sgrid^{n}_{\mu}$ are almost pairwise independent. We show this via the second moment method and the expander mixing lemma. We use our main theorem \Cref{thm:more-general-matrix-eigenvalue} to show that the random walk on $\sgrid^{n}_{\mu}$ arising from the above random process has good spectral expansion, making the expander mixing lemma applicable in this context.

\subsubsection{Local List Correction for Junta-Sums}\label{subsec:proof-overview-llc}
Our local list corrector (see \Cref{thm:local-list-correction}) is a generalization of \cite[Theorem 1.3.4]{ABPSS25} to $d$-junta-sums and arbitrary grids $\sgrid^{n}$ (instead of degree-$d$ \emph{polynomials} and \emph{Boolean} cube). We do not dwell on the algorithm here, but only highlight and discuss the key technical difference in our work and the previous work of \cite{ABPSS25}. We request the reader to please refer to \Cref{sec:llc} for more details on the algorithm.\\

An important step of our local list corrector involves a random restriction of $\sgrid^{sk}$ to a subgrid $\sgrid^{k}$, as follows: We randomly group the $sk$ coordinates into $k$ groups of size $s$, and identify all the $s$ coordinates in a group together by a single new coordinate. To show that our local list corrector has small error probability, we need the following guarantee from the above random restriction: If a $d$-junta-sum $P \in \mathcal{J}_{d}(\sgrid^{sk})$ is non-zero on the balanced multislice, then with high probability, it continues to be a non-zero junta-sum on $\sgrid^{k}$ after the random restriction. For this, we show that the above-mentioned random process can be interpreted as finding a random copy of the balanced multislice in $\sgrid^{k}$ inside the balanced multislice in $\sgrid^{sk}$. Similar to the distance lemma for multislices (\Cref{thm:dist-junta-sums}), we show that we get a good sampler using \Cref{thm:more-general-matrix-eigenvalue}.

We now briefly touch upon some of the additional challenges in going from the Boolean case of~\cite{ABPSS25} to junta-sums over grids $\sgrid^n$, in the context of local list-correction. For the local correction algorithm in the unique decoding regime, the main idea is to reduce the problem to the Boolean case but over a biased distribution instead of the uniform one; the proof then proceeds by a mostly straightforward generalization of the local corrector from~\cite{ABPSS25} for the uniform distribution. The overall template for proving the combinatorial bound is also similar to that over the Boolean cube, except now we will need more general anti-concentration lemmas and distance lemmas for junta-sums. As already described in the above paragraph, going from the combinatorial bound for list-decodability to the local list-corrector is the main technical challenge we overcome in this work by making use of the fact that a certain random embedding of the multislice of $\sgrid^k$ inside the multislice of $\sgrid^n$ is a good sampler.

\subsection{Organization}
In \Cref{sec:prelims}, we give some definitions that we are going to use throughout the article. In \Cref{sec:eigenvalue}, we prove the main theorem of our work (\Cref{thm:more-general-matrix-eigenvalue}), which itself is organized as follows: we start with giving some necessary background on representation theory for the symmetric group, then use it to prove \Cref{thm:more-general-matrix-eigenvalue}, and finally show that ``typical'' random-walk matrices are good spectral expanders. In the subsequent sections, we give applications of our main theorem. In \Cref{sec:sz}, we prove a near-optimal distance lemma for junta-sums and polynomials over balanced multislice (see \Cref{thm:dist-junta-sums} and \Cref{thm:dist-multislice}). In \Cref{sec:llc}, we give a local list corrector for $d$-junta-sums over $\sgrid^{n}$ (see \Cref{thm:local-list-correction}).

\section{Preliminaries}
\label{sec:prelims}

We begin by describing some standard notation and terminology we will use throughout the paper.

For a set of parameters $\alpha_1,\ldots, \alpha_t$, the notation $\bigO_{\alpha_1,\ldots, \alpha_t}(\cdot)$ hides factors depending on $\alpha_1,\ldots, \alpha_t$. Similarly for $\Theta_{\alpha_1,\ldots, \alpha_t}(\cdot), \Omega_{\alpha_1,\ldots, \alpha_t}(\cdot)$ and so on. Although this is not standard, we will use the notation $\widetilde{\bigO}(\cdot)$ to hide $(\log \log n)^{\bigO(1)}$ factors (generally this notation is used to hide $(\log n)^{\bigO(1)}$) factors). 
 We use $|\x|$ to denote the Hamming weight of $\x$, i.e., the number of non-zero coordinates. Let $\bern(p)^n$ denote the distribution over $\{0,1\}^n$ where each bit is chosen from the Bernoulli distribution $\bern(p)$ independently. For two distributions $X,Y$ over the same finite domain, we let $\text{SD}(X,Y)$ denote the statistical distance between the distributions. We let $\|{\bf v}\|_2$ denote the $\ell_2$ norm of a vector ${\bf v}\in \R^N$.\newline
 
 For any $s \in \mathbb{N}$, we use $\mathbb{Z}_{s}$ to denote the cyclic group $\mathbb{Z}/s\mathbb{Z}$, and not to be confused by the $p$-adic field $\mathbb{Z}_{s}$. We say that a group is a {\em torsion group} if all its elements have finite order. The {\em exponent} of a torsion group is the least common multiple of the orders of all its elements.

Let $n$ and $s$ be two natural numbers where $n$ is divisible by $s$ and let $\mathcal{S}^n_\mu$ denote the balanced multislice over a finite set $\mathcal{S}$ of size $s$, i.e., \begin{align*}
     \sgrid^n_\mu := \setcond{\veca \in \sgrid^n }{\forall \sigma \in \sgrid, |\{i\in[n] | a_i = \sigma\}| = \frac{n}{s}}.
\end{align*}
Similarly, for any $\lambda = (\lambda_{0},\ldots,\lambda_{s-1})$ with $\lambda_{0}+\ldots+\lambda_{s-1} = n$, we define the multislice $\sgrid^{n}_{\lambda}$ as follows:
\begin{align*}
    \sgrid^n_\lambda := \setcond{\veca \in \sgrid^n }{\forall \sigma \in \sgrid, |\{i\in[n] | a_i = \sigma\}| = \lambda_{i}}.
\end{align*}

Then, we define the generalized Hamming distance between points in the (balanced) multislice as follows:\\

\begin{definition}[{\bf Generalized Hamming distance}]\label{defn:profile}
We define the generalized Hamming distance $\Delta({\bf a},{\bf b})$ between two points ${\bf a},{\bf b}\in \mathcal{S}^n_\mu$ to be the $\mathcal{S} \times \mathcal{S}$ matrix where the $(\sigma,\tau)$-th entry is given by $|\{i\in [n]:a_i=\sigma\text{~and~}b_i=\tau\}|$. 

\end{definition}

\noindent
\begin{example}[A generalized Hamming distance matrix for $n = 9$ and $s = 3$.]
Let $\mathbf{u} = 000111222$ and $\mathbf{v} = 110201022$. Then,
\begin{align*}
    \Delta({\bf u},{\bf v}) \; = \; \begin{bmatrix}
        1 & 2 & 0 \\
        1 & 1 & 1 \\
        1 & 0 & 2
    \end{bmatrix}.
\end{align*}
\end{example}

We now define the notion of nearly balanced generalized Hamming distance (or {\em $C$-balanced} profiles to be more precise).\\

\begin{definition}[{\bf $C$-balanced generalized Hamming distance}]\label{defn:bal-profiles}
    For $C\ge 0$, we say that a generalized Hamming distance parameter $P\in \Z^{\mathcal{S} \times \mathcal{S}}$ w.r.t~ a multislice $\mathcal{S}^n_\mu$ is {\em $C$-balanced} if every entry of $P$ is in the range $\frac{m}{s}\pm \sqrt{C m\log m}$ where $m=n/s$. 

\end{definition}

We now use generalized Hamming distance matrices to define a random walk on the balanced multislice $\sgrid^{n}_{\mu}$.\\

\begin{definition}[{\bf Random walk determined by a generalized Hamming distance matrix}]\label{defn:w-delta}
    We say that a generalized Hamming distance matrix $P \in \Z^{\mathcal{S} \times \mathcal{S}}$ {\em determines} a random walk matrix, denoted $W_P$, if for each vertex ${\bf a} \in \mathcal{S}^n_\mu$ in the multislice, the random step corresponding to $W_P$ is obtained by picking, uniformly at random, a vertex ${\bf b}\in \mathcal{S}^n_\mu$ such that $\Delta({\bf a},{\bf b})=P$. That is, the ${\bf a}$-th row of $W_P$ (denoted $W_P({\bf a})$) is the uniform distribution over $\{{\bf b}\in \mathcal{S}^n_\mu :  \Delta({\bf a},{\bf b}) = P \}$.
\end{definition}

We now give some necessary background for random walk matrices more generally. We refer the reader to the survey~\cite{vadhan-pseudorandomness} for more discussion.\\

\noindent
We say that a matrix $W\in \R^{N \times N}$ is a {\em random walk matrix} if for every $i\in [N]$ (which may be referred to as {\em vertices}), the $i$-th row of the matrix, denoted $W(i)$, is a probability distribution over $[N]$. It is clear that every random walk matrix has an eigenvector of ${\bf 1}$ with eigenvalue 1. If $W$ is symmetric, then it has real eigenvalues $1=\lambda_1 \ge \dots \ge \lambda_N \ge -1$; and we define  $\lambda_2(W)=\max\{|\lambda_2|,|\lambda_n|\}$ to be the second largest eigenvalue of $W$ in absolute value. Equivalently, one can show that
\begin{align*}
    \lambda_2(W) = \max_{{\bf v}\in \R^N: {\bf v}^\top {\bf 1}=0} \frac{\|W{\bf v}\|_2}{\|{\bf v}\|_2}.
\end{align*}

We will also deal with random walk matrices that are not necessarily symmetric. We say a square matrix is {\em stochastic} if all its entries are non-negative and each row element sums to 1. We say that a matrix is {\em doubly stochastic} if both the matrix and its transpose are stochastic. We observe that doubly stochastic matrices have ${\bf 1}$ as both a left eigenvector and right eigenvector. Furthermore, it has singular values $1=\sigma_1 \ge \sigma_2 \ge \dots \sigma_N \ge 0$, where $N$ is the order of the matrix; and we use $\sigma_2(W)$ to mean $\sigma_2$. Similar to the case of symmetric matrices, we have for every doubly stochastic matrix $W\in \R^{N\times N}$:
\begin{align*}
    \sigma_2(W) = \max_{{\bf v}\in \R^N: {\bf v}^\top {\bf 1}=0} \frac{\|W{\bf v}\|_2}{\|{\bf v}\|_2}.
\end{align*}

For symmetric matrices, singular values are simply the absolute values of the eigenvalues. Hence, if $W$ is a symmetric random walk matrix with eigenvalues $\lambda_1\ge \dots\ge \lambda_N$ and singular values $\sigma_1\ge \dots \ge \sigma_N$, then $\lambda_1 = \sigma_1 = 1$ and $\lambda_2(W) = \sigma_2(W)$.

We observe the following property of the random walks determined by generalized Hamming distance between points on the multislice (\Cref{defn:w-delta}).\\

\begin{observation}\label{obs:w-delta-props}
    For every generalized Hamming distance matrix $P\in \Z^{s\times s}$ defined with respect to a multislice $\mathcal{S}^n_\mu$, we have that $W_P^\top = W_{P^\top}$ is a random walk matrix. In particular, $W_P$ is doubly stochastic.
\end{observation}

We will now show how to bound the eigenvalues of a convex combination of random walk matrices.\\

\begin{lemma}[{\bf Singular value bound for convex combinations}]\label{lem:convex}
    Suppose $W=\sum_{i\in [t]} \alpha_i W_i$, where $W_i \in \R^{N \times N}$ are doubly stochastic matrices and $\alpha_1\dots, \alpha_t\in [0,1]$ are such that $\sum_{i\in [t]} \alpha_i = 1$. Let $S\subseteq [t]$ be arbitrary. Then $W$ is also a doubly stochastic matrix with $$\sigma_2(W) \le \max_{i\in S}\{\sigma_2(W_i)\} + \sum_{i\notin S} \alpha_i.$$
\end{lemma}

\begin{proof}
    We observe that each row (similarly column) of $W$ is a convex combination of probability distributions, so is also a probability distribution; hence $W$ is indeed doubly stochastic. In other words, ${\bf 1}$ is both a left eigenvector and right eigenvector. Hence, we have that $$\sigma_2(W) = \max_{{\bf u}\in \R^N:{\bf u}^\top {\bf 1}=0\text{~and~}\|{\bf u}\|_2=1} \|W{\bf u}\|_2.$$ Now letting ${\bf u}\in \R^{N}$ be an arbitrary vector such that $\|{\bf u}\|_2  =1$ and ${\bf u}^\top {\bf 1} = 0$, we will bound $\|W{\bf u}\|_2$. We have
    \begin{align*}
        \|W{\bf u}\|_2  = \bigg \|\sum_{i\in [t]}\alpha_{i} W_{i}{\bf u}\bigg\|_2
         & \le \sum_{i\in [t]} \alpha_{i} \|W_{i}{\bf u}\|_2 
          = \sum_{i\in S} \alpha_{i} \|W_{i}{\bf u}\|_2 + \sum_{i\notin S} \alpha_{i} \|W_{i}{\bf u}\|_2\\
         & \le \paren{\sum_{i\in S} \alpha_{i}} \paren{ \max_{i\in S}\|W_{i}{\bf u}\|_2 }+ \paren{\sum_{i\notin S} \alpha_{i}} \paren{1} 
         \le \max_{i\in S}\{\sigma_2(W_i)\} + \sum_{i\notin S} \alpha_i, 
    \end{align*} where we are using the triangle inequality for the first inequality, and that each $W_i$ is a random walk matrix for the second inequality. 
\end{proof}

We now move on to the definitions needed for our local list-correction application in~\Cref{sec:llc}.

\subsection*{Local Correction and Junta-Sums} 

We say that a family of functions $\mathcal{F}$ from a finite domain $D$ to a (finite or infinite) co-domain $G$, is {\em $(q,\varepsilon)$-locally correctable} if there exists a $q$-query algorithm $\mathcal{A}$, which when given query access to a function $f:D\to G$ such that $\delta(f,P) \le \varepsilon$ for some $P \in \mathcal{F}$, and an input index $i\in D$, outputs $P(i)$ with probability at least $3/4$. In words, the algorithm $\mathcal{A}$ is able to ``correct'' any given index of the received word $f$ by making only a few queries. Since $P$ has to be unique for such an algorithm to exist, we are always in the regime when the fraction of errors is less than half the distance of the code i.e., $\varepsilon < \delta(\mathcal{F})/2$. 

We say that $\mathcal{F}$ is {\em $(\varepsilon, L)$-list-decodable} if if for every function $f:D\to \mathcal{F}$, there exists at most $L$ functions $P \in \mathcal{F}$ such that $\delta(f,P) \le \varepsilon$. While this is a purely combinatorial guarantee for the code, the notion of local list-correction makes it more ``algorithmic''.

We say that $\mathcal{F}$ is {\em $(\varepsilon,q_1,q_2,L)$ locally list-correctable} if there exists a $q_1$-query algorithm $\mathcal{A}$, which when given query access to $f$, outputs at most $L$ many $q_2$-query local correction algorithms $\mathcal{A}_1,\mathcal{A}_2,\dots,\mathcal{A}_L$ such that for every $P\in \mathcal{F}$ such that $\delta(f,P) \le \varepsilon$, there exists at least one index $i\in [L]$ such that $\mathcal{A}_i$ is a local correction algorithm for $P$ i.e., on input $i\in D$, it makes $q_2$ queries to $f$ and outputs $P(i)$ with probability at least $3/4$. 

For an Abelian group $G$, let $\mathcal{J}_d(\sgrid^n, G)$ (or simply $\J_{d}$ when $\sgrid$ and $G$ are clear from context) denote the family of functions from $\sgrid^n \to G$ that can be expressed as a sum of $d$-juntas (i.e., a {\em $d$-junta-sum}). We may sometimes also refer to $d$-junta-sums as functions of {\em junta-degree} $d$.
We then have the following observation regarding $d$-junta-sums.\\

\begin{claim}[{\em Distance of junta-sums}, see e.g. \cite{AmireddySS}, Claim 2.7]\label{clm:dist-junta-sums}
    For every two distinct junta-sums $P\ne Q\in \J_d(\mathcal{S}^n, G)$ where $|\mathcal{S}|=s$, we have
    $$\Pr_{{\bf a}\sim \mathcal{S}^n} [P({\bf a})\ne Q({\bf a})] \ge \frac{1}{s^d}.$$
\end{claim}

That is, junta-sums form a code of distance $\delta_{\mathcal{J}} = 1/s^d$ where $s=|\mathcal{S}|$.
Indeed, the local correction and list-decodability properties of this family only depends on the size of $\mathcal{S}$, so we will often assume that $\mathcal{S}=\Z_s$ or $\mathcal{S}=[s]$ without loss of generality. We also use the following claim, where for $a\in \Z_s$, the function $\delta_a:\Z_s \to \Z$ is defined as: $\delta_0(x)=1$, and for $a\ne 0$, we define $\delta_a(x) = 1$ if $x=a$, and $\delta_a(x)=0$ otherwise. \\

\begin{claim}[{\bf Junta-polynomial representation}, \cite{AmireddySS} Claim 2.5]\label{clm:normal-form}
    Every $P\in \mathcal{J}_{d}(\Z_s^n, G)$ can be uniquely expressed as:
    \begin{align*}
        P(\x) \; = \; \sum_{{\bf a} \in \Z_s^n: |{\bf a}|\le d} g_{\bf a} \cdot \prod_{i\in [n]:a_i \ne 0} \delta_{a_i}(x_i), && \text{ where each } g_{\mathbf{a}} \in G.
    \end{align*}
\end{claim}

We call the above representation $\sum_{{\bf a}\in \Z_s^n} g_{{\bf a}} \cdot \prod_{i\in [n]:a_i \ne 0} \delta_{a_i}(x_i)$ as a {\em junta-polynomial} and its {\em junta-degree} is the size of the largest $|{\bf a}|$ such that the coefficient $g_{{\bf a}} \ne 0$; in particular, we call the terms being added as {\em monomials}. Generalizing~\Cref{clm:normal-form} one can show that every function $f:\Z_s^n \to G$ has a unique junta-polynomial representing it and $f$ is a $d$-junta-sum if and only if the degree of that junta-polynomial is at most $d$. In turn, this immediately implies that $f$ {\em depends} on the $i$-th coordinate if and only if the variable $x_i$ appears (as $\delta_{a}(x_i)$ for some $a \in \Z_s \setminus \{0\}$) in a non-zero monomial in the junta-polynomial representation. 

\subsection*{Partitions and Tableaux}
We end this section with some more terminology about integer partitions and multislices, which will be needed in our proofs. All the definitions in this subsection are standard and can be found in any standard text on algebraic combinatorics or representation theory for the symmetric group. For example, see \cite[Chapter 2]{Sagan} or \cite{Stanley1, Stanley2}.

\paragraph*{Partitions}For every natural number $n \in \mathbb{N}$, let $\mathcal{P}(n)$ denote the set of partitions of $n$. We will frequently use Ferrers diagram to represent partitions. Let $\lambda^{\ast} \in \mathcal{P}(n)$ denote the dual partition of $\lambda$.\newline

\paragraph*{SYT and SSYT}For a partition $\lambda \in \mathcal{P}(n)$, a \emph{standard Young tableau} is a tableau of shape $\lambda$ in which the entries in each row and each column are \emph{strictly} increasing. A \emph{semi-standard Young tableau} is a tableau of shape $\lambda$ in which the entries in each row are \emph{weakly} increasing and entries in each column are \emph{strictly} increasing. For a pair of partitions $\lambda, \mu \in \mathcal{P}(n)$, the set $\mathrm{SSYT}(\lambda,\mu)$ denotes the set of semi-standard Young tableaux of shape $\lambda$ and type $\mu$. Similarly, $\mathrm{SYT}(\lambda)$ denotes the set of standard Young tableaux of shape $\lambda$.\\

For any $\lambda, \mu \in \mathcal{P}(n)$, we associate two quantities:
\begin{itemize}
    \item $f_{\lambda}$ denotes the number of Standard Young Tableaux of shape $\lambda$ with content $[n]$, i.e. $f_{\lambda} = |\mathrm{SYT}(\lambda)|$.
    \item $K_{\lambda \mu}$ denotes the number of distinct Semi-Standard Young Tableaux of shape $\lambda$ and type $\mu$, i.e. $K_{\lambda\mu} = |\mathrm{SSYT}(\lambda,\mu)|$. This is also known as the Kostka number of the pair $(\lambda, \mu)$. 
\end{itemize}

\paragraph*{Dominance Order}For two partitions $\lambda,\mu \in \mathcal{P}(n)$, \emph{dominance order} is a partial order on partitions, defined as follows: Suppose $\lambda = (\lambda_{1},\ldots,\lambda_{\ell})$ and $\mu = (\mu_{1},\ldots,\mu_{m})$, then $\lambda \trianglerighteq \mu$ if for every $1 \leq i \leq \min\set{\ell,m}$,
\begin{align*}
    \lambda_{1} + \ldots + \lambda_{i} \; \geq \; \mu_{1} + \ldots + \mu_{i}.
\end{align*}

\paragraph*{}For every positive integer $n$, $\mathrm{Sym}_{n}$ denotes the group of permutations on $n$ elements and $\mathrm{Sym}[\sgrid]$ denotes the group of permutations on $\sgrid$.

\subsection*{Linear Algebra}
\paragraph*{Singular Value Decompositions}For a matrix $M \in \mathbb{R}^{n \times n}$, the \emph{singular value decomposition} (SVD) of $M$ is given by orthonormal matrices $U,V \in \mathbb{R}^{n \times n}$ such that:
\begin{align*}
    M \; = \; UDV^{-1} \; \Leftrightarrow \; M \; = \; U D V^{T}, && (V^{-1} = V^{T} \text{ for orthonormal } \, V),
\end{align*}
where $D$ is a diagonal entries and the diagonal entries of $D$ are the \emph{singular values} of $M$. We will denote the singular values of $M$ by $\sigma_{1} \geq \sigma_{2} \geq \cdots \geq \sigma_{n}$.\newline
In particular, if $M$ has all real eigenvalues, then $V = U$ and the singular values correspond to the absolute values of the eigenvalues.\\

\noindent
We now explain how SVDs behave under the tensor product. For two matrices $M_{1}$ and $M_{2}$ with SVDs
\begin{align*}
    M_{1} \; = \; U_{1} D_{1} V_{1}^{T} \quad \text{ and } \quad M_{2} \; = \; U_{2} D_{2} V_{2}^{T},
\end{align*}
then the SVD of $M_{1} \otimes M_{2}$ is:
\begin{equation}\label{eqn:SVD-tensor}
    M_{1} \otimes M_{2} \; = \; (U_{1} \otimes U_{2}) \cdot (D_{1} \otimes D_{2}) \cdot (V_{1} \otimes V_{2})^{T}.
\end{equation}

\begin{definition}[Volume of a parallelepiped]\label{defn:volume}
Suppose $\set{v_{1},\ldots,v_{r}} \in \mathbb{R}^{r}$ is a set of linearly independent vectors. Fix an arbitrary total order '$\prec$' on the vectors i.e., there exists a $\pi \in \mathrm{Sym}_{r}$ such that
\begin{align*}
    v_{\pi(1)} \prec v_{\pi(2)} \prec \cdots \prec v_{\pi(n)}.
\end{align*}
Let $\widetilde{v}_{\pi(1)} = v_{\pi(1)}$ and for every $2 \leq i \leq r$, let $\widetilde{v}_{\pi(i)}$ denote the vector orthogonal to $\mathrm{span}\setcond{v_{\pi(j)}}{j < i}$.\newline
Then the \emph{volume} of the parallelepiped spanned by $v_{1},\ldots,v_{r}$, denoted by $\mathrm{Vol}(\set{v_{1},\ldots,v_{r}})$, is defined to be
\begin{align*}
    \prod_{j=1}^{r} \| \widetilde{v}_{i} \|,
\end{align*}
where $\| \cdot \|$ is the norm with respect to the standard inner product on $\mathbb{R}$.\newline
It also turns out that the volume is equal to $|\det(\Lambda)|$ where the columns of $\Lambda$ are $v_{1},\ldots,v_{r}$.    
\end{definition}

\noindent
For any matrix $A \in \mathbb{R}^{r \times r}$, we will denote by $\| A \|_{2}$ the spectral norm of $A$ i.e.
\begin{equation}\label{eqn:spectral-norm}
    \| A \|_{2} \; = \; \sup_{\mathbf{x} \neq \mathbf{0}} \dfrac{\| A \mathbf{x} \|_{2}}{\| \mathbf{x} \|_{2}} \; = \; \max_{\sigma \; \text{ is a singular value of } \; A} \sigma.
\end{equation}

\subsection*{Subgrids of $\sgrid^{n}$}
It will be very useful in our algorithms to be able to restrict the given function to a smaller subgrid and analyze this restriction. We construct such subgrids by first permuting a subset of the variables and then identifying them into a smaller set of variables. More precisely, we have the following definition.\\

\begin{definition}[Embedding a smaller grid into $\sgrid^n$]\label{defn:random-embedding}
Fix any $k \in \mathbb{N}$ and $k \leq n$. Let $h : [n] \to [k]$ be a hash function. For each $i \in [n]$, let $\Pi_{i} \in \mathrm{Sym}[\sgrid]$ (i.e. a permutation on elements of $\sgrid$) and $\Pi = (\Pi_1, \dots, \Pi_n)$.
For every $\mathbf{y} \in \sgrid^{k}$, define $x_{h,\Pi}(\mathbf{y}) \in \sgrid^{n}$ as follows:
\begin{align*}
    x_{h,\Pi}(\mathbf{y})_{i} = \Pi_i(y_{h(i)}), && \text{ for all } \; i \in [n]
\end{align*}
and the subset $C_{h,\Pi} \subset \sgrid^{n}$ is defined as:
\begin{align*}
    C_{h,\Pi} \; = \; \setcond{x_{h,\Pi}(\mathbf{y})}{\mathbf{y} \in \sgrid^{k}}
\end{align*}
Further, a random subgrid $C_{h,\Pi}$ is obtained by sampling a uniformly random permutation $\Pi_i \sim \mathsf{Sym}[\sgrid]$ independently for all $i \in [n]$ and sampling a uniformly random hash function $h: [n] \to [k]$.
\end{definition}

\noindent
In simple words, the above definition gives us a way to embed a $k$-dimensional grid $\sgrid^{k}$ into a $n$-dimensional grid $\sgrid^{n}$, where the hash function $h$ governs how the $k$-coordinates are mapped into $n$-coordinates and $\Pi$ governs which value the $i^{th}$ coordinate takes.

The following sampling lemma (proved in~\Cref{app:sampling}) will be useful for local (list) correction of junta-sums. \\

\begin{lemma}[{\bf Sampling lemma for random subgrids}]
\label{lemma:sampling-subgrid}
Let $C_{h,\Pi} \subset \sgrid^n$ be a subgrid sampled randomly as per \Cref{defn:random-embedding}. Fix any $T\subseteq \sgrid^n$ and let $\mu:= |T|/s^n.$ Then, for any $\varepsilon, \eta > 0$
\begin{align*}
    \Pr_{h,\Pi} \, \left[\left|\frac{|T\cap C_{h,\Pi}|}{s^k} - \dfrac{|T|}{s^{n}} \right| \geq \varepsilon \right] \; < \; \eta
\end{align*}
as long as $k\geq \max\left\{\frac{A}{\varepsilon^8\eta^4}\cdot \log\left(\frac{1}{\varepsilon\eta}\right), B\cdot s^4\log s \right\}$ for a large enough absolute constants $A, B > 0.$
\end{lemma}

\section{Singular Value Bounds for Random Walks on Balanced Multislices} \label{sec:eigenvalue}
\paragraph{Organization of this section.}In this section, we will prove \Cref{thm:more-general-matrix-eigenvalue}. At a high level, the proof proceeds as follows:
\begin{enumerate}
    \item We give the necessary background on representation theory for finite groups in \Cref{prop:standard-rep-theory}. We then instantiate it for $\mathrm{Sym}_{n}$ and state the exact requirements we need to prove for our purpose in \Cref{thm:instantiation-rep-thy-sym}. These two steps can be found in \Cref{subsec:rep-theory-primer}.

    \item We then in \Cref{subsec:main-thm-proof} argue that \Cref{thm:instantiation-rep-thy-sym} is sufficient to prove \Cref{thm:more-general-matrix-eigenvalue}.

    \item We devote \Cref{subsec:alt-special-vectors} and \Cref{subsec:proof-instantiation} to prove \Cref{thm:instantiation-rep-thy-sym}. In particular, we start with describing some ``special vectors'' which we use to prove \Cref{thm:instantiation-rep-thy-sym}. The description of these vectors is combinatorial in nature, and we prove certain properties about them. Finally, we prove \Cref{thm:instantiation-rep-thy-sym} in \Cref{subsec:proof-instantiation}.
\end{enumerate}

\paragraph{Notation}For any two natural numbers $n$ and $s$ with $n$ divisible by $s$, $\mu$ denotes the $s$-tuple $(n/s,\ldots,n/s)$ and $\sgrid^{n}_{\mu}$ is the set of all points in $\sgrid^{n}$ which are on the balanced multi-slice $\mu$. Let $N := |\sgrid^{n}_{\mu}| = \binom{n}{n/s,\ldots,n/s}$. For any $n$, $\mathcal{P}(n)$ denotes the set of partitions of $n$. Throughout this section, we will assume that $s$ is an absolute constant.

\subsection{Representation Theory Primer}\label{subsec:rep-theory-primer}
In this and the following subsections, whenever we mention a representation, we refer to a complex finite-dimensional representation of finite groups. For interested readers, we refer to \cite[Chapter 1]{Sagan} for the relevant background on the representation theory of finite groups.

\paragraph{}Let $G$ be a finite group and $V$ be a $\mathbb{C}$-vector space with $\dim(V) < \infty$. Let $\langle \cdot, \cdot \rangle : V \times V \to \mathbb{C}$ be an inner product that is preserved under the representation $\rho$, i.e. for every $g \in G$, for every $u,v \in V$,
\begin{align*}
    \langle \rho(g)u, \; \rho(g)v \rangle \; = \; \langle u, v \rangle.
\end{align*}
 Basic results of representation theory imply the following. \\

\begin{proposition}[Standard facts on representations for finite groups]\label{prop:standard-rep-theory}
Suppose $\rho: G \to \mathrm{GL}(V)$ is a representation of $G$ and $W \in \mathrm{End}(V,V)$ commutes with the representation $\rho$, i.e. for every $g \in G$,
\begin{align*}
    \rho(g) \circ W \; \equiv \; W \circ \rho(g) && \text{(equivalent as linear operators).}
\end{align*}
In other words, $W$ is an intertwining operator from $(\rho, V)$ to itself. Then,
\begin{enumerate}
    \item There exists sub-representations $V_{1},\ldots,V_{r}$ such that
    \begin{align*}
        V \; \cong \; \bigoplus_{i=1}^{r} V_{i}
    \end{align*}
     where $\set{V_{1},\ldots,V_{r}}$ are orthogonal subspaces (with respect to the inner product mentioned above).\newline
     Moreover, for every $i \in [r]$, the following holds. There exists an irreducible representation $U_{i}$ and an integer $m_{i} \geq 1$ such that
    \begin{align*}
        V \; \cong \; \bigoplus_{j=1}^{m_{i}} \; V_{i,j} \quad \text{ and } \quad V_{i,1} \cong \cdots \cong V_{i,m_{i}} \cong U_{i}.
    \end{align*}
    The subrepresentation $V_{i}$ is called as the \emph{isotypic component} of $(\rho, V)$ corresponding to the irreducible representation $U_{i}$. This means
    \begin{align*}
        \dim(V_{i}) = m_{i} \cdot \dim(U_{i}) \quad \text{ and } \quad \dim(V) = \sum_{i=1}^{r} m_{i} \cdot \dim(U_{i}).
    \end{align*}

    \item Fix any $i \in [r]$. Let $\mathcal{B}_{i,1}$ be an ordered basis for $V_{i,1}$. For every $2 \leq j \leq m_{i}$, there exists an unique isomorphism $\mathcal{L}_{i,j} : V_{i,1} \to V_{i,j}$. Let $\mathcal{B}_{i,j}$ denote the image of $\mathcal{B}_{i,1}$ under $\mathcal{L}_{i,j}$, and $\mathcal{B}_{i,j}$ is a basis for $V_{i,j}$. Let $\mathcal{B}_{i}$ be an ordered basis for $V_{i}$ obtained by concatenating $\mathcal{B}_{i,1},\ldots, \mathcal{B}_{i,m_i}$ in that order. Similarly, $\mathcal{B}$ obtained by concatenating $\mathcal{B}_{1},\ldots, \mathcal{B}_r$ is an ordered basis for $V$.

    \item The linear map $W$ preserves the isotypic components, i.e. for each $i \in [r]$, $W|_{V_{i}} \in \mathrm{End}(V_{i}, V_{i})$. In particular, under the ordered basis $\mathcal{B}$, the map $W$ (when viewed as a $\dim(V) \times \dim(V)$ matrix) has the following structure:
    \begin{align*}
        W \; = \; \bigoplus_{i=1}^{r} W_{i}, && (\text{direct sum of matrices})
    \end{align*}
    where for each $i \in [r]$, $W_{i}$ is a $\dim(V_{i}) \times \dim(V_{i})$ dimensional matrix.\newline
    Furthermore, for every $i \in [r]$, there exists a unique $m_{i} \times m_{i}$ dimensional matrix $A_{i}$ such that
    \begin{align*}
        W_{i} \; = \; A_{i} \, \otimes \, \mathrm{Id}_{\dim(U_{i})}, && \text{where } \, \mathrm{Id_{k}} \text{ is the $k \times k$ dimensional identity matrix.}
    \end{align*}

    \item Fix any $i \in [r]$. For every non-zero $v \in V_{i,1}$, define a $\mathbb{C}$-space $Y_{i,v} := \mathrm{span}\set{v,\mathcal{L}_{i,2}(v),\ldots,\mathcal{L}_{i,m_{i}}(v)}$. Then $W_{i}|_{Y_{i,v}} \in \mathrm{End}(Y_{i,v}, Y_{i,v})$ and $W_{i}|_{Y_{i,v}} = A_{i}$ when we represent the linear map in the ordered basis $(v,\mathcal{L}_{i,2}(v),\ldots,\mathcal{L}_{i,m_{i}}(v)).$
\end{enumerate}
\end{proposition}

\noindent
The following corollary is immediate from the third item of \Cref{prop:standard-rep-theory}.\\

\begin{corollary}\label{coro:singular-value-indexing}
We follow the same notation from \Cref{prop:standard-rep-theory}. Suppose $\set{\beta^{i}_{1},\ldots,\beta^{i}_{m_{i}}}$ is the multi-set of singular values of $A_{i}$ where each $\beta^{i}_{j} \in \mathbb{C}$. Then using \Cref{eqn:SVD-tensor}, we get that in the multi-set of singular values of $W_{i}$, the frequency of $\beta^{i}_{j}$ is equal to the frequency of $\beta^{i}_{j}$ in the multi-set $\set{\beta^{i}_{1},\ldots,\beta^{i}_{m_{i}}}$ times $\dim(U_{i})$. 
\end{corollary}

\paragraph{}Now we turn to the representation theory for $\mathrm{Sym}_{n}$. In particular, we will be considering the representation of $\mathrm{Sym}_{n}$ on the space of functions on a slice of $\sgrid^{n}$.

\paragraph{\underline{Space of functions on a slice}}For any partition $\lambda \in \mathcal{P}(n)$, let $\mathbb{M}^{\lambda}$ denote the $\mathbb{C}$-vector space of functions over the slice $\sgrid^{n}_{\lambda}$ i.e.
\begin{align*}
    \mathbb{M}^{\lambda} \; = \; \set{f:\sgrid^{n}_{\lambda} \to \mathbb{C}}.
\end{align*}
It is easy to see that $\dim(\mathbb{M}^{\lambda}) = |\sgrid^{n}_{\lambda}|$. There is a natural action of the symmetric group $\mathrm{Sym}_{n}$ on $\mathbb{M}^{\lambda}$: For all $\pi \in \mathrm{Sym}_{n}$ and for all $f \in \mathbb{M}^{\lambda}$,
\begin{align*}
    (\pi f)(\mathbf{x}) \, = \, f(\pi^{-1} \mathbf{x}), && \text{ where } \; \pi^{-1} \mathbf{x} = (x_{\pi^{-1}(1)},\ldots,x_{\pi^{-1}(n)})
\end{align*}

\paragraph{\underline{Representation of $\mathrm{Sym}_{n}$}}Let $(\rho, \mathbb{M}^{\mu})$ denote the following $\mathbb{C}$ representation of $\mathrm{Sym}_{n}$:
\begin{align*}
    \rho: \mathrm{Sym}_{n} \to \mathrm{GL}(\mathbb{M}^{\mu})
\end{align*}
\begin{equation}\label{eqn:rep-definition}
     (\rho(\pi) \; f)(\mathbf{x}) \; = \; f(\pi^{-1}\mathbf{x}), \quad \quad \forall \; \pi \in \mathrm{Sym}_{n}, \quad \forall \; f \in \mathbb{M}^{\mu}.
\end{equation}

\paragraph{\underline{Invariant inner product}}Next we mention an inner product $\langle \cdot, \cdot \rangle$ on the space $\mathbb{M}^{\mu} \times \mathbb{M}^{\mu}$ which will be invariant under the representation $\rho$. The inner product is defined as follows:
\begin{align*}
    \langle \cdot, \cdot \rangle : \mathbb{M}^{\mu} \times \mathbb{M}^{\mu} \to \mathbb{R}_{\geq 0}    
\end{align*}
\begin{equation}\label{eqn:inner-prod}
    \langle f,g \rangle \; = \; \mathbb{E}_{\mathbf{x} \sim \sgrid^{n}_{\mu}} [f(\mathbf{x}) \cdot g(\mathbf{x})] \; = \; \dfrac{1}{N} \sum_{\mathbf{x} \in \sgrid^{n}_{\mu}} [f(\mathbf{x}) \cdot g(\mathbf{x})].
\end{equation}
It is not hard to see that the above inner product is invariant under the representation $\rho$, i.e., for every $f,g \in \mathbb{M}^{\mu}$, the following holds:
\begin{align*}
    \langle \rho(\pi)f, \; \rho(\pi)g \rangle \; = \; \langle f, \; g \rangle.
\end{align*}

This representation is quite well-studied in the representation theory for finite groups. There is a complete understanding of the decomposition of $\mathbb{M}^{\lambda}$ into its irreducible representations. In particular, the irreducible representations of $(\rho, \mathbb{M}^{\mu})$ are given by the \emph{Specht modules} $\mathbb{S}^{\lambda} \subset \mathbb{M}^{\lambda}$. See \cite[Chapter 2]{Sagan} for an excellent exposition on the irreducible decompositions. Next, we state \Cref{thm:instantiation-rep-thy-sym}, which we will use to prove \Cref{thm:more-general-matrix-eigenvalue}. We do not require specific details of the irreducible representation, so we only state what is sufficient for our purpose.\\

\begin{theorem}\label{thm:instantiation-rep-thy-sym}
Fix any $n,s \in \mathbb{N}$ where $n$ is divisible by $s$ and let $\mu = (n/s,\ldots,n/s) \in \mathcal{P}(n)$. The following holds:
\begin{enumerate}
    \item The subrepresentations of $\mathbb{M}^{\mu}$ are indexed by $\lambda \in \mathcal{P}(n)$ and in particular, there exists subrepresentations $V_{\lambda,1},\ldots,V_{\lambda,m_{\lambda}}$ for an integer $m_{\lambda} \in \mathbb{N}$ such that:
    \begin{align*}
        \mathbb{M}^{\mu} \; \cong \; \bigoplus_{\lambda \trianglerighteq \mu} \; \bigoplus_{j=1}^{m_{\lambda}} V_{\lambda,j},
    \end{align*}
    where $V_{\lambda,1} \cong \cdots \cong V_{\lambda,m_{\lambda}}$.
    \begin{enumerate}
        \item For $\lambda = (n)$, $m_{\lambda} = 1$ and $\dim(V_{\lambda,1}) = 1$. This corresponds to the trivial subrepresentation spanned by the function that takes the value $1$ at each point of $\mathcal{S}^n_\mu$.
    \end{enumerate}

    \item  Let $c \in \mathbb{N}$ denote an absolute constant $> 1$. For every partition $\lambda = (\lambda_{1},\ldots,\lambda_{\ell}) \in \mathcal{P}(n)$, we have
    \begin{enumerate}
        \item If $\lambda_{2} > c$, then $\dim(V_{\lambda,1}) \, = \cdots = \, \dim(V_{\lambda,m_{\lambda}}) \; \geq \; n^{0.8c}$.

        \item If $\lambda_{2} \leq c$, then $m_{\lambda} \; \leq \; s^{cs}.$ As $s$ and $c$ are constants, $m_{\lambda} = \bigO_{s,c}(1)$.
    \end{enumerate}

    \item For every constant $c \in \mathbb{N}$, the following holds. Fix any $\lambda = (\lambda_{1},\ldots,\lambda_{\ell}) \in \mathcal{P}(n)$ such that $\lambda_{2} \leq c$. Then there exists vectors $u^{\lambda}_{1},\ldots,u^{\lambda}_{m_{\lambda}} \in \mathbb{M}^{\mu}$ where for every $j \in [m_{\lambda}]$, the vector $u^{\lambda}_{j} \in V_{\lambda,j}$, satisfying the following conditions:
    \begin{enumerate}
        \item For every $j > 1,$ $u^\lambda_j$ is the image of $u^\lambda_1$ under the unique isomorphism between representations $V_{\lambda,1}$ and $V_{\lambda,j}$.
        \item For every $j \in [m_{\lambda}]$, $\| u^{\lambda}_{j} \|_{2} = \Theta_{s,c}(1)$. Here, the norm is with respect to the invariant inner product stated in \Cref{eqn:inner-prod}.
        \item If $\mathcal{D}$ is a probability distribution on the balanced multislice $\sgrid^{n}_{\mu}$ such that $\mathcal{D}$ is an $\varepsilon$-almost $k$-wise independent and uniform distribution for some $k \geq cs$, then,
        \begin{align*}
            \bigg|\mathbb{E}_{\mathbf{x} \sim \mathcal{D}}[u^{\lambda}_{j}(\mathbf{x})] \bigg| \; \leq \; \bigO_{s,c}(\varepsilon), && \text{ for all } \; j \in [m_{\lambda}].
        \end{align*}
        \item The volume of the parallelepiped (see \Cref{defn:volume}) formed by $\set{u^{\lambda}_{j}}_{j=1}^{m_{\lambda}}$ is at least a constant, i.e., $\mathrm{Vol}(u^{\lambda}_{1},\ldots,u^{\lambda}_{m_{\lambda}}) \; = \; \Omega_{s,c}(1)$.
    \end{enumerate}
\end{enumerate} 
\end{theorem}

\noindent
\begin{remark}
    \label{rem:instantiation-rep-thy-sym}
    We note that items $1$ and $2$ in the \Cref{thm:instantiation-rep-thy-sym} are standard results in the representation theory of $\mathrm{Sym}_{n}$, or simple consequences thereof. Also see \Cref{thm:irrep-decomposition} for more details on item $1$. The new technical observations we make are in proving items 3.(c) and 3.(d) of \Cref{thm:instantiation-rep-thy-sym}. As we will elaborate later in \Cref{subsec:alt-special-vectors}, our proof for item $3$ of \Cref{thm:instantiation-rep-thy-sym} uses \cite{complexity-symmetric-group}. In this work, we analyze a set of functions described already in \cite{complexity-symmetric-group} and show that they satisfy additional properties, which allows us to prove our main theorem (\Cref{thm:more-general-matrix-eigenvalue}). 
    
\end{remark}

We will first show how \Cref{thm:instantiation-rep-thy-sym} implies \Cref{thm:more-general-matrix-eigenvalue}. We defer the proof of \Cref{thm:instantiation-rep-thy-sym} to \Cref{subsec:proof-instantiation}. We will also require the following lemma on estimating the singular values of small matrices. It says that if we have a set of linearly independent vectors whose parallelepiped has a significant volume, then they are ``useful'' in estimating the singular values.\\

\begin{restatable}[Estimating singular values using special vectors]{lemma}{estimatingeigenvalue}\label{lemma:estimating-eigenvalue}
As a special case, in this lemma, we will work with the standard inner product in Euclidean space $\mathbb{R}^r.$ The lengths, volumes etc. below are defined using this standard inner product.

Let $Q \in \mathbb{R}^{r \times r}$ be a matrix and $\set{v_{1},\ldots,v_{r}}$ be a set of linearly independent vectors satisfying the following conditions:
\begin{enumerate}
    \item For each $j \in [r]$, $\| v_{j} \|_{2} \leq m$ for some $m \in \mathbb{R}_{> 0}$.
    \item For each $j \in [r]$, $\| Qv_{j} \|_{2} \leq q$ for some $q \in \mathbb{R}_{> 0}$.
    \item The volume (recall \Cref{defn:volume}) $\mathrm{Vol}(v_{1},\ldots,v_{r}) \geq \tau$.
\end{enumerate}
Then\footnote{Recall the notation in \Cref{eqn:spectral-norm}.} $\| Q \|_{2} \leq r \cdot  \dfrac{\max\set{m^{r},1} \cdot r!}{\tau} \cdot q$.
\end{restatable}
\begin{proof}[Proof of \Cref{lemma:estimating-eigenvalue}]
 By definition of spectral norm,
\begin{align*}
    \| Q \|_{2} \; = \; \sup_{\mathbf{x}: \| \mathbf{x}\|_{2}=1} \| Q \mathbf{x} \|_{2}, && \text{ where } \; \mathbf{x} \in \mathrm{span}(v_{1},\ldots,v_{r}).
\end{align*}
Choose an arbitrary $\mathbf{x} \in \mathrm{span}(v_{1},\ldots,v_{r})$ with $\| \mathbf{x} \|_{2} = 1$. We know there exists coefficients $\alpha_{1},\ldots,\alpha_{r} \in \mathbb{C}$ such that $\mathbf{x} = \sum_{i=1}^{r} \, \alpha_{i} v_{i}$. We first upper bound $|\alpha_{i}|$ for all $i \in [r]$.

\noindent
Let $\Lambda \in \mathbb{R}^{r \times r}$ denote the matrix whose columns are $v_{1},\ldots,v_{r}$. Let $\bm{\alpha} = (\alpha_{1},\ldots,\alpha_{r})$. Then $\Lambda \bm{\alpha} = \mathbf{x}$. In other words,
\begin{align*}
    \alpha_{i} \; = \; \dfrac{\det(\Lambda_{i})}{\det(\Lambda)},
\end{align*}
where $\Lambda_{i}$ is the matrix whose $i^{th}$ column $v_{i}$ is replaced by $\mathbf{x}$. We have $\| \mathbf{x} \|_{\infty} \leq \| 
\mathbf{x} \|_{2} = 1$ and for every $j \in [r]$, $\| v_{j} \|_{\infty} \leq \| v_{j} \|_{2} \leq m$. Combining these two, we have $\| \Lambda_{i} \|_{\infty} \leq \max\set{m,1}$. This implies that $\det(\Lambda_{i}) \leq r! \cdot \| \Lambda_{i} \|_{\infty}^{r} \leq \max\set{m^{r},1} \cdot r!$. Recall also that the volume of the parallelepiped spanned by $v_1,\ldots, v_r$ is given by $|\det(\Lambda)|.$ Thus we have,
\begin{align*}
    |\alpha_{i}| \; \leq \; \dfrac{\max\set{m^{r},1} \cdot r!}{\tau}, && \text{ for all } \; i \in [r].
\end{align*}
Now let us consider $\| Q \mathbf{x}  \|_{2}$:
\begin{align*}
    \| Q \mathbf{x} \|_{2} \; \leq \; \sum_{i=1}^{r} \, |\alpha_{i}| \cdot \|Q v_{i} \|_{2} \; \leq \; r \cdot  \dfrac{\max\set{m^{r},1} \cdot r!}{\tau} \cdot q.
\end{align*}
This finishes the proof of \Cref{lemma:estimating-eigenvalue}.
\end{proof}

\subsection{Proof of Our Main Technical Theorem (\Cref{thm:more-general-matrix-eigenvalue})}\label{subsec:main-thm-proof}
In this subsection, we give the proof of our main technical theorem (\Cref{thm:more-general-matrix-eigenvalue}), assuming \Cref{thm:instantiation-rep-thy-sym}. We defer the proof of \Cref{thm:instantiation-rep-thy-sym} to \Cref{subsec:proof-instantiation}.\\

\begin{proof}[Proof of \Cref{thm:more-general-matrix-eigenvalue} (using \Cref{thm:instantiation-rep-thy-sym})]

The first condition on the matrix $M$ implies that $M$ commutes with the representation $(\rho, \mathbb{M}^{\mu})$ (see \Cref{eqn:rep-definition}). Using the third item of \Cref{prop:standard-rep-theory} and \Cref{coro:singular-value-indexing}, we know that the singular values of $M$ can be divided into groups indexed by partitions $\lambda \in \mathcal{P}(n)$. We can classify the singular values of $M$ into three categories:
\begin{itemize}
    \item Singular values corresponding to the partition $\lambda = (n)$. As stated in 1.(a) of \Cref{thm:instantiation-rep-thy-sym}, it corresponds to the $1$-dimensional vector space and thus has singular value $1$.

    \item Singular values indexed by partitions $\lambda = (\lambda_{1},\ldots,\lambda_{\ell}) \in \mathcal{P}(n)$ with $\lambda_{2} > c > 1$ (here $c$ is the constant from \Cref{thm:instantiation-rep-thy-sym}).

    \item Singular values indexed by partitions $\lambda = (\lambda_{1},\ldots,\lambda_{\ell}) \in \mathcal{P}(n)$ with $\lambda_{2} \leq c$ and $\lambda \neq (n)$ (here $c$ is the constant from \Cref{thm:instantiation-rep-thy-sym}).
\end{itemize}

\paragraph{}To bound $\sigma_{2}(M)$, we only need to upper bound the singular values in the second and third categories. Before proceeding, we set some notation for convenience. Applying the third item of \Cref{prop:standard-rep-theory} on the matrix $M$, we know the following: For every partition $\lambda \in \mathcal{P}(n)$ with $\lambda \trianglerighteq \mu$, there exists a square matrix $\widetilde{M}_{\lambda}$ of dimensions $m_{\lambda} \times m_{\lambda}$ such that
\begin{align*}
    M \; = \; \bigoplus_{\lambda \trianglerighteq \mu} \; \paren{\widetilde{M}_{\lambda} \otimes \mathrm{Id}_{\dim(V_{\lambda})}}, && \text{ where } \dim(V_{\lambda}) \, = \, m_{\lambda} \cdot \dim(V_{\lambda,1}).
\end{align*}

\noindent
\Cref{coro:singular-value-indexing} tells us that the multiset of singular values of $M$ is essentially governed by the multiset of singular values of $\widetilde{M}_{\lambda}$'s for different $\lambda$'s. For every $\lambda \trianglerighteq \mu$, let $\set{\beta^{\lambda}_{1},\ldots,\beta^{\lambda}_{m_{\lambda}}}$ denote the multiset of singular values of $\widetilde{M}_{\lambda}$ (i.e. we account for repetitions too).

\paragraph{}To upper bound the singular values of the second category, we use the upper bound on the \textbf{Frobenius norm} of $M$. More precisely, we will show the following lemma.

\begin{lemmabox}
\begin{lemma}\label{lemma:large-block-tiny-eigenvalue}
Let $\kappa, c_{1} > 0$ be constants such that the Frobenius norm $\| M \|_{F} \leq c_{1} \cdot n^{\kappa}$ (see \Cref{thm:more-general-matrix-eigenvalue}) and let $c = 4 \kappa$. Let $\lambda = (\lambda_{1},\ldots,\lambda_{\ell}) \in \mathcal{P}(n)$ with $\lambda_{2} > c$. Then\footnote{Recall the notation in \Cref{eqn:spectral-norm}.},
\begin{align*}
    \| \widetilde{M}_{\lambda} \|_{2} \; \leq \; \dfrac{1}{c_{1}' \cdot n^{\kappa}},
\end{align*}
where $c_{1}'$ is a constant depending on $c_{1}$ and $\kappa$.
\end{lemma}
\end{lemmabox}
\begin{proof}[Proof of \Cref{lemma:large-block-tiny-eigenvalue}]
We know that the square of the Frobenius norm equals the sum of singular values squared, i.e. if $\set{\beta^{\gamma}_{1},\ldots,\beta^{\gamma}_{m_{\gamma}}}$ is the multiset of singular values of $\widetilde{M}_{\gamma}$, then
\begin{align*}
     \sum_{\gamma \trianglerighteq \mu} \dim(V_{\gamma,1}) \cdot \sum_{i=1}^{m_{\gamma}} |\beta^{\gamma}_{i}|^{2} \; \leq \; \| M \|_{F}^{2}.
\end{align*}
Fix an arbitrary $\lambda \in \mathcal{P}(n)$ with $\lambda_{2} > c$. As every term on the left side of the above inequality is a non-negative number, we have the following inequality for any $i \in [m_{\lambda}]$:
\begin{align*}
\dim(V_{\lambda,1}) \cdot|\beta^{\lambda}_{i}|^{2} \; \leq \; \| M \|_{F}^{2} \\ \\
\Rightarrow \quad \Omega_{\kappa}( n^{4 \kappa} ) \cdot |\beta^{\lambda}_{i}|^{2} \; \leq \;  c_{1}^{2} \cdot n^{2\kappa} && \text{(Using 2.(a) of \Cref{thm:instantiation-rep-thy-sym} and 2 of \Cref{thm:more-general-matrix-eigenvalue}}) \\ \\
\Rightarrow |\beta^{\lambda}_{i}| \; \leq \; \dfrac{1}{c_{1} \cdot n^{\kappa}}.
\end{align*}
Since the above upper bound holds for every $i \in [m_{\lambda}]$, we get the desired bound on $\| \widetilde{M}_{\lambda} \|_{2}$. This finishes the proof of \Cref{lemma:large-block-tiny-eigenvalue}.    
\end{proof}

\noindent
Next, we have to upper bound the singular values of the third category, and this requires more steps in comparison to the previous lemma. We use the \textbf{$\varepsilon$-almost  $k$-wise independence} of $M$ in this step. We start by stating the bound.

\begin{lemmabox}
\begin{lemma}\label{lemma:small-block-small-eigenvalue}
 Let $c = 4 \kappa$ (the same constant from \Cref{lemma:large-block-tiny-eigenvalue}) and $\lambda \in \mathcal{P}(n)$ be a partition with $\lambda \trianglerighteq \mu$, $\lambda_{2} \leq c$, and $\lambda \neq (n)$. Then,
\begin{align*}
    \| \widetilde{M}_{\lambda} \|_{2} \; \leq \; \bigO_{s,\kappa}(\varepsilon),
\end{align*}
where $\varepsilon$ is the distance parameter in the third item of \Cref{thm:more-general-matrix-eigenvalue}.
\end{lemma}
\end{lemmabox}
\begin{proof}[Proof of \Cref{lemma:small-block-small-eigenvalue}]
Fix a partition $\lambda \trianglerighteq \mu$ with $\lambda_{2} \leq c$ and $\lambda \neq (1,1,\ldots,1)$ for rest of the proof. Let $u^{\lambda}_{1},\ldots,u^{\lambda}_{m_{\lambda}}$ be the vectors guaranteed from the third item of \Cref{thm:instantiation-rep-thy-sym}. The idea is to use 
\Cref{lemma:estimating-eigenvalue} on the vectors $u^{\lambda}_{j}$'s and the matrix $\widetilde{M}_{\lambda}$, but we need to be careful, as we explain below.

\paragraph*{}From the third and fourth items of \Cref{prop:standard-rep-theory}, we have,
\begin{align*}
    Mu^{\lambda}_{j} \; = \; \widetilde{M}_{\lambda} u^{\lambda}_{j}, && \text{where $\widetilde{M}_{\lambda}$ is the operator $M|_{Y_{\lambda}}$} \\ \\
    \Rightarrow \; \| Mu^{\lambda}_{j} \|_{2} \; = \; \| \widetilde{M}_{\lambda}u^{\lambda}_{j} \|_{2},
\end{align*}
where both the norms are with respect to the invariant inner product defined in \Cref{eqn:inner-prod}.

\paragraph*{Fixing an orthonormal basis.}Let
\begin{align*}
    Y_{\lambda} := \mathrm{span}(u^{\lambda}_{1},\ldots,u^{\lambda}_{m_{\lambda}})
\end{align*}
and let $(w_{1},\ldots,w_{m_{\lambda}})$ be an ordered \textit{orthonormal} (with respect to the \emph{invariant} inner product defined in \Cref{eqn:inner-prod}) basis for the space $Y_{\lambda}$.\newline
Let $\widetilde{A}_{\lambda}$ denote the $m_{\lambda} \times m_{\lambda}$ matrix representing the operator $\widetilde{M}_{\lambda}$ under the orthonormal basis $(w_{1},\ldots,w_{m_{\lambda}})$.\newline

\noindent
The singular values remain invariant under the choice of basis\footnote{To see this quickly, note that singular values of a matrix $A$ are the positive square roots of the eigenvalues of $AA^{T}$ and eigenvalues are independent of the choice of basis.}, thus it is enough to bound $\| \widetilde{A}_{\lambda} \|_{2}$. The idea is to use \Cref{lemma:estimating-eigenvalue} on $\widetilde{A}_{\lambda}$ and vectors $u^{\lambda}_{j}$'s to bound $\| \widetilde{A}_{\lambda} \|_{2}$. There is some subtelty regarding norms in using \Cref{lemma:estimating-eigenvalue}, which one needs to be careful about.

\paragraph*{Expressing the $u^{\lambda}_{j}$'s in the orthonormal basis.}For every $j \in [m_{\lambda}]$, let
\begin{align*}
    u^{\lambda}_{j} = \alpha_{j,1} w_{1} + \ldots + \alpha_{j,m_{\lambda}} w_{m_{\lambda}} \quad \text{ and } \quad \bm{\alpha}_{j} := (\alpha_{j,1},\ldots,\alpha_{j,m_{\lambda}}).
\end{align*}
Then,
\begin{align*}
    \| u^{\lambda}_{j} \|_{2} \; = \; \| \bm{\alpha}_{j} \|_{2},
\end{align*}
where the left norm is with respect to the \textit{invariant} inner product defined in \Cref{eqn:inner-prod} and the right norm is the \textit{standard} inner product on $\mathbb{R}^{m_{\lambda}}$. Using 3.(b) of \Cref{thm:instantiation-rep-thy-sym}, we get that for every $j \in [m_{\lambda}]$, $\| \bm{\alpha}_{j} \|_{2} = \Theta_{s,c}(1)$.

\paragraph*{Norm after applying the operator $\widetilde{M}_{\lambda}$.}Now we have the following equality:
\begin{align*}
    \widetilde{M}_{\lambda} u^{\lambda}_{j} \; = \; \widetilde{A}_{\lambda} \bm{\alpha}_{j} \quad
    \Rightarrow \quad \| \widetilde{M}_{\lambda} u^{\lambda}_{j} \|_{2} \; = \; \| \widetilde{A}_{\lambda} \bm{\alpha}_{j} \|_{2},
\end{align*}
where the left norm is with respect to the invariant inner product defined in \Cref{eqn:inner-prod} and the right norm is with respect to the standard inner product. Hence, we get
\begin{align*}
    \| Mu^{\lambda}_{j} \|_{2} \; = \; \| \widetilde{A}_{\lambda} \bm{\alpha}_{j} \|_{2}, && \text{ for every } \; j \in [m_{\lambda}].
\end{align*}

\paragraph*{Upper bounding the norm after applying the operator}We now show that for every $j \in [m_{\lambda}]$, 
\begin{align*}
    \| \widetilde{A}_{\lambda} \bm{\alpha}_{j} \|_{2} \leq \bigO_{s,\kappa}(\varepsilon),
\end{align*}
where the norm is with respect to the standard inner product. From the previous paragraph, it is enough to show that for every $j \in [m_{\lambda}]$, the norm $ \| M u^{\lambda}_{j} \|_{2} \leq \bigO_{s,c}(\varepsilon)$, where the norm is with respect to the invariant inner product.\newline

\noindent
Using the definition of the invariant inner product from \Cref{eqn:inner-prod}, we get,
\begin{align*}
    \| M u^{\lambda}_{j} \|_{2}^{2} \; = \; \mathbb{E}_{\mathbf{x} \sim \sgrid^{n}_{\mu}} \paren{ \mathbb{E}_{\mathbf{y} \sim M(\mathbf{x})}[u^{\lambda}_{j}(\mathbf{y})] }^{2}
\end{align*}
For every $\mathbf{x} \in \sgrid^{n}_{\mu}$, the third item of \Cref{thm:more-general-matrix-eigenvalue} says that $M(\mathbf{x})$ is $\varepsilon$-almost $k$-wise independent for $k = 10s\kappa \geq cs$. Applying item 3.(c) of \Cref{thm:instantiation-rep-thy-sym} on $M(\mathbf{x})$ for an arbitrary $\mathbf{x} \in \sgrid^{n}_{\mu}$, we get
\begin{align*}
    \mathbb{E}_{\mathbf{y} \sim M(\mathbf{x})}[u^{\lambda}_{j}(\mathbf{y})] \; = \; \bigO_{s,\kappa}(\varepsilon).
\end{align*}
As this holds for every $\mathbf{x} \in \sgrid^{n}_{\mu}$, we get,
\begin{align*}
    \mathbb{E}_{\mathbf{x} \sim \sgrid^{n}_{\mu}} \paren{ \mathbb{E}_{\mathbf{y} \sim M(\mathbf{x})}[u^{\lambda}_{j}(\mathbf{y})] }^{2} \; = \; \mathbb{E}_{\mathbf{x} \sim \sgrid^{n}_{\mu}}\paren{\bigO_{s,c}(\varepsilon^{2})} \; = \; \bigO_{s,\kappa}(\varepsilon^{2}).
\end{align*}
Hence we have shown that for every $j \in [m_{\lambda}]$, we get
\begin{align*}
     \| \widetilde{A}_{\lambda} \bm{\alpha}_{j }\|_{2} \; = \; \| M u^{\lambda}_{j} \|_{2} \; = \; \bigO_{s,\kappa}(\varepsilon).
\end{align*}

\paragraph*{Volume of the parallelepiped.}
Fix an arbitrary order on the $u_{j}$'s and recall from \Cref{defn:volume} that
\begin{align*}
    \mathrm{Vol}(u^{\lambda}_{1},\ldots,u^{\lambda}_{m_{\lambda}}) = \| \Tilde{u}^{\lambda}_{1} \|_{2} \cdots \| \Tilde{u}^{\lambda}_{m_{\lambda}} \|_{2},
\end{align*}
where $\widetilde{u}_{j}$ is as defined in \Cref{defn:volume} and the above norms are with respect to the invariant inner product defined in \Cref{eqn:inner-prod}. Similarly, we have $\mathrm{Vol}(\bm{\alpha}_{1},\ldots,\bm{\alpha}_{m_{\lambda}})$, in which the norm is with respect to the standard inner product.\newline

\noindent
Observe that for every $j \in [m_{\lambda}]$, $\mathrm{span}(\Tilde{\bm{\alpha}}_{1},\ldots,\Tilde{\bm{\alpha}}_{j-1}) \cong \mathrm{span}(\Tilde{u}^{\lambda}_{1},\ldots,\Tilde{u}^{\lambda}_{j-1})$, i.e. they are isometric as inner product spaces. Now the component of $u^{\lambda}_{j}$ orthogonal to the $(j-1)$ dimensional subspace has the same norm (in the invariant inner product) as the component of $\bm{\alpha}_{j}$ has norm under the standard inner product. Thus,
\begin{align*}
    \mathrm{Vol}(u^{\lambda}_{1},\ldots,u^{\lambda}_{m_{\lambda}}) \; = \; \mathrm{Vol}(\bm{\alpha}_{1},\ldots,\bm{\alpha}_{m_{\lambda}}) = \Omega_{s,\kappa}(1),
\end{align*}
where the final lower bound is from item 3.(d) of \Cref{thm:instantiation-rep-thy-sym}.\\

\noindent
Now we apply \Cref{lemma:estimating-eigenvalue} on $\widetilde{A}_{\lambda}$ and vectors $\bm{\alpha}_{1},\ldots,\bm{\alpha}_{m_{\lambda}}$. Using 2.(b) of \Cref{thm:instantiation-rep-thy-sym}, we know that $m_{\lambda} = \bigO_{s,c}(1)$. This gives us the desired bound and finishes the proof of \Cref{lemma:small-block-small-eigenvalue}.
\end{proof}

\noindent
Thus, we have proved \Cref{lemma:large-block-tiny-eigenvalue} and \Cref{lemma:small-block-small-eigenvalue}, which gives an upper bound on the singular values of the second and third categories, respectively. Combining \Cref{lemma:large-block-tiny-eigenvalue} and \Cref{lemma:small-block-small-eigenvalue}, we finish the proof of \Cref{thm:more-general-matrix-eigenvalue}.
\end{proof}

\subsection{Construction of Special Vectors}\label{subsec:alt-special-vectors}

In this section, we describe the vectors specified in the third item of \Cref{thm:instantiation-rep-thy-sym}. The construction combines standard literature on the representation theory of the symmetric group~\cite{Sagan} with the recent work of \cite{complexity-symmetric-group}. We need to recall the definition here to show that they satisfy the properties claimed in \Cref{thm:instantiation-rep-thy-sym}.

Throughout this section, fix a partition $\lambda = (\lambda_1,\ldots, \lambda_\ell) \in \mathcal{P}(n)$ and assume that $2\leq \ell \leq s.$ We will  consider (Young) tableaux $T$ of shape $\lambda$, which contain cells $T[i,j]$ where $1\leq i\leq \ell$ and for $1\leq j\leq \lambda_{i}$. Further, we will also consider permutations of such tableaux by permutations that rearrange the elements in each column of $T$. Let $C_{\lambda} := \mathrm{Sym}_{\lambda_1^*}\times\cdots \times \mathrm{Sym}_{\lambda_{\lambda_1}^*}$ and given a permutation $\sigma \in C_\lambda$, we denote by $T^\sigma$ the tableau obtained by rearranging the contents of the cells of $T$ according to $\sigma.$ For every $\sigma = (\sigma^{(1)}, \ldots, \sigma^{(\lambda_{1})}) \in C_{\lambda}$, $\mathrm{sgn}(\sigma) := \mathrm{sgn}(\sigma^{(1)}) \cdots \mathrm{sgn}(\sigma^{(\lambda_{1})})$.

We define $T_0$ to be the canonical tableau of shape $\lambda$ where the cells are labelled as follows:
\begin{equation}\label{eqn:canonical-tableau}
    T_{0}[i,j] \; := \; \sum_{p < i}\lambda_p + j, \quad \quad i \in [\ell], \; j \in [\lambda_{i}].
\end{equation}
The following is a diagram for the canonical tableau $T_{0}$ for some partition $\lambda \in \mathcal{P}(n)$.
\begin{align*}
T_0 \quad = \quad 
\ytableausetup{notabloids, boxframe=normal, boxsize=3em} 
\begin{ytableau}
1 & 2 &  \none[\dots] & \lambda_1-1 & \lambda_1 \\
\lambda_1+1 & \lambda_{1} + 2 & \none[\dots] & \lambda_2 \\
\none[\vdots] & \none[\vdots] \\
n-2 & n-1 \\
n
\end{ytableau}
\end{align*}
 
 Next, we define a polynomial on $\sgrid^{n}_{\mu}$ for every semi-standard Young tableau of shape $\lambda$ and content\footnote{i.e. a tableau with $\mu_1$ many $0$s, $\mu_2$ many $1$s, and so on until $\mu_s$ many $(s-1)$s} $\mu$. Recall that $\mu \in \mathcal{P}(n)$ is the partition of $n$ given by $(\frac{n}{s},\ldots, \frac{n}{s})$.\\

 \begin{definition}\label{defn:chi-T}
 \cite[Section 5.2]{complexity-symmetric-group}. Given a tableau $T'$ of shape $\lambda$ with distinct labels from $[n]$ and another tableau $T$ of shape $\lambda$ with content $\mu$, we define a corresponding $\mathbb{R}$-valued function $e_{T',T}:\sgrid^n_\mu\rightarrow \{0,1\}$ by 
\begin{equation}\label{eq:defn-eT}
e_{T',T}(\bm{x}) = \begin{cases}
    1, & \text{if $\{x_{T'[i,1]},\ldots, x_{T'[i,\lambda_i]}\} = \{T[i,1],\ldots, T[i,\lambda_i]\}$ as multisets for each $1 \leq i \leq \ell$, (*)} \\
    0, & \text{otherwise.}
\end{cases}
\end{equation}

Finally, given a $T\in \mathrm{SSYT}(\lambda,\mu)$, define the function $\chi_T: \sgrid^n_\mu\rightarrow \mathbb{Z}$ by 
\begin{equation}
    \label{eq:defn-chiT}
    \chi_T(\bm{x}) = \sum_{\sigma \in C_\lambda} \mathrm{sgn}(\sigma)\cdot e_{T_0^\sigma, T}(\bm{x}),
\end{equation}
where $T_{0}^{\sigma}$ is the tableau obtained after $\sigma$ acts on the canonical tableau $T_{0}$.\\
\end{definition}

\begin{observation}
\label{obs:chiT-junta}
    Note that condition (*) in \Cref{eq:defn-eT} could equivalently have been stated in terms of rows $i\in \{2,\ldots, \ell\}$ as the condition for $i=1$ is implied by the others (since the input $\bm{x}$ is a point in $\sgrid^n_\mu$). Overall, this implies that $e_{T_0^\sigma,T}$ (and hence $\chi_T$) depends only on variables whose index appears in one of the first $\lambda_2$ columns of $T$. In particular, if $\lambda_2 \leq c$, this implies that $\chi_T$ is a $\gamma$-junta for $\gamma \leq c\ell \leq cs.$
\end{observation}

Next we define a total order on the set $\mathrm{SSYT}(\lambda,\mu)$.\\

\begin{definition}[Total order on SSYTs]
Let $\lambda \in \mathcal{P}(n)$ and $\mu = (n/s,\ldots,n/s)$. Given two distinct SSYTs $S,T\in \mathrm{SSYT}(\lambda,\mu),$ we say that $S < T$ if there exists $2 \leq i \leq \ell$ and $j \in [\lambda_{i}]$ such that the following holds:
\begin{enumerate}
    \item For every $k > i$ and for every $j' \in [\lambda_{k}]$, we have $S[k, j'] = T[k, j']$, i.e. the $k^{th}$ rows of $S$ and $T$ are equal.
    \item For every $\lambda_{i} \geq j' > j$ such that $S[i, j'] = T[i,j']$.
    \item Finally, $S[i,j] < T[i,j].$ 
\end{enumerate}
We leave it to the reader to check that this defines a total order on $\mathrm{SSYT}(\lambda,\mu).$\\
\end{definition}

We will need the following claim regarding the aforementioned ordering.\\

\begin{claim}
    \label{clm:SSYT-order}
    Assume that $S,T\in \mathrm{SSYT}(\lambda,\mu)$ and $\sigma\in C_\lambda$ are such that
    \begin{itemize}
        \item either $S < T$
        \item or $S = T$ and $\sigma$ is not the identity permutation.
    \end{itemize}
Then, for any $\sigma\in C_\lambda$, there exists an $i\in \{2,\ldots, \ell\}$ such that the multisets $\{S^{\sigma}[i,1],\ldots, S^{\sigma}[i,\lambda_i]\}$ and $\{T[i,1],\ldots, T[i,\lambda_i]\}$ are distinct.
\end{claim}

\begin{proof}
    Choose $i$ to be the largest number such that $\sigma$ moves the contents of some cell in the $i$th row of $S$, assuming $\sigma$ is not the identity; otherwise, set $i = 0$. Assuming that $i\neq 0,$ for each cell in the $i^{th}$ row moved by $\sigma$, we note that the contents of this row can only decrease, since the columns of $S$ are strictly increasing and $\sigma$ does not change the contents of any row $i' > i$. In particular, this implies the claim in the case that $S = T$. We therefore assume that $S\neq T$ and $S <  T$ for the rest of the proof.
         
    Let $i_0$ be the largest number such that the $i_{0}^{th}$ rows of $S$ and $T$ differ. Note that $i_0\in \{2,\ldots, \ell\}.$ Further, fix $j_0$ to be the rightmost cell on this row where $S$ and $T$ differ. Note that $S[i_0,j_0] < T[i_0,j_0].$

 We note that we are immediately done if $i < i_0$ since in this case 
    \[
    \{S[i_0,1],\ldots, S[i_0,\lambda_{i_0}]\} = \{S^{\sigma}[i_0,1],\ldots, S^{\sigma}[i_0,\lambda_{i_0}]\} \neq \{T[i_0,1],\ldots, T[i_0,\lambda_{i_0}]\}
    \]
    So we may assume that $i \geq i_0$, and in particular that $\sigma$ is not the identity.

    Now, we consider two cases.
    \begin{itemize}
        \item If $i > i_0,$ then we have 
        \[
            \sum_{j\in [\lambda_i]}S^{\sigma}[i,j] <  \sum_{j\in [\lambda_i]}S[i,j] = \sum_{j\in [\lambda_i]} T[i,j]
        \]
        implying the claim in this case.
        
        \item If $i = i_0,$  consider the rightmost cell (numbered $j$, say) where $S^\sigma$ and $T$ differ on this row. Note that $S^{\sigma}[i,j_0] \leq S[i,j_0] < T[i,j_0]$ and hence $j \geq j_0.$

        Consider the multiplicities of the element $t := T[i,j]$ in the $i$th rows of $S, S^\sigma$ and $T$, which we denote $m_S, m_{S^\sigma}$ and $m_T$ respectively. Note that $m_S \leq m_T$ because $j\geq j_0$ and $S < T.$ We also know that $S^{\sigma}[i,j]\neq t$ by definition of $j$. Finally, note that for any $j' < j$ we have either $S^{\sigma}[i,j'] = S[i,j']$ or $S^{\sigma}[i,j'] < S[i,j'] \leq S[i,j]\leq T[i,j]$ with the latter two inequalities following from the fact that $S$ is an SSYT and the fact that $j\geq j_0.$ This implies that $m_{S^\sigma} < m_T$, the multisets defined by the $i$th row in the two tableaux $S^\sigma$ and $T$ cannot be equal.
    \end{itemize}
    This finishes the proof of the claim.
\end{proof}

The main result of this subsection is the following lemma, which shows the existence of the special vectors as stated in \Cref{thm:instantiation-rep-thy-sym}.\\

\begin{lemma}
    \label{lem:special-vecs}
    Let $\lambda \in \mathcal{P}(n)$ and $c \in \mathbb{N}$. Assume that $\lambda_2 \leq c$. For every $T \in \mathrm{SSYT}(\lambda,\mu)$, the function $\chi_T$ satisfy the following properties:
    \begin{enumerate}
        \item For each $T \in \mathrm{SSYT}(\lambda,\mu)$, we have $\lVert \chi_T\rVert_2 = \bigO_{s,c}(1).$
        \item Let $\varepsilon > 0$ be arbitrary and assume $k$ is an integer such that $k\geq cs$. For any $\varepsilon$-almost $k$-wise independent distribution $\mathcal{D}$ supported on $\sgrid^n_\mu,$ we have
        \begin{equation}
            \label{eq:special-vecs}
            \left| \mathbb{E}_{\mathbf{x}\sim \mathcal{D}}[\chi_T(\mathbf{x})]\right| \leq \bigO_{s,c}(\varepsilon).
        \end{equation}
        \item We have,
        \begin{align*}
         \mathrm{Vol}(\setcond{\chi_{T}}{T \in \mathrm{SSYT}(\lambda,\mu)}) \; = \; \Omega_{s,c}(1).
        \end{align*}
    \end{enumerate}
\end{lemma}

\begin{proof}
    The first item follows almost immediately from the definition of $\chi_T$ in \Cref{eq:defn-chiT} above. From this definition, we get
    \begin{equation}
    \label{eq:bound-chiT}
    \lVert \chi_T\rVert_2 \leq \lVert \chi_T\rVert_\infty = \max_{\mathbf{x}\in \sgrid^n_\mu} |\chi_T(\bm{x})| \leq  |C_\lambda|\cdot \max_{\sigma,\mathbf{x}} |e_{T_0^\sigma, T}(\mathbf{x})| \leq |C_\lambda| \leq (s!)^c
    \end{equation}
    where the first inequality is trivial, the second is the triangle inequality applied to \Cref{eq:defn-chiT}, the third follows from the fact $|e_{T_0^\sigma, T}(\bm{x})|\leq 1$ for each $\mathbf{x}$, and the last follows from the fact that $\lambda$ is $c$-good.

    For the second item, we note that by \Cref{obs:chiT-junta} and the $\varepsilon$-almost $k$-wise independence of $\mathcal{D}$, we have
    \[
    \mathbb{E}_{\mathbf{x} \sim \mathcal{D}}[\chi_{T}(\mathbf{x})] \quad  \leq \quad \varepsilon  \cdot \max_{\mathbf{x} \in \sgrid^{n}} |\chi_{T}(\mathbf{x})|  \; + \; \mathbb{E}_{\mathbf{a} \sim \sgrid^{n}_\mu}[\chi_{T}(\mathbf{a})].
    \leq \bigO_{s,c}(\varepsilon) + \mathbb{E}_{\mathbf{x} \sim \sgrid^{n}_\mu}[\chi_{T}(\mathbf{x})]
    \]
    where the second inequality uses the bound on $|\chi_T(\mathbf{x})|$ proved above. To bound the latter term, we note that for by symmetry, for each $\sigma\in C_\lambda$, the quantity $\mathbb{E}_{\mathbf{x} \sim \sgrid^{n}_\mu}[e_{T_0^\sigma, T}(\mathbf{x})]$ is exactly the same. Since the signed sum defining $\chi_T$ has the same number of positive and negative signs, the sum of the expectations is $0.$ This proves the second item of the claim.

    The third item needs a definition. Given a $S\in \mathrm{SSYT}(\lambda, \mu),$ define subset $A_{S} \subset \sgrid^{n}_{\mu}$ as follows:
    \begin{align*}
        A_{S} \; := \; \setcond{\mathbf{x} \in \sgrid^{n}_{\mu}}{x_{T_0[i,j]} = S[i,j], \; i \in [\ell], \; j\leq \min \{ \lambda_2,\lambda_i \}},
    \end{align*}
  i.e. we define $A_S$ using the first $\lambda_2$ columns of $S$.

    Note that for each $\mathbf{x}\in A_S$, we have the following:
    \begin{itemize}
        \item $e_{T_0,S}(\mathbf{x}) = 1$. This follows immediately from the definition of $e_{T_0,S}$ above.
        \item Now fix $T\in \mathrm{SSTY}(\lambda,\mu)$ and $\sigma\in C_\lambda$ such that either $T > S$ or $S=T$ and $\sigma$ is not the identity permutation (here the identity permutation in $C_{\lambda}$ refers to $\mathrm{id} \times \cdots \times \mathrm{id}$). We claim that $e_{T_0^\sigma, T}(\mathbf{x}) = 0.$ To see this, start by labelling each cell of $T_0$ with the value of the corresponding variable, which leads to a tableau $S'$ which agrees with $S$ on all cells in the first $\lambda_2$ columns. Since $e_{T_0^\sigma,T}$ depends only the variables in these columns, we may change $\mathbf{x}$ in the other coordinates to ensure that $S' = S.$ 
        
        Now, we observe that for any $\sigma\in C_\lambda$, the multiset $\{x_{T_0^\sigma[i,1]},\ldots, x_{T_0^\sigma[i,\lambda_i]}\}$ is equal to the multiset $\{S^{\sigma}[i,1],\ldots, S^\sigma[i,\lambda_i]\}$. In particular, by \Cref{clm:SSYT-order}, there exists an $i\in [\ell]$ so that the  multiset $\{x_{T_0^\sigma[i,1]},\ldots, x_{T_0^\sigma[i,\lambda_i]}\}$ is not equal to $\{T[i,1],\ldots, T[i,\lambda_i]\}$, implying that $e_{T_0^\sigma,T}(\mathbf{x}) = 0.$
    \end{itemize}
    The above implies that for each $\mathbf{x}\in A_S$, we have
    \begin{itemize}
        \item $\chi_S(\mathbf{x}) = 1$ and
        \item $\chi_T(\mathbf{x}) = 0$ for each $T > S$.
    \end{itemize}

\noindent
For each $S\in \mathrm{SSYT}(\lambda,\mu)$, let $\tilde{\chi}_S$ denote the projection of $\chi_S$ to the vector space orthogonal to the span of $\{\chi_T\ |\ S < T\}$.\newline
    To bound $\lVert\tilde{\chi}_S\rVert_2$, we recall that $\tilde{\chi}_S = \chi_S - \chi$ for some $\chi$ in the span of $\{\chi_T\ |\ S < T\}.$ By the above argument, we know that $\chi(\mathbf{x}) = 0$ and hence that $\tilde{\chi}_S(\bm{x}) = 1$ for each $\mathbf{x}\in A_S.$ Hence, we get
    \[
    \lVert\tilde{\chi}_S\rVert_2^2 = \mathbb{E}_{\mathbf{x}\sim \sgrid^n_\mu} \tilde{\chi}_S(\mathbf{x})^2 \geq \frac{|A_S|}{N}
    \]
    where $N = |\sgrid^n_\mu|.$ So to prove the claim, it suffices to show that the latter quantity is $\Omega_{s,c}(1).$

    For each $i \in \{0,\ldots,s-1\}$, let $\gamma_{i}$ denote the number of cells in the first $\lambda_2$ columns of $S$ that are $i$

    Define $\gamma := \gamma_{0}+\ldots+\gamma_{s-1} \leq cs$. Using Stirling's approximation and $\gamma \leq cs \leq n/2$, we get,
\begin{align*}
    |A_{S}| \; = \; \binom{n-\gamma}{\frac{n}{s}-\gamma_{0},\ldots,\frac{n}{s}-\gamma_{s-1}} \; \geq \; \dfrac{(n-\gamma)^{n-\gamma} \cdot s^{n-\gamma}}{(n-s)^{n-\gamma}} \cdot \Omega\paren{ \dfrac{\sqrt{\pi n}}{(2 \pi (\frac{n}{s}))^{s/2}} }
\end{align*}
Again using Stirling's approximation for $N = \binom{n}{n/s,\ldots, n/s}$, we get,
\begin{gather*}
    \dfrac{|A_{S}|}{|\sgrid^{n}_{\mu}|} \; \geq \; \dfrac{(n-\gamma)^{n-\gamma}}{(n-s)^{n-\gamma}} \, \cdot \, \dfrac{1}{s^{\gamma}} \cdot \Omega \paren{\dfrac{\sqrt{\pi n}}{(2\pi(\frac{n}{s}))^{s/2}} \cdot \dfrac{(2\pi(\frac{n}{s}))^{s/2}}{\sqrt{2 \pi n}} } \\ \\
    \geq \; \paren{1 - \dfrac{\gamma-s}{n-s}}^{n-\gamma} \cdot \dfrac{1}{s^{\gamma}} \cdot  \Omega(1) \\ \\
    \geq \; \Omega\paren{ \paren{ 1-\dfrac{2\gamma}{n}}^{n/2\gamma \cdot 2\gamma} \cdot \dfrac{1}{s^{\gamma}} } \; \geq \; \Omega\paren{ \paren{\dfrac{1}{e^{2}s}}^{\gamma} } \; = \; \Omega_{s,\gamma}(1).
\end{gather*}
As $\gamma \leq cs$, we get that $|A_{S}|/N = \Omega_{s,c}(1)$.\\

\noindent
The volume of the parallelepiped is equal to the product of $\| \tilde{\chi}_{T} \|'s$. As we showed above, $\widetilde{\chi}_{T} = \Omega_{s,c}(1)$, and thus we get,
\begin{align*}
    \mathrm{Vol}(\setcond{\chi_{T}}{T \in \mathrm{SSYT}(\lambda,\mu)}) \; \geq \; \prod_{T \in \mathrm{SSYT}(\lambda,\mu)} \, \| \tilde{\chi}_{T} \|_{2} \; = \; \Omega_{s,c}(1).
\end{align*}
This finishes the proof of \Cref{lem:special-vecs}.
\end{proof}

\subsection{Putting Everything Together}\label{subsec:proof-instantiation}
Now we are ready to combine everything and prove \Cref{thm:instantiation-rep-thy-sym}. To do so, we will use the following standard result on the representation of $\mathrm{Sym}_{n}$. The proof can be found in standard texts on representation theory for the symmetric group or \cite{Sagan}.\\

\begin{theorem}[Young's Rule]\label{thm:irrep-decomposition}
(See for e.g. \cite[Corollary 2.11.2]{Sagan}). Fix any $n,s \in \mathbb{N}$ where $n$ is divisible by $s$ and let $\mu = (n/s,\ldots,n/s) \in \mathcal{P}(n)$. 
For every $\lambda \in \mathcal{P}(n)$ with $\lambda \trianglerighteq \mu$, let $V_{\lambda,j}$ and $m_{\lambda}$ be as defined in \Cref{thm:instantiation-rep-thy-sym}. Then,
\begin{align*}
    \dim(V_{\lambda,1}) = \cdots = \dim(V_{\lambda,m_{\lambda}}) = f_{\lambda} \quad \text{ and } \quad m_{\lambda} = K_{\lambda \mu},
\end{align*}
where $f_{\lambda}$ and $K_{\lambda \mu}$ are defined in \Cref{sec:prelims}.
\end{theorem}

\noindent
Next we prove two claims regarding $f_{\lambda}$ and $m_{\lambda}$ for certain partitions $\lambda \in \mathcal{P}(n)$. These two claims will be used to prove the item $2$ of \Cref{thm:instantiation-rep-thy-sym}.\\

\begin{restatable}[Lower bound on the algebraic multiplicity of certain eigenvalues]{claim}{dimensionlowerbond}\label{claim:large-alg-mult}
\cite[Lemma 2]{Ellis-Friedgut-Pilpel-2011}\footnote{There is a minor typo in the statement of Lemma 2 in \cite{Ellis-Friedgut-Pilpel-2011}. It should be ``of length \textbf{at most}...'' instead of ``of length greater than...''}. Let $c \in \mathbb{N}$ be a constant with $c > 10s$. Then for any partition $\lambda \in \mathcal{P}(n)$ with $\lambda_{2} > c$,
\begin{align*}
    f_\lambda > \Omega_{c}(n^{c}).
\end{align*}
\end{restatable}

\noindent
\begin{claim}[Multiplicity for $c$-good partitions]\label{claim:multiplicity-c-good}
Let $c \in \mathbb{N}$ be a constant and $\lambda \in \mathcal{P}(n)$ with $\lambda_{2} \leq c$ and $\lambda \neq (n)$. Let $m_{\lambda}$ be as defined in the statement of \Cref{thm:instantiation-rep-thy-sym}. Then, $m_{\lambda} \leq s^{sc} = \bigO_{s,c}(1)$.
\end{claim}
\begin{proof}[Proof of \Cref{claim:multiplicity-c-good}]
From \Cref{thm:irrep-decomposition}, we know that $m_{\lambda} = K_{\lambda\mu}$, i.e. the Kostka numbers for shape $\lambda$ and type $\mu$.
We are interested in upper bounding $K_{\lambda\mu}$ for a $c$-good partition $\lambda$. We have $\lambda_{2} + \ldots + \lambda_{\ell} \leq cs$. For each cell in the second row till the last row, there are at most $s$ many choices. As there are $\leq cs$ such cells, we get that $K_{\lambda \mu} \leq s^{cs}$. This finishes the proof of \Cref{claim:multiplicity-c-good}.
\end{proof}

\noindent
Now we are ready to put all the claims and lemmas together to finish the proof of \Cref{thm:instantiation-rep-thy-sym}.\\

\noindent
\begin{proof}[Proof of \Cref{thm:instantiation-rep-thy-sym}]
The first item follows by combining the first item of \Cref{prop:standard-rep-theory} and \Cref{thm:irrep-decomposition}. Item $2$.(a) follows from \Cref{claim:large-alg-mult} and item $2$.(b) follows from \Cref{claim:multiplicity-c-good}.\newline
Finally, we show that the vectors $\chi_{T}$'s meet the conditions stated in the third item.
\begin{enumerate}
    \item For 3.(a), we note that the literature on the representation theory of $\mathrm{Sym}_n$ (see \cite[Section 2.9 \& Section 2.10]{Sagan}\footnote{In the literature of representation theory, these isomorphisms are stated in the language of \textbf{tabloids} and \textbf{polytabloids}. In  \Cref{app:tabloids}, we provide a translation between the language of tabloids/polytabloids and points/functions.}) identifies for each $\lambda\in \mathcal{P}(n)$ exactly $m_\lambda$ many linearly independent ways of embedding the irreducible representation $\mathbb{S}^\lambda$ ($\mathbb{S}^{\lambda}$ is the unique irreducible representation, or Specht module, corresponding to partition $\lambda$) into the representation $\mathbb{M}^{\mu}.$ These embeddings are indexed by elements of $\mathrm{SSYT}(\lambda,\mu)$ and denoted by $\Theta_{T}: \mathbb{S}^{\lambda} \to \mathbb{M}^{\mu}$. Given $T\in \mathrm{SSYT}(\lambda,\mu)$, let $V_{\lambda,T}$ denote the image of $\mathbb{S}^{\lambda}$ under the corresponding embedding $\Theta_T.$
    
    It can be checked that the various $\chi_T$ are the images of the same element $v\in \mathbb{S}^{\lambda}$ under $\Theta_T$ (see also \cite{complexity-symmetric-group}). This implies that for $S,T\in \mathrm{SSYT}(\lambda,\mu)$ $\chi_T$ is the image of $\chi_S$ under the unique isomorphism from $V_{\lambda,S}$ to $V_{\lambda,T}.$ This proves 3.(a). 
    \item Item $1$ of \Cref{lem:special-vecs} shows that they satisfy 3.(b).
    \item Item $2$ of \Cref{lem:special-vecs} shows that they satisfy 3.(c).
    \item Item $3$ of \Cref{lem:special-vecs} shows that they satisfy 3.(d).
\end{enumerate}
This finishes the proof of \Cref{thm:instantiation-rep-thy-sym}.
\end{proof}

\subsection{Singular Value Bound for Nearly Balanced Random Walks} \label{subsec:eigbound}

We now use the statement of~\Cref{thm:more-general-matrix-eigenvalue} to derive the singular value bounds for nearly balanced random walks on the multislice, as stated in~\Cref{cor:gen-bal-cor}.

For this, we will need the following lemma.\\

\begin{lemma}\label{cor:johnson-spl-case}
    For every $s\ge 2$ and $C<\infty$, there exists $\tau > 0$ such that for every finite set $\mathcal{S}$ of size $s$ and sufficiently large $n\in\N$, if a generalized Hamming distance parameter $\Delta\in \Z^{\mathcal{S} \times\mathcal{S}}$ over the multislice $\mathcal{S}^n_\mu$ is $C$-balanced, we have that $\sigma_2(W_\Delta) \le 1/n^{\tau}$, where $W_\Delta$ is the random walk matrix determined by $\Delta$.\\
\end{lemma}

\noindent
The above statement implies the claimed general result (i.e., \Cref{cor:gen-bal-cor}) for random walk matrices that are not necessarily given by a single generalized Hamming distance parameter, but as long as they are {\em supported} on balanced generalized Hamming distance parameters.

We first prove~\Cref{cor:johnson-spl-case}.

\begin{proof}[Proof of~\Cref{cor:johnson-spl-case}]
    Here we directly apply our main result~\Cref{thm:more-general-matrix-eigenvalue}, bounding the singular values of matrices satisfying certain properties. For this, we show that $W_{\Delta}$ satisfies the three properties needed to apply~\Cref{thm:more-general-matrix-eigenvalue}. 
    \begin{itemize}
        \item {\bf Permutation invariance:} For every permutation $\pi$ of $[n]$, $W_{\Delta}$ is unchanged if the rows and columns are changed according to the permutation induced by $\pi$ (denoted $\pi()$) on the balanced multislice (denoted $V$ in this proof). This is because the value of the entry $W_{\Delta}({\bf a},{\bf b})$ only depends on $\Delta({\bf a},{\bf b})$, which doesn't get altered by $\pi$, i.e., we have $\Delta({\bf a},{\bf b}) = \Delta(\pi({\bf a}),\pi({\bf b}))$.
        \item {\bf Bounded Frobenius norm:} We will show that $\|W_{\Delta}\|_F \le n^{O_s(1)}$. Denoting $m:=n/s$ and the rows of $\Delta$ by ${\bf p}(0),\dots,{\bf p}(s-1)$, we note that for each ${\bf a}\in V$, there are exactly $D:={m\choose {\bf p}(0)}\dots {m \choose {\bf p}(s-1)}$ \footnote{Here, for a vector of integers ${\bf p}=(p_1,\dots,p_s)$, ${m\choose {\bf p}}$ denotes ${m\choose p_1,\dots,p_s}$.}  points ${\bf b}\in V$ such that $\Delta({\bf a},{\bf b}) = \Delta$. Hence, we have that
        \begin{align*}
            W_{\Delta}({\bf a},{\bf b}) \; = \; \begin{cases}
             1/D, & \text{if } \; \Delta({\bf a},{\bf b}) = P \\
             0, & \text{otherwise}. 
        \end{cases}
        \end{align*}
        Therefore, we have 
        \begin{align}
            \|W_{\Delta}\|_F^2 & = \sum_{{\bf a},{\bf b}\in V} W_{\Delta}({\bf a},{\bf b})^2 \nonumber \\
            & = {|V|D}/{D^2} \nonumber \\
            & = {{sm \choose m,\dots, m}}\bigg /\paren{{m\choose {\bf p}(0)} \dots {m\choose {\bf p}(s-1)}} \label{eqn:frob}
        \end{align}

\paragraph{} In order to bound the above quantity, let ${\bf q} = (q_0,\dots,q_{s-1})\in \Z^{\Z_s}$ be such that $\sum_{j\in \Z_s} q_j = m$ and $|q_j-q_{j'}| \le 1$ for all $j,j'\in \Z_s$ (such a ${\bf q}$ always exists; indeed each $\lfloor m/s \rfloor \le q_j \le \lceil m/s \rceil$). We will first show that $\frac{p_0!\dots p_{s-1}!}{q_0!\dots q_{s-1}!}$ is upper bounded by $m^{\bigO_s(1)}$, where ${\bf p} := {\bf p}(\alpha) = (p_0,\dots,p_{s-1})$ for an arbitrary $\alpha\in \Z_s$.\newline

\begin{claim}
    $\frac{p_0!\dots p_{s-1}!}{q_0!\dots q_{s-1}!} \le m^{\bigO_s(1)}$.
\end{claim}

\begin{proof}

We consider the following sequence of vectors: $ {\bf p} = {\bf p}^{(0)},  {\bf p}^{(1)}, \dots ,  {\bf p}^{(t)} =  {\bf q}$ (for some finite $t$), where two adjacent ${\bf p}^{(i-1)}$ and ${\bf p}^{(i)}$ differ in exactly two coordinates (say $c_i\ne c_i'\in \Z_s$) such that $p_{c_i}^{(i)} = p_{c_i}^{(i-1)}+1$ and $p_{c_i'}^{(i)} = p_{c_i'}^{(i-1)}-1$, for all $i\in [t]$. We note that since each $p_j \in \frac{m}{s} \pm \sqrt{Cm\log m}$ (as $P$ is a {\em balanced} generalized Hamming distance matrix), such a sequence can be realized with $t\le s \sqrt{Cm\log m}$ by repeatedly picking the smallest and largest elements of ${\bf p}$ and adding one to the smallest element and subtracting one from the largest one. This will also ensure that for each intermediate $i\in [t]$, we have the invariant $p_j^{(i)} \in \frac{m}{s}\pm \sqrt{Cm\log m}$ for all $j\in \Z_s$.

Now, we note that for all $i\in [t]$,
\begin{align*}
    \dfrac{\prod_j p^{(i-1)}_j!}{\prod_j p^{(i)}_j!} \; = \; \dfrac{p^{(i-1)}_{c_i}! p^{(i-1)}_{c_i'}!}{p^{(i)}_{c_i}! p^{(i)}_{c_i'}!} \; = \; \dfrac{p^{(i-1)}_{c_i'}}{p^{(i)}_{c_i}} \; \leq \;  1+\bigO\paren{\sqrt{\frac{Cs^2\log m}{m}}}.
\end{align*}
Using the above bound for all $i\in [t]$ and multiplying them, we get
\begin{align}\label{eqn:fg}
\dfrac{p_0!\dots p_{s-1}!}{q_0!\dots q_{s-1}!} \; \le \; \paren{1+ \bigO\paren{\sqrt{\frac{Cs^2\log m}{m}}}}^t \; \le \; \paren{1+\bigO\paren{\sqrt{\dfrac{Cs^2\log m}{m}}}}^{s\sqrt{Cm \log m}} \le m^{\bigO(Cs^2)},
\end{align} 
where for the last inequality, we are using the inequality $1+x \le e^x$. 

\end{proof}

Now, continuing the computation of~\eqref{eqn:frob}, we have
        \begin{align*}
            \|W_{\Delta}\|_F^2 & = {sm \choose m, \dots, m} \bigg / \paren{{m \choose {\bf p}(0)} \cdots {m \choose {\bf p}(s-1)}}\\ \\
            & \leq \; m^{O(Cs^3)}\cdot\frac{(sm)!q_0!^s\dots q_{s-1}!^s}{m!^{2s}} \tag{using~\Cref{eqn:fg}}\\ \\
            & \le \; m^{\bigO(Cs^3)}\cdot\frac{(sm)!}{m!^{s}}\tag{as $q_j \le \lceil m/s \rceil$} \cdot \paren{\frac{\lceil m/s \rceil!^s}{m!}}^s\\
            & \le \; m^{\bigO(Cs^3)}\cdot s^{sm} \cdot \paren{\frac{(m/(es))^m}{(m/e)^m}}^s \tag{using Stirling's inequality}\\ \\
            & \le m^{\bigO(Cs^3)}.
        \end{align*}
         
        Hence $\|W_{\Delta}\|_F \le n^{\bigO_{s,C}(1)}$.
        \item {\bf $\varepsilon$-almost $k$-wise independence:} We will show that for every $k\le \bigO_s(1)$, $W_{\Delta}$ is $\varepsilon$-almost $k$-wise independent (see \Cref{defn:eps-close-k-wise-matrix}) for some $\varepsilon = 1/n^{\Omega_s(1)}$. That is, for every ${\bf a}\in V$ and $T\in {[n]\choose k}$, we will show that
        \begin{align*}
            \mathrm{SD}(W_{\Delta}({\bf a})|_T, U_T) \; \leq \; \varepsilon,
        \end{align*}
        where $U_T$ denotes the uniform distribution over the coordinates given by $T$.\\

        \noindent
        Let $T= T^{(0)} \cup \dots \cup T^{(s-1)}$ be a partition of $T$, where $T^{(i)} = {\bf a}^{-1}(i) \cap T$ for $i\in \Z_s$. We will fix an arbitrary ${\bf b}\in \Z_s^T$ and upper bound the difference $\abs{\Pr[W_{\Delta}({\bf a})|_T={\bf b}]-\frac{1}{s^k}}$. For this, let ${\bf e}=(e_j)_{j\in \Z_s}$ denote the number of occurrences of $j\in \Z_s$ in ${\bf b}$. Furthermore, let ${\bf e}^{(i)} = (e^{(i)}_j)_{j\in \Z_s}$ where $e^{(i)}_j$ denotes the number of occurrences of $j\in \Z_s$ in ${\bf b}$ when restricted to $T^{(i)}$. To make the notation cleaner, for the rest of the proof, we will use the notation ${\bf p}^{(i)}$ to mean ${\bf p}(i)$. We then have:
        \begin{align}\label{eqn:prob-bound}\Pr[W_{\Delta}({\bf a})|_T = {\bf b}] = {m-|T^{(0)}|\choose {\bf p}^{(0)}-{\bf e}^{(0)}}\dots {m-|T^{(s-1)}| \choose {\bf p}^{(s-1)}-{\bf e}^{(s-1)}} \bigg / \paren{{m\choose {\bf p}^{(0)}}\dots {m\choose {\bf p}^{(s-1)}}}.\end{align}
        For each $i\in \Z_s$, we have
        \begin{align*}
            \frac{{m-|T^{(i)}|\choose {\bf p}^{(i)}-{\bf e}^{(i)}}}{{m\choose {\bf p}^{(i)}}} & = \frac{(m-e^{(i)}_0-\dots-e^{(i)}_{s-1})!}{(p_0^{(i)}-e^{(i)}_0)!\dots(p_{s-1}^{(i)}-e^{(i)}_{s-1})!} \cdot \frac{p_0^{(i)}!\dots p_{s-1}^{(i)}!}{m!}\\ \\
            & = \frac{(p_0^{(i)}\dots (p_0^{(i)}-e^{(i)}_0+1))\dots(p_{s-1}^{(i)}\dots (p_{s-1}^{(i)}-e^{(i)}_{s-1}+1))}{m\dots (m-e^{i}_0-\dots e^{(i)}_{s-1}+1)}\\ \\
            & \in \bigg[\paren{\frac{\frac{m}{s} - \sqrt{Cm\log m} - |T^{(i)}|}{m}}^{|T^{(i)}|}, \paren{\frac{\frac{m}{s}+\sqrt{Cm\log m}}{m-|T^{(i)}|}}^{|T^{(i)}|} \bigg] \tag{as each $p_j^{(i)}\in \frac{m}{s}\pm \sqrt{Cm\log m}$}\\
            & \subseteq \frac{1}{s^{|T^{(i)}|}}\bigg[1\pm \frac{1}{m^{\Omega_{s,C}(1)}}\bigg]^{|T^{(i)}|} \tag{using $|T^{(i)}| \le k \le O_s(1)$}.
        \end{align*}
        
        Plugging the above bound into~\eqref{eqn:prob-bound} and using $\sum_i |T^{(i)}| = |T| = k$ gives that
        \begin{align*}
            \Pr[W_{\Delta}({\bf a})|_T={\bf b}] \; \in \; \prod_{i\in \Z_s} \dfrac{1}{s^{|T^{(i)}|}}\bigg[1\pm \frac{1}{m^{\Omega_s(1)}}\bigg]^{|T^{(i)}|} \; \subseteq \; \bigg[ \frac{1}{s^k} \pm \frac{k}{m^{\Omega_s(1)}} \bigg].
        \end{align*}

        Therefore, we have
        \begin{align*}
            \textrm{SD}(W_{\Delta}({\bf a})|_T, U_T) \; = \; \frac{1}{2}\sum_{{\bf b}\in \mathcal{S}^T} \bigg|\Pr[W_{\Delta}({\bf a})|_T={\bf b}]-\frac{1}{s^k}\bigg| \; \leq \; \frac{s^k \cdot k}{m^{\Omega_{s,C}(1)}} \; \leq \; \frac{1}{n^{\Omega_{s,C}(1)}}.
        \end{align*}

    \end{itemize}
    Thus we have shown that the three conditions needed to apply~\Cref{thm:more-general-matrix-eigenvalue} hold for $W_{\Delta}$. Therefore, we obtain $\sigma_2(W_{\Delta}) \le 1/n^{\Omega_{s,C}(1)}$, finishing the proof of~\Cref{cor:johnson-spl-case}.
\end{proof}

We now finish the proof of~\Cref{cor:gen-bal-cor} using~\Cref{cor:johnson-spl-case}.\\

\genbalcor*

\begin{proof}[Proof of~\Cref{cor:gen-bal-cor}]
    The idea is to express $W$ as a convex combination of $W_\Delta$ for $\Delta$ being $C$-balanced generalized Hamming distance parameters.  We first show that for every ${\bf a}\in \sgrid^n_\mu$, the ${\bf a}$-th row of $W$ can be expressed as a convex combination of the ${\bf a}$-th rows of the random walk matrix determined by the individual generalized Hamming distance parameters (i.e, $W_{\Delta}$). As $W$ respects symmetries, for every ${{\bf a},{\bf b}}\in \sgrid^n_\mu$ and permutation $\pi \in \Sym_n$, we have that $W({\bf a},{\bf b}) = W({\pi({\bf a})},{\pi({\bf b})})$. Now we note that if it holds that $\Delta({\bf a},{\bf b}) = \Delta({\bf a},{\bf c})$ for some ${\bf a},{\bf b},{\bf c}$, then there exists a permutation $\pi \in \Sym_n$ such that $\pi({\bf a})={\bf a}$ and $\pi({\bf b}) = {\bf c}$ (this can be obtained by permuting the coordinates of ${\bf a}$ that take identical values); hence we have $W({\bf a},{\bf b}) = W(\pi({\bf a}),\pi({\bf b})) = W({\bf a},{\bf c})$. Since $W$ has positive entries only at cells corresponding to balanced generalized Hamming distance, we can thus express the ${\bf a}$-th row of $W$ as the following convex combination:
    \begin{align}\label{eqn:convcomb}
    W({\bf a}) = \sum_{{\Delta}~\text{is~} C\text{-balanced}} \alpha_{{\bf a},\Delta} W_\Delta({\bf a}),\end{align} for some $\alpha_{{\bf a},\Delta} \ge 0$ such that $\sum_{\Delta\text{~is~}C\text{-balanced}}{\alpha_{{\bf a},\Delta}} = 1$. We now show that $\alpha_{{\bf a},\Delta}=\alpha_{{\bf b},\Delta}$ for every ${\bf a},{\bf b}\in \sgrid^n_\mu$ and generalized Hamming distance parameter $\Delta$. Let $\pi\in \Sym_n$ be a permutation such that $\pi({\bf a}) = {\bf b}$ and let ${\bf c}\in \sgrid^n_\mu$ be an arbitrary point such that $\Delta({\bf a},{\bf c})=\Delta$. Further, let $t_\Delta$ denote the number of points ${\bf d}\in \sgrid^n_\mu$ such that $\Delta({\bf a},{\bf d})=\Delta$ (note that this does not depend on ${\bf a}$). Then using~\eqref{eqn:convcomb} we have the following.

    \begin{align}\label{eqn:conv1}
        W({\bf a},{\bf c}) = \alpha_{{\bf a},\Delta} \cdot \frac{1}{t_\Delta}
    \end{align}

    Since $W({\bf a},{\bf c}) = W(\pi({\bf a}),\pi({\bf c})) = W({\bf b},\pi({\bf c}))$  and $\delta({\bf b},\pi({\bf c})) = \Delta$, using~\eqref{eqn:convcomb} again, we have:

    \begin{align}\label{eqn:conv2}
        W({\bf b},\pi({\bf c})) = \alpha_{{\bf b},\Delta} \cdot \frac{1}{t_\Delta}
    \end{align}

    From~\eqref{eqn:conv1} and~\eqref{eqn:conv2}, we get that $\alpha_{{\bf a},\Delta} = \alpha_{{\bf b},\Delta}=\Delta$; hence we can simply denote $\alpha_{{\bf a},\Delta}$ by $\alpha_\Delta$. Now, using~\eqref{eqn:convcomb} for all the rows ${\bf a}\in \sgrid^n_\mu$ of $W$, we obtain 
    $$W = \sum_{{\Delta}\text{~is~}C\text{-balanced}} \alpha_\Delta W_\Delta,$$ where $\alpha_\Delta \ge 0$ and $\sum_{\Delta\text{~is~}C\text{-balanced}} \alpha_{\Delta}=1$.

Since $W_\Delta^\top = W_{\Delta^\top}$ is stochastic, we note that $W_\Delta$ is doubly stochastic. Thus, by applying~\Cref{lem:convex} we conclude that $\lambda_2(W') \le \max_{{\Delta}\text{~is~}C\text{-balanced}}\{\lambda_2(W'_{\Delta})\} \le 1/n^{\tau}$, where $\tau > 0$ is a constant given by~\Cref{cor:johnson-spl-case}.  This finishes the proof of~\Cref{cor:gen-bal-cor}.

\end{proof}

\section{Near-Optimal Distance Lemmas Over Balanced Multislices}\label{sec:sz}

In this section, we derive near-optimal polynomial lemmas for junta-sums and polynomials over the balanced multislice. More formally, for a finite set $\mathcal{S}$ of size $s\ge 2$, integer $d\ge 0$, positive integer $n$ divisible by $s$ and $\mu=(n/s,\dots, n/s)$ (repeated $s$ times), we recall that $\mathcal{S}^n_\mu \subseteq \mathcal{S}^n$ denotes the set of points in which each element $i\in \mathcal{S}$ appears $n/s$ many times. 

We also recall that $\J_d(\Z_s^n, G)$ denotes the family of $d$-junta sums from the domain $\Z_s^n$ to an Abelian group $G$. Similarly, we let $\mathcal{P}_d(\F_q^n)$ denote the family of polynomials of degree at most $d$ over a finite field $\F_q$. The well-known ODLSZ lemma states that $\mathcal{P}_d(\F_q^n)$ forms a code of relative distance $\delta = \delta(q,d)$ independent of $n$. Stated more formally,\\

\begin{lemma}[Polynomial distance lemma (ODLSZ lemma)]\label{thm:odlsz}
(See e.g. \cite[Lemma 9.4.1]{GRS-codingbook}). For every finite field $\F=\F_q$, if a polynomial $P\in \mathcal{P}_d(\F^n)$ is such that $P({\bf a})\ne 0$ for some ${\bf a}\in \F^n$,  then $$\Pr_{{\bf b}\sim {\F^n}} [P({\bf b})\ne 0] \ge \delta(q,d),$$ where $\delta(q,d) = (1-\beta/q)q^{-\alpha}$, where $\alpha$ and $\beta$ are the quotient and remainder respectively when $d$ is divided by $q-1$.
\end{lemma}

With this setup, we prove the following two main theorems in this section. \\

\begin{thmbox}
\begin{theorem}[{\bf Distance of junta-sums over multislice}]\label{thm:dist-junta-sums}
    If a junta-sum $P \in \mathcal{J}_{d}(\sgrid^{n},G)$ is such that $P({\bf a}) \ne 0$ for some ${\bf a} \in \mathcal{S}^n_\mu$, then $$\Pr_{{\bf b}\sim \mathcal{S}^n}[P({\bf b})\ne 0] \ge \frac{1}{s^d}-\frac{1}{n^{\Omega_{s}(1)}}.$$
\end{theorem}
\end{thmbox}

As noted in \Cref{sec:intro} we also prove a similar theorem for algebraic degree as opposed to junta-degree. We recall the theorem statement below.\\

\begin{thmbox}
\algebraicszlemma*
\end{thmbox}

We first prove~\Cref{thm:dist-junta-sums} below followed by the proof of~\Cref{thm:dist-polys-intro}, which is almost identical.

\begin{proof}[Proof of~\Cref{thm:dist-junta-sums}]
    Without loss of generality, we will assume that $\mathcal{S} = \Z_s$, so addition and subtraction of elements of $\mathcal{S}$ make sense. At a high level, the proof proceeds as follows. We consider a random walk matrix $W$ over the multislice which we will describe below, and bound its eigenvalues. We will then use the ``expander mixing lemma'' to derive the required distance lower bound and finish the proof. 

    We shall define a random walk matrix $W_\ODLSZ$ over the points in the multislice, i.e., 
    $V=\mathcal{S}^n_\mu$ and we let $N=|V|={sm \choose {m,\dots,m}}$ where $m=n/s$ is an integer. For each ${\bf a}\in V$, we define the distribution over the neighbors of ${\bf a}$ according to $W_\ODLSZ$ (or equivalently, the ${\bf a}$-th row of $W_\ODLSZ$, denoted $W_\ODLSZ({\bf a})$) as being the random variable output by the  algorithm below (\Cref{alg:ab}).

\begin{algobox}
\begin{algorithm}[H]
\caption{The random walk matrix $W_\ODLSZ$}
\label{alg:ab}

\DontPrintSemicolon

\KwIn{${\bf a}\in V$}
\vspace{3mm}

For $j\in \Z_s$, letting ${\bf a}^{-1}(j) \in {[n]\choose m}$ denote the coordinates of ${\bf a}$ with value $j$, sample uniformly random bijections $M_j:{\bf a}^{-1}(j) \to [m]$ independently for all $j\in \Z_s$. \;
Sample ${\bf y}=(y_1,\dots,y_m) \sim \Z_s^m$ u.a.r. \;
Define ${\bf b}=(b_1,\dots,b_n)$ as follows: For $i\in [n]$, we let $j:=a_i$ and $b_i:=y_{M_j(i)}+j$.\;
\Return{{\bf b}}
\end{algorithm}
\end{algobox}

We first note that ${\bf b}$ is always on the balanced multislice, i.e., ${\bf b}\in V$, so $W_\ODLSZ$ is a well-defined random walk matrix over the balanced multislice. We now argue that for every fixed ${\bf a} \in V$ such that $P({\bf a}) \ne 0$ for a junta-sum $P\in \J_d(\mathcal{S}^n, G)$, it holds that 
\begin{align*}\label{eqn:random-step}
    \Pr_{{\bf b}\sim W_\ODLSZ({\bf a})}[P({\bf b})\ne 0] \ge 1/s^d.
\end{align*}

To see this, we fix the bijections $M_j:a^{-1}(j)\to [m]$ in Step 1 of~\Cref{alg:ab} arbitrarily and get the probability bound over the uniformly random choice of ${\bf y}$ in Step 2. More precisely, letting $$Q({\bf y}) = P({\bf b}) = P((y_{M_{a_i}(i)}+a_i)_{i\in [n]}),$$ we note that $Q:\mathcal{S}^m \to G$ is a $d$-junta-sum since $P$ is a $d$-junta-sum. Moreover, $Q({\bf 0}) = P({\bf a}) \ne 0$. Therefore, by applying~\Cref{clm:dist-junta-sums}, we get that $\Pr_{{\bf y}\sim \mathcal{S}^m}[Q({\bf y})\ne 0] \ge 1/s^d, $ and thus $\Pr_{{\bf b}\sim W_\ODLSZ({\bf a})}[P({\bf b})\ne 0] \ge 1/s^d.$

Letting $U\subseteq V$ denote the set of points in $V$ which evaluate $P$ to a non-zero value, from the above discussion, we have that 
\begin{equation}\label{eqn:forlla-b}
\forall {\bf a}\in U, ~\Pr_{{\bf b}\sim W_\ODLSZ({\bf a})}[{\bf b} \in U] \ge 1/s^d.
\end{equation}

We now use the expander mixing lemma. 

\begin{theorem}[{\bf Expander mixing lemma} see e.g.~\cite{HLW} Lemma 2.5]\label{thm:eml}
    For every symmetric random walk matrix $W\in \R^{V\times V}$ over a finite vertex set $V$ and $U \subset V$, 
    $$\Pr_{\substack{{\bf a}\sim V\\ {\bf b}\sim  W({\bf u})}}[{\bf a}\in U\text{~and~}{\bf b}\in U] \le \paren{\frac{|U|}{|V|}}^2 + \lambda_2(W)\paren{\frac{|U|}{|V|}},$$ where $\lambda_2(W)$ denotes the second largest eigenvalue of $W$ in absolute value. 
\end{theorem}

In order to apply the above theorem, we will need to show that the random walk matrix $W_\ODLSZ$ we defined is symmetric and has a small $\lambda_2(W_\ODLSZ)$. 

\begin{lemma}\label{lem:eig-val-sz}

 The random walk matrix $W_\ODLSZ$ as defined in \Cref{alg:ab} is symmetric and satisfies $\lambda_2(W_\ODLSZ) \le 1/n^{\Omega_s(1)}$. 
\end{lemma}

We prove this lemma in \Cref{subsec:sz-eigen}.  

We can now finish the proof of~\Cref{thm:dist-junta-sums} assuming the above lemma. On the one hand,~\eqref{eqn:forlla-b} implies that $$\Pr_{\substack{{\bf a}\sim V\\ {\bf b}\sim  W_\ODLSZ({\bf u})}}[{\bf a}\in U\text{~and~}{\bf b}\in U] \ge \paren{\frac{|U|}{|V|}}\frac{1}{s^d},$$ and on the other hand,~\Cref{thm:eml} and~\Cref{lem:eig-val-sz} imply that $$\Pr_{\substack{{\bf a}\sim V\\ {\bf b}\sim  W_\ODLSZ({\bf u})}}[{\bf a}\in U\text{~and~}{\bf b}\in U]\le \paren{\frac{|U|}{|V|}}\paren{{\frac{|U|}{|V|}+\frac{1}{n^{\Omega_s(1)}}}}.$$ Putting them together, we obtain that $\frac{|U|}{|V|} \ge \frac{1}{s^d} - \frac{1}{n^{\Omega_s(1)}}$, thus finishing the proof of~\Cref{thm:dist-junta-sums}.
\end{proof}

We now prove the near-optimal distance lemma for algebraic degree (\Cref{thm:dist-polys-intro}).

\begin{proof}[Proof of~\Cref{thm:dist-polys-intro}]
    The proof follows exactly the same approach as that of the distance lemma for junta-sums over the balanced multislice (i.e., \Cref{thm:dist-junta-sums}). All the additions and subtraction of the domain elements are now instead done over the field $\F$ instead of the group $\Z_s$. The only other difference is in~\eqref{eqn:forlla-b} where we now get a lower bound of $\delta(q,d)$ instead of $1/s^d$. This is because the restricted function $Q:\F^m \to \F$ is now a function of degree at most $d$, so we can apply the standard ODSLZ lemma (\Cref{thm:odlsz}) instead of the junta-sum distance lemma (\Cref{clm:dist-junta-sums}) to get this bound. Due to its similarity with the proof of~\Cref{thm:dist-junta-sums}, we omit the rest of the details. 
\end{proof}

Hence, it only remains to prove the eigenvalue bounds for $W_\ODLSZ$, i.e.,~\Cref{lem:eig-val-sz}, which we do in the next subsection.

\subsection{Eigenvalue Bounds for $W_\ODLSZ$}\label{subsec:sz-eigen}

Before we proceed with the proof of~\Cref{lem:eig-val-sz}, we remark that it doesn't immediately follow from our result for nearly {\em balanced} random walks (i.e,~\Cref{cor:gen-bal-cor} from~\Cref{sec:eigenvalue}) since $W_{\ODLSZ}$ can potentially have non-zero weights even for edges whose generalized Hamming distance is far from being balanced. Moreover, it doesn't even immediately follow from our more general theorem (\Cref{thm:more-general-matrix-eigenvalue}) from~\Cref{sec:eigenvalue} as it requires a bounded Frobenius norm which isn't the case with $W_\ODLSZ$. However, we are able to reduce it to a setting where~\Cref{cor:gen-bal-cor} actually applies and use it to get the final bound.

\begin{proof}[Proof of~\Cref{lem:eig-val-sz}]
    At a high level, we prove this in the following steps. We first provide an alternate description of the random walk matrix $W_\ODLSZ$ (defined in~\Cref{alg:ab}) using generalized Hamming distance matrices. Then, we express $W_\ODLSZ$ as a convex combination $W_\ODLSZ=\sum_{i\in [t]}\alpha_i W_i$ for some random walk matrices $W_i$ where $\sum_i \alpha_i =1$. Then, we use our expansion result for nearly balanced walks (\Cref{cor:gen-bal-cor}) from~\Cref{sec:eigenvalue} to bound $\lambda_2(W_i)$ for ``most'' $i\in [t]$, and use this to finally bound $\lambda_2(W_\ODLSZ)$. Before we go into the actual proof, we need to recall a few definitions.

    For ${\bf a}, {\bf b}\in V$, we recall (from~\Cref{defn:profile}) that $\Delta({\bf a},{\bf b})\in \Z^{\Z_s\times \Z_s}$ denotes the {\em generalized Hamming distance matrix}, i.e., the $(i,j)$-th entry of the matrix equals the number of coordinates where ${\bf a}$ takes value $i$ and ${\bf b}$ takes value $j$. We now recall the definition of $W_\ODLSZ$ (from~\Cref{alg:ab}): For each ${\bf a}\in V$, its random neighbor ${\bf b}\sim W_\ODLSZ({\bf a})$ is obtained by setting $b_i = y_{M_j(i)}+j$, where $j=a_i$ and $M_j:{\bf a}^{-1}(j) \to [m]$ are bijections chosen u.i.a.r., and $\y \sim \Z_s^n$ is chosen independently. Hence, we see that $$\Delta({\bf a},{\bf b})(i,j) = f_{j-i},$$ where $f_j$ denotes the number of times $j\in \Z_s$ appears in $\y$. In fact, conditioned on $\Delta({\bf a},{\bf b}) = P$ for some fixed $P$, the conditional distribution of ${\bf b}$ is uniform over all points ${\bf b}$ such that $\Delta({\bf a},{\bf b}) = P$, since $M_j$'s are uniform and independent bijections. In particular, this alternate description of $W_\ODLSZ$ shows that it is symmetric.

    Now, for a ``frequency vector'' ${\bf f} = (f_0,\dots, f_{s-1})\in \Z^{\Z_s}$ where $\sum_j f_j = m$, we let $W_{{\bf f}}$ denote the random walk matrix where for each ${\bf a}\in V$, $W_{{\bf f}}({\bf a})$ is the uniform distribution over \begin{align}\label{eqn:wf-defn}\bigg\{{\bf b}\in V: \Delta({\bf a},{\bf b}) = (f_{j-i})_{(i,j)\in \Z_s^2}\bigg\}.\end{align} Then by our previous discussion, for each ${\bf a}\in V$, by conditioning on the choice of the frequency vectors resulting from ${\bf y}\sim \mathcal{S}^m$ and using the total probability law, we obtain $$ W_\ODLSZ({\bf a}) = \sum_{\substack{{\bf f}\in\Z^s\\ \sum_j f_j = m}} \alpha_{{\bf f}} W_{{\bf f}}({\bf a}),$$ where $\alpha_{{\bf f}} = {m \choose {\bf f}}/s^m$ denotes the probability of getting the frequency vector ${\bf f}$ from a uniformly random $\y \in \mathcal{S}^m$.

    The idea now is to apply our eigenvalue bound (\Cref{cor:gen-bal-cor}) from~\Cref{sec:eigenvalue} to the $W_{{\bf f}}$'s and then bound the eigenvalues of $W_\ODLSZ$. However, there are two issues: First, the eigenvalue bound from~\Cref{cor:gen-bal-cor} requires the matrix to be supported only on edges with nearly balanced generalized Hamming distance, which isn't the case for {\em all} $W_{{\bf f}}$'s. Regardless, we show that this holds true for the ``typical'' $W_{{\bf f}}$'s and that this suffices. 
    And secondly, we remark that each $W_{{\bf b}}$ need not be a symmetric matrix. However, since we already know that $W_\ODLSZ=\sum_{{\bf f}} \alpha_{{\bf f}}W_{{\bf f}}$ is symmetric, we have that 
    $$W_\ODLSZ = \sum_{{\bf f}} \alpha_{{\bf f}}\paren{\frac{W_{{\bf f}}+W_{{\bf f}}^\top}{2}},$$ where now we see that the ``components'' $\frac{W_{{\bf f}}+W_{{\bf f}}^\top}{2}$ are symmetric.

    We say that a frequency vector ${\bf f}$ is ``bad'' if there exists a $j\in \Z_s$ such that $f_j \notin \frac{m}{s}\pm \sqrt{\frac{10m\log m}{s}}$, and say that ${\bf f}$ is ``good'' otherwise. We have, by a Chernoff bound, that
    \begin{align}\label{eqn:cher}
        \sum_{{\bf f}\text{~bad}} \alpha_{{\bf f}} \le \frac{s}{m^{\Omega(1)}}. 
    \end{align} Now we claim that for every good ${\bf f}$, it holds that $\lambda_2\paren{W'_{{\bf f}}}$ is small where $W'_{{\bf f}}:=\frac{W_{{\bf f}}+W_{{\bf f}}^\top}{2}$:

    \begin{claim}[{\bf Eigenvalue bounds for $W'_{{\bf f}}$}]\label{lem:johnson}
        Suppose ${\bf f} = (f_0,\dots,f_{s-1})$ is such that $f_j \in \frac{m}{s} \pm \sqrt{\frac{10m\log m}{s}}$ for all $j\in \Z_s$. Then, $\lambda_2(W'_{{\bf f}}) \le \frac{1}{n^{\Omega_s(1)}}$, where $W'_{{\bf f}}:=\frac{W_{{\bf f}}+W_{{\bf f}}^\top}{2}$. 
    \end{claim}

\begin{proof}
    We note that the matrix $W'_{{\bf f}}$ respects symmetries and has non-zero entries only on entries corresponding to a balanced generalized Hamming distance of either $\Delta$ or $\Delta^\top$ (both of which are $(10/s)$-balanced). Hence the proof follows directly by applying~\Cref{cor:gen-bal-cor} to the matrix $W_{{\bf f}}'$.
\end{proof}
    We now bound the eigenvalues of $W_\ODLSZ$ and finish the proof of~\Cref{lem:eig-val-sz}. By applying~\Cref{lem:convex} with $S$ being the set of good ${\bf f}$, we have 
    \begin{align*}
        \lambda_2(W_\ODLSZ) & = \lambda_2\paren{\sum_{{\bf f}}\alpha_{{\bf f}}W'_{{\bf f}}}\\
        & \le  \max_{{\bf f}\text{~good}}\{\lambda_2(W_{{\bf f}}')\} + \sum_{{\bf f}\text{~bad}} \alpha_{{\bf f}}\tag{using~\Cref{lem:convex}}\\
        & \le  \paren{\frac{1}{n^{\Omega_s(1)}} }+ \paren{\sum_{{\bf f}\text{~bad}} \alpha_{{\bf f}}} \tag{applying~\Cref{lem:johnson} to the random walk matrix $W'_{{\bf f}}$}\\
        & \le \frac{1}{n^{\Omega_s(1)}}. \tag{using~\eqref{eqn:cher} }
    \end{align*}
    
    The above bound shows that all the eigenvalues of $W_\ODLSZ$ except the largest one must be bounded above by $1/n^{\Omega_s(1)}$ in absolute value, i.e., $\lambda_2(W_\ODLSZ) \le 1/n^{\Omega_s(1)}$ proving~\Cref{lem:eig-val-sz}.
\end{proof}

\section{Local List Correction of Junta-Sums}\label{sec:llc}
In this section, we will prove the following theorem which is a restatement of \Cref{thm:local-list-correction-intro} with explicit bounds on the query complexity.

\begin{thmbox}
\begin{restatable}[Local List Correction]{theorem}{locallistcorr}\label{thm:local-list-correction}
For every Abelian group $G$ and for every $\varepsilon>0$, the space $\mathcal{J}_{d}(\sgrid^{n}, G)$ is $(1/s^{d}-\varepsilon, \; \bigO_{\varepsilon}(1), \; \Tilde{\bigO}_{\varepsilon}(\log n)^{d}, \; \bigO_{\varepsilon}(1))$-locally list correctable.\\

In particular, there is a randomized algorithm $\mathcal{A}$ such that for a function $f: \sgrid^{n} \to G$ and a parameter $\varepsilon > 0$, $A^{f}(\varepsilon)$ outputs with probability $\geq 3/4$ a list of randomized algorithms $\set{\phi_{i}}_{i=1}^{L}$ ($L = \bigO_{\varepsilon}(1)$)  such that the following holds. For each junta degree-$d$ function $P \in \mathcal{J}_{d}$ that is $(1/s^{d} - \varepsilon)$-close to $f$, there exists at least one randomized algorithm $\phi_i$ such that $\phi_{i}^{f}$ computes $P$ correctly on every input in $\sgrid^{n}$ with probability at least $3/4.$\\

The algorithm $\mathcal{A}$ makes $\bigO_{\varepsilon}(1)$ queries to $f$, while each $\phi_i$ makes $\Tilde{\bigO}_{\varepsilon}(\log n)^{d}$ oracle queries to $f$.
\end{restatable}
\end{thmbox}

We remark that for the Boolean case,~\cite{ABPSS25} proves that one can reduce the number of queries to a constant depending only on $\varepsilon$ and the {\em torsion} (or {\em exponent}) of the group (see~\Cref{sec:prelims} for a definition). Similarly, in this case, we get a similar statement where the algorithm $\A$ makes $\bigO_\varepsilon(1)$ queries, and each $\phi_i$ makes $\bigO_{M,\varepsilon}(1)$ queries where $M$ is the exponent of the torsion Abelian group $G$. More formally, we prove the following. \\

\begin{theorem}\label{thm:const-exp}
    For every torsion Abelian group $G$ of exponent $M>0$ and every $\varepsilon>0$, the family $\J_d(\sgrid^n, G)$ is $(1/s^d-\varepsilon, \bigO_\varepsilon(1), \bigO_{M,\varepsilon}(1), \bigO_\varepsilon(1))$-locally list correctable. 
\end{theorem}

As stated earlier, \cite{ABPSS25} gave a local list corrector for degree-$d$ \emph{polynomials} over $\sgrid = \Boo$. We note that most of their proof can be extended to \emph{junta sums} and general $\sgrid$ with some extensions to their arguments. However, a key challenge was to show that certain random walk matrix has good spectral expansion. In particular, \cite[Lemma 5.1.1]{ABPSS25} is proved by analyzing the eigenvalues of matrices defined on Johnson graphs. To extend their argument to general grids, we have to analyze the eigenvalues of random walk matrices on the balanced multi-slice. In this section, we describe the random walk matrix arising from the analysis of our local list corrector and show that it has ``large'' spectral gap, using \Cref{cor:gen-bal-cor}. We first give a quick overview of the algorithm, which is an extension of \cite[Algorithm 3 and Algorithm 4]{ABPSS25}.

\paragraph*{Overview of the local list corrector.}Similar to the work of \cite{ABPSS25}, our local list corrector goes as follows:

\begin{itemize}
    \item We design a local corrector $\mathcal{J}_{d}(\sgrid^{n}, G)$ (see \Cref{thm:local-correction}).
    \item We show a combinatorial list decoding bound for $\mathcal{J}_{d}(\sgrid^{n},G)$ (see \Cref{thm:comb-bd}).
    \item We design \emph{approximating oracles} for $\mathcal{J}_{d}(\sgrid^{n}, G)$ (see \Cref{thm:approx-oracles-list-decoding}).
    \item Combining an approximating oracle with local corrector, we get a local list corrector. The bound on query complexity follows from the combinatorial list decoding bound.
\end{itemize}

Our key technical contribution is in analyzing the approximating oracles. We use a very similar algorithm for approximating oracles as in \cite{ABPSS25}, however the correctness is more involved. The first two steps are again analogous to \cite[Section 3 and Section 4]{ABPSS25}. Most of the arguments follow with a simple extension from $\Boo$ to $\sgrid$, and few arguments require a bit more careful analysis. For the sake of completeness, we give a proof for local corrector and combinatorial list decoding bound. We state the results for them below, and after that we will proceed with the local list corrector.

\begin{thmbox}
\begin{restatable}[\em Local correction of junta-sums]{theorem}{localcorr}\label{thm:local-correction} 
     For every $\varepsilon > 0$, finite set $\sgrid$ of size $s\ge 2$ and $d\ge 0$, Abelian group $G$, the family $\mathcal{J}_{d}(\sgrid^{n}, G)$ is $(\widetilde{\bigO}_{\varepsilon}(\log n)^{d},\delta_\mathcal{J}/2 - \varepsilon)$-locally correctable where $\delta_\J:=1/s^d$.
     
    Moreover, if $G$ is a torsion Abelian group of exponent $M$, then the number of queries can be made $O_{M,\varepsilon}(1)$, i.e., $\J_d(\sgrid^n, G)$ is $(\bigO_{M,\varepsilon}(1), \delta_J/2-\varepsilon)$-locally correctable.
\end{restatable}
\end{thmbox}

\begin{thmbox}
\begin{restatable}[Combinatorial List Decoding Bound]{theorem}{combbound}\label{thm:comb-bd}
    For every $\varepsilon>0$, positive integers $s,d$, and Abelian group $G$, the family $\mathcal{J}_{d}(\sgrid^{n}, G)$ is $(1/s^d-\varepsilon, \bigO_{\varepsilon}(1))$-list decodable.
\end{restatable}    
\end{thmbox}

\noindent
For every $f: \sgrid^{n} \to G$ which is $(\frac{1}{s^{d}} - \varepsilon)$-close to $\mathcal{J}_{d}(\sgrid^{n},G)$, let $\mathsf{List}_{\varepsilon}(f)$ denote the set of $d$-junta-sums that have distance $\leq (1/s^{d} - \varepsilon)$ to $f$, i.e.
\begin{align*}
    \mathsf{List}_{\varepsilon}(f) \; = \; \setcond{P \in \mathcal{J}_{d}(\sgrid^{n}, G)}{\delta(f,P) \leq \frac{1}{s^{d}} - \varepsilon}.
\end{align*}

We give a proof for \Cref{thm:local-correction} and \Cref{thm:comb-bd} later. We informally state a standard observation\footnote{This is also used in \cite{ABPSS24, ABPSS25}. See \cite[Section 5]{ABPSS25} for a more elaborate discussion on it.} in the literature of local list correctors which says that given local correctors, it is enough to design approximating oracles (see \Cref{thm:approx-oracles-list-decoding}):\newline

\noindent
\emph{If there exists a local corrector, then it suffices to design an algorithm which outputs a list of algorithms with the guarantee - For every junta-sum $P$ in the list, there exists an algorithm $A$ in the list which computes $P$ correctly on sufficiently large fraction of $\sgrid^{n}$, and then we can run the local corrector on $A$.}

\noindent
So the focus in this is to design the approximating oracles. The following theorem is larger-grid analogue to \cite[Theorem 5.0.1]{ABPSS25}.\\

\begin{thmbox}
\begin{restatable}[Approximate oracles]{theorem}{approxlocallistcorrection}\label{thm:approx-oracles-list-decoding}
Fix $n \in \mathbb{N}$, $\varepsilon > 0$. Let $f: \sgrid^{n} \to G$ be any function and $L(\varepsilon) := |\mathsf{List}_{\varepsilon}(f)|$. There exists a randomized algorithm $\mathcal{A}_{1}^{f}$ that makes at most $\bigO_{\varepsilon}(1)$ oracle queries and outputs deterministic algorithms $\Psi_{1},\ldots, \Psi_{L'}$ satisfying the following property:\newline
With probability at least $3/4$, for every junta-sum $P \in \mathsf{List}_{\varepsilon}^{f}$, there exists a $j \in [L']$ such that
\begin{enumerate}
    \item $\delta(\Psi_{j}, P) < 1/(10 \cdot 2^{d+1})$
    \item For every $\mathbf{x} \in \sgrid^{n}$, $\Psi_{j}$ computes $P(\mathbf{x})$ by making at most $\bigO_{\varepsilon}(1)$ oracle queries to $f$.
\end{enumerate}
Here $L' = \bigO(L(\varepsilon/2)\log L(\varepsilon)) = \bigO_\varepsilon(1)$.
\end{restatable}
\end{thmbox}

\noindent
We first show that using \Cref{thm:approx-oracles-list-decoding} and \Cref{thm:local-correction}, we can prove \Cref{thm:local-list-correction} (and~\Cref{thm:const-exp}).

\begin{proof}[Proof of \Cref{thm:local-list-correction} and~\Cref{thm:const-exp}]
We first employ \Cref{algo:error-reduction} with oracle access to $f$ and it outputs deterministic algorithms $\psi_{1},\ldots,\psi_{L'}$ where $L' = \bigO(L(\varepsilon/2) \log L(\varepsilon))$. Next, we run the local corrector for $\mathcal{J}_{d}(\sgrid^{n},G)$ on each of $\psi_{j}$. This completes the description of our local list corrector.\newline
The correctness and query complexity now follow by combining \Cref{thm:approx-oracles-list-decoding} and \Cref{thm:local-correction}.
\end{proof}

\paragraph{Organization of the section}We start by proving a sampling lemma for the \emph{balanced slice} of the grid $\sgrid^{n}$ in  \Cref{subsec:sampling-balanced-slice}. The key tool to prove this sampling lemma will be to show that a certain random walk matrix (it arises from our sampling procedure) is a ``good spectral expander'' (see \Cref{thm:W-eigenvalue-bound}). We will prove by employing \Cref{cor:gen-bal-cor}. After the sampling lemma, we then prove a sub-optimal distance lemma for $d$-juntas on multi-slices of $\sgrid^{n}$ (see \Cref{thm:dist-multislice}). Combining the sampling lemma (\Cref{lemma:sampler-balanced-slice}) and the distance lemma on slices (\Cref{thm:dist-multislice}), we get \Cref{coro:restriction-do-not-kill}. This corollary will be useful in showing that our local list correctors have a small error probability, i.e., \Cref{coro:restriction-do-not-kill} will bound the probability of our local list correctors making a certain type of error. Once we have these statements, we describe a subroutine in \Cref{subsec:error-redn-local-list} and the local list correctors in \Cref{subsec:algo-local-list}. Finally, we analyze the algorithms in \Cref{subsec:analysis-algo-local-list}.

\subsection{A Sampling Lemma for the Balanced Multislice}\label{subsec:sampling-balanced-slice}
\begin{definition}\label{defn:tau}
Let $k,s \in \mathbb{N}$. For a $s$-to-$1$ map $\tau:[sk] \to [k]$, let $\mathsf{C}_{\tau} \subset \sgrid^{sk}$ denote the $k$-dimensional subgrid obtained by identifying coordinates acccording to $\tau$. More precisely, for every $\mathbf{y} \in \sgrid^{k}$, let $x_{\tau}(\mathbf{y}) \in \sgrid^{sk}$ be defined as follows:
\begin{align*}
    x_{\tau}(\mathbf{y})_{i} \; = \; y_{\tau(i)}, && \text{ for all } \; i \in [sk].
\end{align*}
Define $\mathsf{C}_{\tau} := \setcond{x_{\tau}(\mathbf{y})}{\mathbf{y} \in \sgrid^{k}}$.
\end{definition}

\noindent
The main lemma of this subsection is to show that if we sample a uniformly random $s$-to-$1$ map $\tau$, then $\mathsf{C}_{\tau}$ is a \emph{good sampler} for the balanced slice of $\sgrid^{sk}_{k,\ldots,k}$.\\

\begin{lemmabox}
\begin{lemma}[Sampler for the Balanced Slice]\label{lemma:sampler-balanced-slice}
Let $k,s \in \mathbb{N}$. There exists an absolute constant $\eta = \eta(s) > 0$ such that for every subset $S \subseteq \sgrid^{s^{2}k}_{sk,\ldots,sk}$, we have,
\begin{align*}
    \Pr_{\tau}\brac{ \left| \dfrac{|S|}{|\sgrid^{s^{2}k}_{sk,\ldots,sk}|} \; - \; \dfrac{|S \cap \mathsf{C}_{\tau}|}{|\sgrid^{sk}_{k,\ldots,k}|} \right| \; \geq \; \dfrac{1}{k^{\eta}} } \; \leq \; \bigO_{s}\paren{\dfrac{1}{k^{\eta}}},
\end{align*}
where the probability is over the choice of a random $s$-to-$1$ map $\tau:[s^{2}k] \to [sk]$.
\end{lemma}
\end{lemmabox}

\paragraph{Description of the matrix $W$:}  For this section, we will assume that $n$ is divisible by $s$. We will use $\mu$ to denote the \emph{balanced partition} of $n$ into $s$ rows, i.e. $\mu = (n/s, n/s, \ldots, n/s)$. Let $N$ denote the number of points in the balanced slice $\sgrid^{n}_{\mu}$, i.e. $N = |\sgrid^{n}_{\mu}| = \binom{n}{n/s,n/s,\ldots,n/s}$.\\

\begin{definition}[The matrix $W$]\label{defn:matrix-W}
    We will define the random walk matrix by the joint distribution over $\mathcal{S}^n_\mu \times \mathcal{S}^n_\mu$ represented by the matrix $W/N$. In particular, it is the joint probability distribution of $({\bf u},{\bf v})$ corresponding to picking a uniformly random vertex ${\bf u}\sim \mathcal{S}^n_\mu$ and ${\bf v}\sim W({\bf a})$ is its random neighbor corresponding to taking a random step according to $W$. We then define $W/N$ according to the distribution of the output of the following steps:
    \begin{enumerate}
        \item Pick ${\bf a},{\bf b}\sim \mathcal{S}^{sk}_{k,\dots,k}$ uniformly and independently at random. 
        \item Pick a $s$-to-1 map $\tau:[s^2k]\to [sk]$ uniformly at random.
        \item Output $({\bf u},{\bf v})=(x_\tau({\bf a}),x_\tau({\bf b}) )$ (see~\Cref{defn:tau} for the definition of $x_\tau(\cdot)$).
    \end{enumerate}
\end{definition}
We note that $W$ is symmetric since the joint probability distribution of $({\bf u},{\bf v})$ above is symmetric w.r.t.~${\bf u}$ and ${\bf v}$. We further claim below that $W$ has good spectral expansion:

\begin{thmbox}
\begin{restatable}[Spectral expansion of the random walk matrix]{theorem}{expandermatrix}\label{thm:W-eigenvalue-bound}
Let $W \in \mathbb{R}^{N \times N}$ be the symmetric random walk matrix as described previously. Denote the second largest eigenvalue of $W$ (in terms of absolute value) by $\lambda_{2}(W)$. Then there exists $\nu = \nu(s) > 0$ such that
\begin{align*}
    \lambda_{2}(W) \; \leq \; \dfrac{1}{n^{\nu}}.
\end{align*}
\end{restatable}
\end{thmbox}

We first prove \Cref{lemma:sampler-balanced-slice} assuming \Cref{thm:W-eigenvalue-bound}.

\begin{proof}[Proof of \Cref{lemma:sampler-balanced-slice}]
Let $\sigma := |S|/|\sgrid^{s^{2}k}_{sk,\ldots,sk}|$. For every $\mathbf{y} \in \sgrid^{k}$, define $Z(\mathbf{y})$ to be the indicator variable which is $1$ if $x_{\tau}(\mathbf{y}) \in S$. For a uniformly random $s$-to-$1$ map $\tau$, for every $\mathbf{y} \in \sgrid^{sk}_{k,\ldots,k}$, the random variable $x_{\tau}(\mathbf{y})$ is uniformly distributed in $\sgrid^{s^{2}k}_{sk,\ldots,sk}$. Thus for every $\mathbf{y} \in \sgrid^{sk}_{k,\ldots,k}$, the $\mathbb{E}_{\tau}[Z(\mathbf{y})] = \sigma$. Let $Z := |S \cap \mathsf{C}_{\tau}| = \sum_{\mathbf{y} \in \sgrid^{sk}_{k,\ldots,k}} Z_{\mathbf{y}}$ and by linearity of expectation, we have $\mathbb{E}_{\tau}[Z] = |\sgrid^{sk}_{k,\ldots,k}| \cdot \sigma$. We will now bound the variance of $Z$.\\

Using the linearity of expectation, we have,
\begin{align*}
    \mathbb{E}_{\tau}[Z^{2}] \; = \; \sum_{\mathbf{a},\mathbf{b} \in \sgrid^{sk}_{k,\ldots,k}} \mathbb{E}_{\tau}[Z(\mathbf{a}) \cdot Z(\mathbf{b})] \; = \; \sum_{\mathbf{a},\mathbf{b} \in \sgrid^{sk}_{k,\ldots,k}} \Pr_{\tau}[x_{\tau}(\mathbf{a}) \in S \, \wedge \, x_{\tau}(\mathbf{b}) \in S ]
\end{align*}
If we sample $\mathbf{a}, \mathbf{b}$ uniformly and independently at random from $\sgrid^{sk}_{k,\ldots,k}$, then by the definition of the matrix $W$ (\Cref{defn:matrix-W}), we get the following equality:
\begin{equation}\label{eqn:joint-dist}
    \Pr_{\substack{\mathbf{a},\mathbf{b} \sim \sgrid^{sk}_{k,\ldots,k}  \\ \tau}}[x_{\tau}(\mathbf{a}) \in S \, \wedge \, x_{\tau}(\mathbf{b}) \in S ] \; = \; \Pr_{\substack{\mathbf{u} \sim \sgrid^{s^{2}k}_{sk,\ldots,sk}   \\ \mathbf{v} \sim W(\mathbf{u})}}[\mathbf{u} \in S \; \wedge \; \mathbf{v} \in S].
\end{equation}

Using the Expander Mixing Lemma (see e.g.~\Cref{thm:eml}),
\begin{align*}
\Pr_{\substack{\mathbf{u} \sim \sgrid^{s^{2}_{sk,\ldots,sk}} \\ \mathbf{v} \sim W(\mathbf{u})}}[\mathbf{u} \in S \, \wedge \, \mathbf{v} \in S] \; \leq \; \sigma^{2} + \lambda_{2}(W) \\
\Rightarrow \mathbb{E}_{\tau}[Z^{2}] \; \leq \; |\sgrid^{sk}_{k,\ldots,k}|^{2} \cdot (\sigma^{2} + \lambda_{2}(W)) \\
\Rightarrow \mathrm{Var}[Z] \; = \; \mathbb{E}[Z^{2}] - (\mathbb{E}[Z])^{2} \; \leq \; |\sgrid^{sk}_{k,\ldots,k}|^{2} \cdot \lambda_{2}(W).
\end{align*}
Now using Chebyshev's inequality on $Z$, we get,
\begin{align*}
    \Pr_{\tau}\brac{\left| Z - \sigma \cdot |\sgrid^{sk}_{k,\ldots,k}|   \right|  \geq \dfrac{1}{k^{\eta}} \cdot |\sgrid^{sk}_{k,\ldots,k}|  } \; \leq \; \mathrm{Var}(Z) \cdot \dfrac{k^{2\eta}}{|\sgrid^{sk}_{k,\ldots,k}|^{2}} \; \leq \; \lambda_{2}(W) \cdot k^{2 \eta}
\end{align*}
From \Cref{thm:W-eigenvalue-bound}, we know that $\lambda_{2}(W) \leq 1/k^{ \eta}$, which implies that
\begin{align*}
    \Pr_{\tau}\brac{\left| Z - \sigma \cdot |\sgrid^{sk}_{k,\ldots,k}|   \right|  \geq \dfrac{1}{k^{\eta}} \cdot |\sgrid^{sk}_{k,\ldots,k}|  } \; \leq \; \dfrac{1}{k^{ \eta}}.
\end{align*}
This finishes the proof of \Cref{lemma:sampler-balanced-slice}.
\end{proof}

We now prove~\Cref{thm:W-eigenvalue-bound} which bounds the eigenvalues of the random walk matrix $W$.

\begin{proof}[Proof of~\Cref{thm:W-eigenvalue-bound}]
    We recall (from~\Cref{defn:matrix-W}) that the random walk $W$ is over the balanced multislice $V:=\mathcal{S}^{s^2k}_\mu=\mathcal{S}^{s^2k}_{sk,\dots,sk}$ and $N=|\mathcal{S}^{n}_\mu|$ denotes the number of vertices where $n=s^2k$. We will use the following equivalent description of $W$. We observe that $W/N$ is the joint probability distribution of $({\bf u},{\bf v})$ corresponding to picking a uniformly random vertex ${\bf u}\sim V$ and ${\bf v}\sim W({\bf a})$ is its random neighbor corresponding to a taking a random step according to $W$. We now rephrase the description of $W/N$ from~\Cref{defn:matrix-W}:
    \begin{enumerate}
        \item Pick ${\bf a},{\bf b}\sim \mathcal{S}^{sk}_{k,\dots,k}$ uniformly and independently at random. 
        \item Let $P=\Delta({\bf a},{\bf b})$ and $\widetilde{P} = sP$.
        \item Output $({\bf u},{\bf v})$ such that $\Delta({\bf u},{\bf v}) = \widetilde{P}$ uniformly at random. 
    \end{enumerate}

    The above output is indeed distributed according to $W/N$ by noting that the map $\tau$ used in~\Cref{defn:matrix-W} is chosen uniformly and independently from ${\bf a},{\bf b}\sim \mathcal{S}^{sk}_{k,\dots, k}$ and for every such $\tau$ used in~\Cref{defn:matrix-W}, we have that $\Delta({\bf u},{\bf v}) = \Delta(x_\tau({\bf a}),{x_\tau({\bf b})}) = s\cdot \Delta({\bf a},{\bf b})$. 
    Now, applying a total probability rule over the choice of ${\bf a}$ and ${\bf b}$ in Step 1, we have:
    \begin{align}\label{eqn:conv}\frac{W}{N} = \sum_{P} \alpha_P \frac{W_{{sP}}}{N},\end{align} where we use $\alpha_P$ to denote the probability that $\Delta({\bf a},{\bf b}) = P$ for ${\bf a},{\bf b}$ chosen uniformly and independently at random; and $W_{sP}$ denotes the random walk over the multislice, determined by the generalized Hamming distance matrix $sP$ (see~\Cref{defn:w-delta}). We now say that a generalized Hamming distance matrix $P$ is {\em good} if it is $(10/s)$-balanced (by~\Cref{defn:bal-profiles}, this is equivalent to saying all entries of $P$ are $\frac{k}{s}\pm \sqrt{\frac{10k\log k}{s}}$) and $P$ is {\em bad} otherwise. It is easy to see that $sP$ is $10$-balanced w.r.t.~the multislice $\mathcal{S}^{s^2k}_\mu$ if $P$ is good. We first show that the mass of $\alpha_P$ on bad $P$ is small: that is, we show that $\sum_{P\text{~bad}} \alpha_P$, which denotes the probability that $\Delta({\bf a},{\bf b})$ is {\em not} $(10/s)$-balanced, is at most $ 1/k^{\Omega_s(1)}$. By fixing ${\bf a}$ and noting that ${\bf b}$ is still uniformly distributed over $\mathcal{S}^{sk}_\mu$, we see that each entry of $P$ is distributed according to a hypergeometric distribution with a total of $sk$ states and $k$ success states, and we are picking $k$ draws without replacement. By applying Hoeffding bound~\cite{hoeffding1994probability}, we get that this probability is at most $1/2^{-\Omega_s(\sqrt{\log k /k})^2k} \le 1/k^{-\Omega_s(1)}$. Now by a union bound over all the $k^2$ entries of the matrix $P$, we get that the probability that $P$ is not $(10/s)$-balanced is at most $1/n^{\Omega_s(1)}$ as claimed. That is, \begin{align}\label{eqn:bad-hoeff-prob}
        \sum_{P\text{~bad}} \alpha_P \le 1/k^{\Omega_s(1)}.
    \end{align}
    We now use our main eigenvalue bound from~\Cref{sec:eigenvalue} to bound $\lambda_2(W'_{sP})$ when $P$ is good, where $W_{sP}':=\frac{W_{sP}+W_{sP}^\top}{2}$. In particular, since we have that $sP$ and $(sP)^\top$ are $(10/s)$-balanced for good $P$, by applying~\Cref{cor:gen-bal-cor}, we have for all good $P$ that 
    \begin{align}\label{eqn:prof-eig-bound}
    \lambda_2(W'_{sP}) \le 1/k^{\Omega_s(1)}.
    \end{align}
    We note that since we showed that $W$ is symmetric,~\eqref{eqn:conv} implies that $$W=\sum_{P} \alpha_P \paren{\frac{W_{sP}+W_{sP}^\top}{2}} = \sum_{P} \alpha_P W'_{sP}.$$ We can therefore apply~\Cref{lem:convex} with the set $S$ being the set of good $P$ to conclude that $$\lambda_2(W) \le \max_{P\text{~good}}\{\lambda_2(W'_{sP})\} + \sum_{P\text{~bad}}\alpha_P \le 1/k^{\Omega_s(1)}\le 1/n^{\Omega_s(1)},$$ by using~\eqref{eqn:bad-hoeff-prob} and~\eqref{eqn:prof-eig-bound}.
    This finishes the proof of~\Cref{thm:W-eigenvalue-bound}.
\end{proof}

\subsection{Sub-optimal Distance Lemma Over Multislices}
In this subsection, we prove that if a junta sum does {\em not} vanish on a multi-slice, then it does not vanish on at least a \emph{constant fraction} of that multi-slice. It is a generalization of the distance lemma for junta-sums  (\Cref{clm:dist-junta-sums}), generalized from grids to slices. In the case of $\sgrid = \Boo$, such a statement was proved in \cite[Lemma 5.1.6]{ABPSS25}. They proved it by induction on the degree $d$. We observe that a similar induction also works for junta sums. We provide a proof below. For $\sgrid = \Boo$, our lower bound matches \cite[Lemma 5.1.6]{ABPSS25}.

\paragraph*{}For this, we will need the following notation. For integers $d\ge 0$ and $s\ge 2$ and $(n_i)_{i\in \sgrid}$, let $n=\sum_{i\in \sgrid} n_i$, ${\bf n} = (n_i)_{i\in \sgrid} \in \sgrid^{n}$. Let the multi-slice $\sgrid^{n}_{\mathbf{n}} \subseteq \sgrid^{n}$ denote the set of points which contain $n_{i+1}$ many occurrences of the element $i$ for all $i\in \sgrid$.  Let ${n \choose {\bf n}} = {n\choose {n_0,n_1,\dots,n_{s-1}}}$ denote the size of $\sgrid^{n}_{\mathbf{n}}$ (so it is zero if some $n_i$ is negative). We also use the notation ${\bf n}-d$ to denote the tuple $((n_i-d))_{i\in \sgrid}$.\\

\begin{lemmabox}
\begin{theorem}[Sub-optimal distance lemma for junta sums on multi-slices]\label{thm:dist-multislice}
    For every ${\bf n} = (n_i)_{i\in \sgrid}$ with $\sum_{i\in \sgrid} n_i = n$, the following holds. If a junta-sum $P \in \mathcal{J}_{d}(\sgrid^{n},G)$ is non-zero on the multi-slice $\sgrid^{n}_{\mathbf{n}}$ i.e. there exists a point $\mathbf{a} \in \sgrid^{n}_{\mathbf{n}}$ such that $P({\bf a}) \ne 0$, then
    \begin{align*}
        |\setcond{\mathbf{a} \in \sgrid^{n}_{\mathbf{n}}}{P(\mathbf{a}) \neq 0}| \; \geq \; \binom{n-sd}{\mathbf{n}-d}.
    \end{align*}
\end{theorem}
\end{lemmabox}
\begin{proof}[Proof of~\Cref{thm:dist-multislice}]
    The proof of the theorem is by induction on $d$. The base case $d=0$ is handled by noting that $P$ is a constant function in this case. Now suppose $d\ge 1$. We shall assume that $n \ge sd+1$ and $n_i \ge d$ for all $i\in \Z_s$, as the theorem statement is trivial otherwise. 

    We will assume that $P$ is not a constant function over ${[n]\choose {\bf n}}$ as otherwise we are done. In particular, we can always find two points ${\bf a},{\bf b} \in {[n]\choose {\bf n}}$ such that $P({\bf a})\ne P({\bf b})$ and they differ in exactly two coordinates; this follows by noting that we can move from any point on the multislice to any other point by swapping elements a finite number of times. Without loss of generality, we can assume that ${\bf a}$ and ${\bf b}$ differ on the first and last coordinates; i.e., $a_1 = b_n = \alpha$ and $a_n = b_1 = \beta$ for some $\alpha \ne \beta \in \Z_s$. Let ${\bf n}'=(n'_i)_{i\in \Z_s}$ be defined by $n'_i = n_i$ for $i\notin \{\alpha,\beta\}$ and $n'_i=n_i-1$ for $i\in \{\alpha, \beta\}$. We now consider the  function $Q:\sgrid^{n-2} \to G$ defined as:
    \begin{align*}
        Q(x_2,\dots,x_{n-1}) = P(\alpha,x_2,\dots,x_{n-1},\beta) - P(\beta,x_2,\dots,x_{n-1},\alpha).
    \end{align*}
    As $P({\bf a}) - P({\bf b}) \ne 0$, we see that $Q$ is not identically zero over ${[n-2] \choose {\bf n}'}$. We also claim that $Q$ is a $(d-1)$-junta-sum. Indeed, if $$P(x_1,\dots,x_n) = \sum_{{\bf c}\in \Z^n_s:|{\bf c}|\le d} g_{{\bf c}} \cdot \prod_{i\in [n]:c_i \ne 0}\delta_{c_i}(x_i),$$ then in the junta-polynomial of $Q$, all the monomials that do not contain either $x_1$ or $x_n$ will be canceled, while the monomials of degree $d$ that contain either $x_1$ or $x_n$ (or both) will reduce in degree. Hence, $Q \in \mathcal{J}_{d-1}(\sgrid^{n-2}, d-1)$. Now, by induction hypothesis, we have that there are at least ${n-2-s(d-1) \choose {\bf n}'-(d-1)}$ choices for ${\bf d} \in {[n-2]\choose {\bf n}'}$ such that $Q({\bf d}) \ne 0$. For each such ${\bf d}$, we have that $Q({\bf d}) = P(\alpha, {\bf d}, \beta) - P(\beta, {\bf d}, \alpha) \ne 0$ so either ${\bf a}' = (\alpha, {\bf d}, \beta)$ or ${\bf b}' = (\beta, {\bf d}, \alpha)$ is a non-zero of $P$. Furthermore, we can verify that ${\bf a}',{\bf b}' \in {[n]\choose {\bf n}'}$. Let ${\bf e}$ denote the tuple which is $1$ at all indices $i\notin \{\alpha,\beta\}$ and is 0 for $i\in \{\alpha, \beta\}$. Hence, the number of non-zeroes of $P$ over ${[n]\choose {\bf n}}$ is at least the number of such ${\bf d}$ which is at least $${n-2-s(d-1) \choose {\bf n}'-(d-1)} = {n-sd+(s-2) \choose {\bf n}'-(d-1)} = {n-sd+(s-2) \choose ({\bf n}-d)+{\bf e}} \ge {n-sd \choose {\bf n}-d}\cdot {s-2 \choose {\bf e}} \ge {n-sd \choose {\bf n}-d}.$$ 
\end{proof}

\paragraph{}Using \Cref{thm:dist-multislice}, we immediately get the following corollary, which gives a lower bound on the fraction of non-zeroes on the balanced multislice.\\

\begin{corollary}\label{coro:weak-distance-balanced-slice}
Let $n,s,d \in \mathbb{N}$ with $n\ge sd$ divisible by $s$. Let $\mu = (n/s,\ldots,n/s)$. If a junta-sum $P\in \mathcal{J}(\sgrid^{n}, d, G)$ is non-zero over $\sgrid^{n}_{\mu}$, i.e. there exists $\mathbf{a} \in \sgrid^{n}_{\mu}$ such that $P(\mathbf{a}) \neq 0$, then:
    \begin{align*}
        \Pr_{\mathbf{x} \sim \sgrid^{n}_{\mu}}[P({\mathbf{x}}) \neq 0] \; \geq \; \dfrac{1}{(sd)^{sd}}.
    \end{align*}
\end{corollary}

\begin{proof}[Proof of \Cref{coro:weak-distance-balanced-slice}]
    Let $n=ms$ for some $m\in \Z$. By \Cref{thm:dist-multislice}, we have that the probability of a random point in $\sgrid^{n}_{\mu}$ being non-zero for $P$ is at least:
    \begin{align*}
        \frac{{n-sd\choose {\bf n}-d}}{{n \choose {\bf n}}} & = \frac{(n-sd)!}{(m-d)!^s} \cdot \frac{m!^s}{n!}\\
        & = \frac{m(m-1)\dots (m-d+1)}{n(n-1)\dots(n-d+1)} \dots \frac{m(m-1)\dots (m-d+1)}{(n-(s-1)d)(n-(s-1)d-1)\dots(n-sd+1)}\\
        & = \paren{\frac{m}{n} \cdot \frac{m-1}{n-1} \cdot \dots \cdot  \frac{m-d+1}{n-d+1}}  \dots \paren{\frac{m}{n-(s-1)d} \cdot \frac{m-1}{n-(s-1)d-1} \cdot  \dots \cdot \frac{m-d+1}{n-sd+1}}\\
        & \ge \paren{\frac{m-d+1}{n-d+1}}^d  \dots \paren{\frac{m-d+1}{n-sd+1}}^d \tag{using $\frac{a}{b} \ge \frac{a-i}{b-i}$ for $0<i<a<b$} \\
        & \ge \paren{\frac{n-sd+s}{s(n-d+1)}}^{sd}\\
        & \ge \frac{1}{(sd)^{sd}}.\tag{using $\frac{n-sd+s}{n-d+1} \ge \frac{1}{d}$ since $n\ge sd-1$}
    \end{align*}
\end{proof}

\paragraph{}Using \Cref{lemma:sampler-balanced-slice} and \Cref{coro:weak-distance-balanced-slice}, we get the following corollary.\\

\begin{lemmabox}
\begin{corollary}\label{coro:restriction-do-not-kill}
There exists an absolute constant $\eta > 0$ for which the following holds. Let $R \in \mathcal{J}_{d}(\sgrid^{s^{2}k}, G)$ be a non-zero function and there exists a $\mathbf{w} \in \sgrid^{s^{2}k}_{sk,\ldots,sk}$ such that $R(\mathbf{w}) \neq 0$. Let $\tau:[s^{2}k] \to [sk]$ be a random $s$-to-$1$ map and $\mathsf{C}_{\tau}$ be the subgrid as defined before. Then,
\begin{align*}
    \Pr_{\tau}[R|_{\mathsf{C}_{\tau}} \; \text{ vanishes on } \; \sgrid^{sk}_{k,\ldots,k}] \; \leq \; \dfrac{1}{k^{\eta}}.
\end{align*}
\end{corollary}
\end{lemmabox}
\begin{proof}[Proof of \Cref{coro:restriction-do-not-kill}]
Let $S$ denote the set of non-zeroes of $R$ on the slice $\sgrid^{s^{2}k}_{sk,\ldots,sk}$, i.e. $S = \setcond{\mathbf{a} \in \sgrid^{s^{2}k}_{sk,\ldots,sk}}{R(\mathbf{a}) \neq 0}$. From \Cref{coro:weak-distance-balanced-slice}, we know that
\begin{align*}
    |S| \geq (1/(sd)^{sd}) \cdot |\sgrid^{s^{2}k}_{sk,\ldots,sk}| \; \Rightarrow \; \dfrac{|S|}{|\sgrid^{s^{2}k}_{sk,\ldots,sk}|} \, = \, \Omega(1).
\end{align*}
$R|_{\mathsf{C}_{\tau}}$ does not vanish on $\sgrid^{sk}_{k,\ldots,k}$ if $S \cap \mathsf{C}_{\tau} \neq \emptyset$. Using \Cref{lemma:sampler-balanced-slice}, we know that the probability of $S \cap \mathsf{C}_{\tau} = \emptyset$ (over the randomness in choice of $\tau$) is at most $1/k^{\eta}$.
\end{proof}

\subsection{Subroutine for Approximating Oracles}\label{subsec:error-redn-local-list}
\begin{definition}[Subgrid containing $\mathbf{b}$]\label{defn:subgrid-span}
Let $\mathsf{C} = C_{h,\Pi}$ be a $k$-dimensional subgrid of $\sgrid^{n}$ as defined in \Cref{defn:random-embedding}, where $h:[n] \to [k]$ is a hash function and $\Pi \in (\mathrm{Sym}[\sgrid])^{n}$ is a tuple of permutations. For an arbitrary $\mathbf{b} \in \sgrid^{n}$ and a permutation $\sigma \in \mathrm{Sym}_{sk}$ define a new hash function $h':[n] \to [sk]$ as follows:
\begin{align*}
    h'(i) \, = \, \sigma(h(i) + k \cdot b_{i}), && \text{ for all } \, i \in [n]
\end{align*}
For every $\mathbf{z} \in \sgrid^{sk}$, define $x_{h',\Pi}(\mathbf{z}) := \Pi_{i}(z_{h'(i)})$. Define the subset $\mathsf{C}^{\mathbf{b}}_{\sigma} \subset \sgrid^{n}$ as follows
\begin{align*}
    \mathsf{C}^{\mathbf{b}}_{\sigma} \; := \; \setcond{x_{h,\Pi}(\mathbf{z})}{\mathbf{z} \in \sgrid^{sk}}.
\end{align*}
\end{definition}

\noindent
We make a few observations from \Cref{defn:subgrid-span}. The first observation is that $\mathbf{b}$ is indeed in $\mathsf{C}^{\mathbf{b}}_{\sigma}$. The second observation is that for random $h, \Pi, \mathbf{b}$ and $\sigma$, the subgrid $\mathsf{C}^{\mathbf{b}}_{\sigma}$ is a random embedding of a $sk$-dimensional subgrid. The third observation is that $\mathsf{C}$ is a subgrid of $\mathsf{C}^{\mathbf{b}}_{\sigma}$ and is obtained by ``randomly pairing'' coordinates.\\

\begin{observation}\label{obs:basic-obs-bigger-subgrid}
\textsf{The point $\mathbf{b} \in \sgrid^{n}$ lies inside the subgrid $\mathsf{C}^{\mathbf{b}}_{\sigma}$, i.e. there exists a string $\mathbf{w} \in \sgrid^{sk}_{k,\ldots,k}$ such that $x_{h',\Pi}(\mathbf{w}) = \mathbf{b}$. More explicitly,
\begin{align*}
    w_{h(i)+k \cdot b_{i}} \, := \, \Pi^{-1}_{i}(b_{i}), && \text{ for all } \; i \in [n].
\end{align*}
Also it is easy to see that the partition of $[n]$ induced by $h'$ (as defined in \Cref{defn:subgrid-span}) is a refinement of the partition induced by $h$. This means $\mathsf{C} \subset \mathsf{C}^{\mathbf{b}}_{\sigma}$.}
\end{observation}

\noindent
\begin{observation}\label{obs:bigger-subgrid-is-also-random-embedding}
\textsf{Let $h,\Pi,$ and $\mathbf{b}$ (as stated in \Cref{defn:subgrid-span}) be randomly chosen. Then $\mathsf{C}^{\mathbf{b}}$ is a random embedding of a $sk$-dimensional subgrid, i.e. there exists a random hash function $H:[n] \to [sk]$ and a random $\Pi' \in (\mathrm{Sym}[\sgrid])^{n}$ such that $\mathsf{C}^{\mathbf{b}}_{\sigma}$ has the same distribution as $C_{H,\Pi'}$.}
\end{observation}

\noindent
\begin{observation}\label{obs:subgrid-restriction-distribution}
\textsf{Let $h, \Pi$, and $\mathbf{b}$ (as stated in \Cref{defn:subgrid-span}) be randomly chosen. Conditioned on the grid $\mathsf{C}^{\mathbf{b}}$, the subgrid $\mathsf{C}$ has the following distribution:\newline
Sample a random $s$-to-$1$ map\footnote{A map is $s$-to-$1$ if the pre-image of every element under the map has size exactly $s$, i.e., exactly $s$ elements from the domain have the same image.} $\tau: [sk] \to [k]$ and we identify $s$ variables together.}
\end{observation}

\subsection{The Algorithm}\label{subsec:algo-local-list}
In this subsection, we give the description of the algorithms to prove \Cref{thm:approx-oracles-list-decoding}. The algorithm proceeds in two steps, and this is similar to the algorithms in \cite[Section 5.2.2]{ABPSS25}, barring a few changes to handle larger grids $\sgrid$. We request the reader to refer to \cite[Section 5.2.2]{ABPSS25} for an overview and discussion on the algorithms.\\

In the following description, let $L(\varepsilon) = |\mathsf{List}_{\varepsilon}(f)|$, where recall that $\mathsf{List}_{\varepsilon}(f)$ is the set of $d$-junta-sums that are $(1/s^{d} - \varepsilon)$-close to $f$. Note that \Cref{algo:high-agreement} is a deterministic algorithm and all the randomness is in \Cref{algo:error-reduction},

\begin{algobox}
\begin{algorithm}[H]
\caption{Approximating Algorithm $\Psi[\mathsf{C},\sigma, Q]$}
\label{algo:high-agreement}
\DontPrintSemicolon

\KwIn{Oracle access to the function $f$, a point $\mathbf{b} \in \sgrid^{n}$}
\vspace{3mm}

Let $\mathsf{C}'$ be a subgrid spanned by $\mathsf{C}$ and $\mathbf{b}$ using $\sigma \in \mathrm{Sym}_{sk}$  \tcp*{see \Cref{defn:subgrid-span}}
Let $\mathbf{w} \in \sgrid^{sk}$ such that $x(\mathbf{w}) \in \mathsf{C}'$ and $x(\mathbf{w}) = \mathbf{b}$ \tcp*{see \Cref{obs:basic-obs-bigger-subgrid},  $|\mathbf{w}| \in \sgrid^{sk}_{k,\ldots,k}$}
\vspace{2mm}
Query $f$ on the subgrid $\mathsf{C}'$ \tcp*{Number of queries is $s^{sk}$}
\vspace{2mm}

Find all degree-$d$ junta-sums $R_{1},\ldots, R_{L''} \in \mathcal{J}_{d}(\sgrid^{sk}, \, G)$ that are $\paren{\frac{1}{s^{d}} - \frac{\varepsilon}{2}}$-close to $f|_{\mathsf{C'}}$ \;
\vspace{2mm}
\If{there exists an $i \in [L'']$ such that $R_{i}|_{\mathsf{C}} = Q$}{
pick any such $i$ and \Return{$R_{i}(\mathbf{w})$}
}
\Else{
\Return{$0$}\tcp*{An arbitrary value}
}

\end{algorithm}
\end{algobox}

Now we describe the randomized \Cref{algo:list-decoding} that returns the descriptions of the deterministic oracles.

\begin{algobox}
\begin{algorithm}[H]
\caption{Algorithm $\mathcal{A}_{1}$}
\label{algo:list-decoding}

\DontPrintSemicolon

\KwIn{Oracle access to the function $f$}
\vspace{3mm}

Choose $k \leftarrow B_{d} \paren{\frac{L(\varepsilon/2)}{\varepsilon}}^{c} $ \tcp*{ $B_{d}$ and $c$ are constants, chosen later in the analysis}
Set $\ell \leftarrow \log L(\varepsilon)$\;
$T \leftarrow \emptyset$\;
\Repeat{$\ell$ times}{
Sample $\Pi \in (\mathrm{Sym}[\sgrid])^{n}$ and a random hash function $h: [n] \to [k]$ \tcp*{the first source of randomness}
Construct the subgrid $\mathsf{C} := C_{h, \Pi}$ \tcp*{see \Cref{defn:random-embedding}}
\vspace{2mm}
Query $f$ on the subgrid $\mathsf{C}$ \tcp*{Number of queries is $2^{k}$}
\vspace{2mm}
Find all junta-sums $Q_{1},\ldots,Q_{L'} \in \mathcal{J}_{d}(\sgrid^{k}, \, G)$ that are $\paren{\frac{1}{s^{d}} - \frac{\varepsilon}{2}}$-close to $f|_{\mathsf{C}}$\;
Pick a uniformly random permutation $\sigma \sim \mathrm{Sym}_{sk}$ \tcp*{the second source of randomness}
$T \leftarrow T \cup \set{(\mathsf{C},\sigma,Q_{1}),\ldots,(\mathsf{C},\sigma,Q_{L'})}$\;
}
\vspace{2mm}
\Return{$\Psi[\mathsf{C},\sigma, Q]$ for all $(\mathsf{C},\sigma, Q) \in T$} \tcp*{Size of $T$ is $\leq \ell L'$}

\end{algorithm}
\end{algobox}

\subsection{Analysis of the Local List Corrector}\label{subsec:analysis-algo-local-list}
In this subsection, we analyze \Cref{algo:error-reduction} and \Cref{algo:high-agreement} to prove \Cref{thm:approx-oracles-list-decoding}. We recall the statement of \Cref{thm:approx-oracles-list-decoding}.\\

\approxlocallistcorrection*

\noindent
We start by  show that in a single iteration of \Cref{algo:list-decoding}, for every junta-sum $P \in \mathrm{List}^{f}_{\varepsilon}$, with probability at least $\geq 99/100$, there exists an approximating oracle $\Psi[\mathsf{C},\sigma,Q]$ such that $\delta(P, \Psi[\mathsf{C,\sigma,Q}])$ is at most $\leq 1/(10 \cdot s^{d+1})$.\\

\begin{lemma}[Error w.r.t a fixed junta-sum in one iteration]\label{lemma:local-list-single-iteration-error}
Fix a junta-sum $P \in \mathrm{List}^{f}_{\varepsilon}$. Then for every iteration of \Cref{algo:list-decoding}, the following holds:\newline
With probability $\geq 99/100$, over the randomness of \Cref{algo:list-decoding}, there exists a tuple $(\mathsf{C},\sigma,Q)$ such that
\begin{align*}
    \delta(P, \Psi[\mathsf{C},\sigma,Q]) \leq \dfrac{1}{10 \cdot s^{d+1}}.
\end{align*}
\end{lemma}

\begin{proof}[Proof of \Cref{lemma:local-list-single-iteration-error}]
Fix a particular iteration of the main loop of \Cref{algo:list-decoding}. In this iteration, there are three sources of errors:
\begin{enumerate}
    \item Event $\mathcal{E}_{1,P}$ (depends on $\Pi$ and $h$): There does not exist a junta-sum $Q \in \mathcal{J}_{d}(\sgrid^{k},G)$ such that $Q \equiv P|_{\mathsf{C}}$.

    \item Event $\mathcal{E}_{2,P}$ (depends on $\Pi,h,\sigma,\mathbf{b}$): Consider a tuple $(\mathsf{C},\sigma,Q_{i}) \in T$ added in this iteration. For the approximating algorithm $\Psi[\mathsf{C},\sigma,Q_{i}]$ (\Cref{algo:high-agreement}), there does not exist a junta-sum $R \in \mathcal{J}_{d}(\sgrid^{sk}, G)$ such that $R \equiv P|_{\mathsf{C}'}$. Observe that this event is independent of $Q_{i}$ and only depends on $\mathsf{C}, \mathbf{b}$, and $\sigma$.

    \item Event $\mathcal{E}_{3,P}$ (depends on $\Pi,h,\sigma,\mathbf{b}$): Consider a tuple $(\mathsf{C},\sigma,Q_{i}) \in T$ added in this iteration. For the approximating algorithm $\Psi[\mathsf{C},\sigma,Q_{i}]$ (\Cref{algo:high-agreement}), there exists two distinct junta-sums $R_{i}, R_{j} \in \mathcal{J}_{d}(\sgrid^{sk}, G)$ such that $R_{i}|_{\mathsf{C}} \equiv R_{j}|_{\mathsf{C}}$ but $R_{i}(\mathbf{w}) \neq R_{j}(\mathbf{w})$. In this situation, Line 6 of \Cref{algo:high-agreement} is not a well-defined instruction. This event also only depends on $\mathsf{C},\mathbf{b}$, and $\sigma$.
\end{enumerate}

The probability of $\mathcal{E}_{1,P}$ and $\mathcal{E}_{2,P}$ can be upper bounded by using \Cref{lemma:sampling-subgrid} on $\mathsf{C}$ and $\mathsf{C}'$ respectively. To upper bound, we use \Cref{coro:restriction-do-not-kill}.\\

\begin{claim}[Probabilities of the first two error events]
Let $\mathcal{E}_{1,P}$ and $\mathcal{E}_{2,P}$ be as defined above. Then,
\begin{align*}
    \Pr_{\Pi,h}[\mathcal{E}_{1,P}] \; \leq \; \dfrac{1}{10000 \cdot s^{d+1}} \quad \text{ and } \quad \Pr_{\Pi,h,\sigma,\mathbf{b}]}[\mathcal{E}_{2,P}] \; \leq \; \dfrac{1}{10000 \cdot s^{d+1}}.
\end{align*}
\end{claim}
\begin{proof}
Let us start with $\mathcal{E}_{1,P}$. Non-existence of a $Q \in \mathcal{J}_{d}(\sgrid^{k},G)$ such that $Q \equiv P|_{\mathsf{C}}$ is equivalent to $\delta(P|_{\mathsf{C}}, f|_{\mathsf{C}}) > (1/s^{d} - \varepsilon/2)$. Using \Cref{lemma:sampling-subgrid}, we get the desired bound.\newline
For $\mathcal{E}_{2,P}$, we use \Cref{obs:bigger-subgrid-is-also-random-embedding} and then proceed as in the case of $\mathcal{E}_{1,P}$. This finishes the proof of the claim.
\end{proof}

\noindent
The next claim is to upper bound the probability of the third error. Upper bounding this error uses the spectral expansion and is very different from the Boolean setting as in \cite{ABPSS25}.\\

\begin{claim}[Probability of the third error event]
Let $\mathcal{E}_{3,P}$ be as defined above. Then,
\begin{align*}
    \Pr_{\Pi,h,\sigma,\mathbf{b}}[\mathcal{E}_{3,P}] \; \leq \; \dfrac{1}{10000 \cdot s^{d+1}}.
\end{align*}
\end{claim}
\begin{proof}
Fix a subgrid $\mathsf{C}'$. This fixes the junta sums $R_{1},\ldots, R_{L''}$ in Line 4 of \Cref{algo:high-agreement}. Consider any two distinct junta sums $R_{i}$ and $R_{j}$ such that they differ on at least one point in $\sgrid^{sk}_{k,\ldots,k}$ (this includes the pairs which differ on $\mathbf{w}$). This means $R := R_{i} - R_{j}$ is non-zero on $\sgrid^{sk}_{k,\ldots,k}$. We want to upper bound the probability that $R_{i}|_{\mathsf{C}} \equiv R_{j}|_{\mathsf{C}}$ i.e. $R|_{\mathsf{C}} \equiv 0$.\\

\noindent
Using \Cref{obs:subgrid-restriction-distribution} and \Cref{coro:restriction-do-not-kill}, for appropriately chosen constants $B_{d}$ and $c$, the probability of $R|_{\mathsf{C}}$ vanishing is $\leq 1/(10000 \cdot s^{d+1} \cdot L(\varepsilon/2)^{2})$. We know that $L'' \leq L(\varepsilon/2)$. Doing an union bound on all possible pairs $(R_{i}, R_{j})$, we get the error probability is $\leq 1/(10000 \cdot s^{d+1})$. This finishes the proof of the claim.
\end{proof}
Combining the above three claims to bound the final error probability is analogous to the proof in \cite[Lemma 5.3.1]{ABPSS25}. As the proof is quite similar, we skip it here.\\
This finishes the proof of \Cref{lemma:local-list-single-iteration-error}.
\end{proof}

\noindent
The above lemma shows that for a fixed $P \in \mathsf{List}_{\varepsilon}(f)$, the algorithm returns an approximating oracle with high probability in a single iteration. We now use it to finish the proof of \Cref{thm:approx-oracles-list-decoding}.

\begin{proof}[Proof of \Cref{thm:approx-oracles-list-decoding}]
We first show the correctness of \Cref{algo:error-reduction}. Fix any $P \in \mathsf{List}_{\varepsilon}(f)$. From \Cref{lemma:local-list-single-iteration-error}, we know that \Cref{algo:error-reduction} returns a tuple $(\mathsf{C},\sigma,Q)$ for which $\Psi[\mathsf{C},\sigma,Q]$ is $\leq 1/(10 \cdot s^{d+1})$-close with probability $\geq 0.99$. \Cref{algo:error-reduction} has $\ell = \log L(\varepsilon)$ many independent iterations. Thus at the end of $\ell$ iterations, the probability of the event that there is no tuple $(\mathsf{C},\sigma,Q)$ added in $T$ such that $\Psi[\mathsf{C},\sigma,Q]$ is $\leq 1/(10 \cdot s^{d+1})$-close to $P$ is $\leq 1/100^{\ell}$. By a union bound over all $P \in \mathsf{List}_{\varepsilon}(f)$, we get the desired correctness probability.\\

In Line 8 of \Cref{algo:error-reduction}, $L' \leq L(\varepsilon/2)$. So in each iteration of \Cref{algo:error-reduction}, at most $L(\varepsilon/2)$ tuples are added in $T$. Thus over $\ell$ iterations, at most $\bigO(L(\varepsilon/2) \log L(\varepsilon))$ tuples are added.\\

It remains to argue about the query complexity. In a single iteration of \Cref{algo:error-reduction}, we make $s^{k} = s^{B_{d}(L(\varepsilon/2)/\varepsilon)^{c}}$ queries to $f$. There are $\ell = \log L(\varepsilon)$ iterations. From \Cref{thm:comb-bd}, we know that $L(\varepsilon/2) = \bigO_{\varepsilon}(1)$. Thus \Cref{algo:error-reduction} outputs the deterministic algorithms $\Psi_{1},\ldots,\Psi_{L'}$ by making $\bigO_{\varepsilon}(1)$ queries to $f$.\newline
For each deterministic algorithm $\Psi[\mathsf{C},\sigma,Q]$, \Cref{algo:high-agreement} makes $s^{sk} = s^{s B_{d} (L(\varepsilon/2)/\varepsilon)^{c}}$ queries to $f$. From \Cref{thm:comb-bd}, we know that $L(\varepsilon) = \bigO_{\varepsilon}(1)$. Thus each $\Psi_{j}$ makes $\bigO_{\varepsilon}(1)$ queries to $f$. This shows the claimed query complexity.\\
This finishes the proof of \Cref{thm:approx-oracles-list-decoding}.
\end{proof}

\medskip

\section*{Acknowledgments}

We would like to thank the anonymous reviewers of RANDOM 2025 for many valuable comments, including pointers to crucial papers in the literature (specifically  \cite{complexity-symmetric-group}), that significantly improved some of our proofs.

\printbibliography[
heading=bibintoc,
title={References}
] 

\appendix

\section{Tabloids, Polytabloids, Multislices, and Functions}\label{app:tabloids}
For a tableau $t$, tabloid of $t$, denoted by $\set{\mathbf{t}}$ is an equivalence class of tableaux (of the same shape) under the row equivalence relation. See \cite[Definition 2.1.4]{Sagan} for a formal definition.
For a partition $\lambda \in \mathcal{P}(n)$, $\mathrm{Tabloids}(\lambda)$ is a set of tabloids of shape $\lambda$. The symmetric group $\mathrm{Sym}_{n}$ acts naturally on tabloids as follows: For a permutation $\pi \in \mathrm{Sym}_{n}$, $\pi$ acts on a $\set{T} \in \mathrm{Tabloids}(\lambda)$ by permuting the entries of $\set{T}$. For example if $\pi = (125)(46) \in S_{6}$, then\\
\begin{align*}
    (125)(46) \; \ytableausetup{boxsize=normal,tabloids} 
\ytableaushort{ 123,45,6 } \; = \; \ytableausetup{boxsize=normal,tabloids} 
\ytableaushort{ 253,61,4 }
\end{align*}

\paragraph{\underline{Tabloids and multislice}}In the remaining section, we will always use $\lambda$ to denote a partition such that $\lambda \trianglerighteq \mu$, where $\mu = (n/s,\ldots,n/s)$. Note that $\ell(\lambda) \leq s$. We will use the convention that $\lambda$ has exactly $s$ many parts, where we append a $\lambda$ with fewer than $s$ parts with $0$'s.\\

We now observe that $\mathrm{Tabloids}(\lambda)$ and $\sgrid^{n}_{\lambda}$ are in bijection, as follows. For any tabloid $\set{\mathbf{t}} \in \mathrm{Tabloids}(\lambda)$, it corresponds to the point $\mathbf{a} \in \sgrid^{n}_{\lambda}$ where,
\begin{align*}
    a_{j} = i \quad \text{ if } j \in (i+1)^{th} \; \text{ row of } \; \set{\mathbf{t}}, && \text{ for all } \; j \in [n].
\end{align*}
Similarly, for any point $\mathbf{a} \in\sgrid^{n}_{\lambda}$, we get a corresponding tabloid $\set{\mathbf{t}} \in \mathrm{Tabloids}(\lambda)$ where for every $j \in [n]$, the $(i+1)^{th}$ row of $\set{\mathbf{t}}$ contains $j$ if $a_{i} = j$. In simple words, the entries in the $(i+1)^{th}$ row of $\set{T}$ correspond to the coordinates which are $i$. Following is an example for $n = 9$ and $\lambda = (4,3,2)$:
\begin{align*}
    001210201 \quad \leftrightarrow \quad \ytableausetup{boxsize=normal,tabloids} 
    \begin{ytableau}
        1 & 2 & 3 & 4 \\
        5 & 6 & 7 \\
        8 & 9
    \end{ytableau}
\end{align*}

\noindent
For a tableau $\mathbf{t}$, a \emph{polytabloid} for $\mathbf{t}$, denoted by $e_{\mathbf{t}}$ is a linear combination of tabloids obtained by permuting the columns of $\mathbf{t}$. See \cite[Definition 2.3.2]{Sagan} for a formal definition. Using the above bijection, it is easy to see that for every tableau $\mathbf{t}$, the associated polytabloid $e_{\mathbf{t}}$ is a function on $\sgrid^{n}_{\lambda}$.

\section{Subgrid Sampling Lemma}\label{app:sampling}

Here we give the proof of the subgrid sampling lemma from~\Cref{sec:prelims}.

\begin{proof}[Proof of~\Cref{lemma:sampling-subgrid}]
    The proof is an application of the second moment method with a consequence of the following hypercontractivity theorem (\Cref{thm:hypercontractivity}) being used to bound the variance.

    \begin{theorem}[{\protect \cite[Section 10.3]{odonnellbook}}]\label{thm:hypercontractivity}
Let $E \subseteq \mathbb{Z}_s^{n}$ be a subset of density $\delta$, i.e. $|E|/s^{n} = \delta$. Let $q \geq 2$. Then for any $0 \leq |\rho| \leq (1/(q-1)) \cdot (1/s)^{1 - 2/q}$,
\begin{align*}
    \Pr_{\substack{\mathbf{x} \sim \mathbb{Z}_s^{n} \\ \mathbf{y} \sim \mathcal{N}_{\rho}(\mathbf{x})}}[\mathbf{x} \in E \; \text{and} \; \mathbf{y} \in E] \leq \delta^{2 - 2/q}.
\end{align*}
\end{theorem}

    More formally, for each $\mathbf{y}\in \mathbb{Z}_s^k$, let $Z_{\mathbf{y}}\in \{0,1\}$ be the indicator random variable that is $1$ exactly when $x(\mathbf{y})\in T.$ Let $Z$ denote the sum of all $Z_{\mathbf{y}}$ ($\mathbf{y}\in \mathbb{Z}_s^k)$. The statement of the lemma is equivalently stated as 
    \begin{equation}
        \label{eq:sampling-subcube}
        \Pr\left[\left|Z - \mu\cdot s^k\right| \geq \varepsilon\cdot s^k \right] < \eta
    \end{equation}
    for $k$ as specified above.

    Since each $x(\mathbf{y})$ is uniformly distributed over $\mathbb{Z}_s^n$, it follows that each $Z_{\mathbf{y}}$ is a Bernoulli random variable that is $1$ with probability $\mu$. In particular, the mean of $Z$ is $\mu\cdot s^k$. 

    We now bound the variance of $Z$. Let $I_{\gamma}$ be the interval $[\frac{(1-\gamma)(s-1)}{s},\frac{(1+\gamma)(s-1)}{s}]$ where $\gamma \leq 1/(s-1)$. We have
    \begin{align}
        \mathrm{Var}(Z) 
        &= \sum_{\by, \by'} \mathrm{Cov}(Z_{\by}, Z_{\by'}) \nonumber \\
        &= \sum_{\by, \by': \delta(\by,\by') \in I_{\gamma}} \mathrm{Cov}(Z_{\by}, Z_{\by'}) + \sum_{\by, \by': \delta(\by,\by') \notin I_{\gamma}} \mathrm{Cov}(Z_{\by}, Z_{\by'}) \nonumber \\
        &\leq \sum_{\by, \by': \delta(\by,\by') \in I_{\gamma}} \mathrm{Cov}(Z_{\by}, Z_{\by'}) + \sum_{\by, \by': \delta(\by,\by') \notin I_{\gamma}} 1 \nonumber \\
        &\leq \sum_{\by, \by': \delta(\by,\by') \notin I_{\gamma}} \mathrm{Cov}(Z_{\by}, Z_{\by'}) + s^{2k} \cdot \exp(- \Omega(\gamma^2 \cdot (k(s-1)/s))). \label{eq: variance calculation sampling lemma}
    \end{align}
    where the final inequality is an application of the Chernoff bound. On the other hand, for any $\mathbf{y}, \mathbf{y}'$ such that $\delta(\mathbf{y}, \mathbf{y}') \in I_{\gamma}$, we have seen above that the pair $(x(\mathbf{y}),x(\mathbf{y}'))$ have the same distribution as a pair of random variables $(\mathbf{z},\mathbf{z}')$ where $\mathbf{z}$ is chosen uniformly at random from $\mathbb{Z}_s^n$ and $\mathbf{z}'$ is sampled from the distribution $\mathcal{N}_\rho(\mathbf{z}),$ where  $\rho = 1-\frac{s\delta(\mathbf{y},\mathbf{y}')}{s-1} \in [-\gamma, \gamma]$. Thus $|\rho| \leq \gamma$.
    

    Choose $\gamma$ such that $\gamma \leq 1/(s-1)$ and
    \begin{align*}
        C_1 \sqrt{\frac{s \log k}{(s-1)k}} \leq \gamma \leq  \min\left\{\frac{1}{4}, \frac{1}{(k/\log k)^{1/4}} \cdot \frac{1}{s} \right\},
    \end{align*}
    for a large enough constant $C_1$. Such a $\gamma$ exists since $k \geq B\cdot s^4\log s$ for a large constant $B$. 
    
     Set $q = (k\log k)^{1/4}.$ From \Cref{thm:hypercontractivity}, and since $\gamma \leq 1/4$, for $(\mathbf{y},\mathbf{y}')$ satisfying $\delta(\by, \by') \in I_{\gamma}$ we have
    \begin{align*}
        \mathrm{Cov}(Z_{\mathbf{y}}, Z_{\mathbf{y}'}) &= 
        \Pr[x(\mathbf{y})\in T \text{ and } x(\mathbf{y}')\in T] -\mu^2\\
        &\leq \mu^{2 - 2/q} - \mu^2 \\
        &\leq \min\{\mu^{1.5}, \mu^2\cdot (\exp(O((1/q)\cdot \log (1/\mu))-1)\}.
    \end{align*}
    Plugging into~\Cref{eq: variance calculation sampling lemma} we get the following inequalities:
    \begin{align*}
    \mathrm{Var}(Z) &\leq  s^{2k}\cdot \mu^{1.5} + s^{2k} \cdot \frac{1}{k} \leq s^{2k}\cdot O\left(\frac{1}{k}\right)\ \ \ \left(\text{if $\mu\leq \frac{1}{k}$}\right)\\
    \mathrm{Var}(Z) &\leq  s^{2k}\cdot \mu^2 \cdot O\left(\left(\frac{\log k}{k}\right)^{1/4}\cdot \log(1/\mu)\right) + s^{2k} \cdot \frac{1}{k} \leq s^{2k}\cdot O\left(\left(\frac{\log k}{k}\right)^{1/4}\right)\ \ \ \left(\text{if $\mu > \frac{1}{k}$}\right)
    \end{align*}
    where we used the fact that $e^x \leq 1+2x$ for $|x|\leq 1/2$ for the first inequality and the fact that $\mu\leq 1$ for the second. 

    Finally, using Chebyshev's inequality, we get
    \begin{align*}
       \Pr\left[\left|Z - \mu\cdot s^k\right| \geq \varepsilon\cdot s^k \right] 
       &= \Pr_{\mathbf{a},h}\left[\left|Z - \E[Z]\right| \geq \varepsilon\cdot s^k \right] \\
       &\leq \frac{\mathrm{Var}(Z)}{\varepsilon^2 s^{2k}}
       \leq \frac{1}{\varepsilon^2}\cdot 
       O\left(\left(\frac{\log k}{k}\right)^{1/4} \right)
       < \eta
    \end{align*}
    using the lower bound on $k$ in the statement of the lemma.
\end{proof}

\section{Local Correction}\label{sec:lc}

In this section, we show that the family of junta-sums can be locally corrected up to error approaching half the distance of the underlying code, i.e., we prove~\Cref{thm:local-correction}:\\

\localcorr*

Similar to the prior work on local correction of low-degree over the Boolean cube~\cite{ABPSS25}, we divide the proof into two main steps: 

\begin{itemize}
    \item {\bf Error reduction:} In this step, we give a way of reducing the error of the oracle $f:[s]^n \to G$ from $1/(2s^d)-\varepsilon$ to $\varepsilon_1$ for any given $\varepsilon_1 \le 1/\Omega_{s,d}(\log n)^d$, by making $q_1 = \widetilde{O}_{\varepsilon}(1)$ queries to $f$. In particular, there exists a $\widetilde{O}_{\varepsilon}(1)$ query algorithm $\A$ such that, when given as oracle $f:[s]^n \to G$ such that $\delta(f,P) \le 1/(2s^d)-\varepsilon$ for some $P\in \J_d([s]^n, G)$, it satisfies
    \begin{align*}
        \Pr[\A^f(\x) \ne P(\x)] \le \varepsilon_1,
    \end{align*} where the above probability is both over the randomness of $\A$ and ${{\bf x}\sim [s]^n}$ is independently and uniformly chosen.
    
    \item {\bf Correction in low-error regime:} Here, we now assume access to a {\em randomized} oracle $f':[s]^n \to G$ such that $\Pr[f'(\x) \ne P(\x)] \le \varepsilon_1$ for $\x \sim [s]^n$ for some $P\in \J_d([s]^n, G)$, and design a $q_2 = O_{s,d}(1/\varepsilon_1)$ query algorithm $\A'$ such that for every $\x \in [s]^n$, we have 
    \begin{align*}
        \Pr[\A'^{f'}(\x) \ne P(\x)] \le 1/4. 
    \end{align*}
\end{itemize}

    Hence, composing the algorithms $\A$ and $\A'$, we get a local corrector for $f$ that uses at most  $q_1 \cdot q_2 = \widetilde{O}_{\varepsilon}(\log n)^d$ queries. For the case of groups with small order, we follow the same line, except we change the threshold $\varepsilon_1$ to be at most $1/\Omega_{M,\varepsilon}(1)$, resulting in $q_1 = q_2 =  O_{M,\varepsilon}(1)$. This would then finish the proof of~\Cref{thm:local-correction}.

    While the error reduction procedure closely follows similar ideas as for the Boolean cube ($s=2$) from prior work, the low-error regime needs some changes. We give the proofs for error reduction in~\Cref{sec:error-redn}, and for the low-error local corrector in~\Cref{sec:low-error}. For the remainder of the section, we fix $G$ to be an arbitrary Abelian group and assume that $s\ge 2$ (as $\J_d([s]^n,G)$ is a trivial family otherwise).

\subsection{Error Reduction} \label{sec:error-redn}

The main goal of this subsection is to prove the following:\\

\begin{lemma}[{\bf Error reduction}]\label{lem:error-redn}
    For every $\varepsilon_1 = 1/\Theta_{s,d}(\log n)^d$, there exists a $q_1 = \widetilde{O}_{s,d,\varepsilon}(1)$ query algorithm $\A$ such that for every $f:[s]^n \to G$ satisfying $\delta(f,P) \le 1/(2s^d) - \varepsilon$ for some $P\in \J_d([s]^n, G)$, the following holds:
    \begin{align*}
        \Pr[\A^f(\x) \ne P(\x)] \le \varepsilon_1,
    \end{align*}
    where the probability is over a uniformly random $\x \sim [s]^n$, and an independent choice of the randomness of $\A$.
\end{lemma}

We will proceed in an almost identical way as done by \cite{ABPSS25} with a natural extension of the notion of a {\em subcube} from $s=2$ (i.e., Boolean cube) to general $s$.   
We show the following two key lemmas: the first one reduces the error from a small enough constant to ``sub-constant'' and the second one reduces it from $1/(2s^d)-\varepsilon$ to a small enough constant. \\

\begin{restatable}[{\bf Reduction from small constant to sub-constant error}]{lemma}{redsmalltosubconsterr}\label{lem:error-reduction-small-constant-main}
Fix any Abelian group $G$, any $s\geq 2,$ and any positive integer $d$. The following holds for $\delta < 1/s^{\bigO(d)}$ and $K = s^{\bigO(d)}$ where the $\bigO(\cdot)$ hides a large enough absolute constant. For any $\eta, \delta$, where $\eta < \delta$, there exists a randomized algorithm $\mathcal{A}$ with the following properties: Let $f: \mathbb{Z}_s^{n} \to G$ be a function and let $P: \mathbb{Z}_s^{n} \to G$ be a junta-degree-$d$ function such that $\delta(f,P) \leq \delta$, and let $\mathcal{A}^{f}$ denote that $\mathcal{A}$ has oracle access to $f$. Then,
\begin{align*}
    \Pr[\delta(\mathcal{A}^{f}, P) > \eta] < 1/10,
\end{align*}
where the above probability is over the internal randomness of $\mathcal{A}^{f}$. Further, for every $\mathbf{x} \in \Boo^{n}$, $\mathcal{A}^{f}$ makes $K^{T}$ queries to $f$ and $T = \bigO\paren{  \log\paren{ \dfrac{\log(1/\eta)}{\log(1/\delta)} }  }  $.
\end{restatable}

We now state the second key error reduction lemma.\\

\begin{restatable}[{\bf Reduction to small constant error}]{lemma}{redsmallconsterr}\label{lem:error-reduction-main}
Fix any Abelian group $G$, any integer $s\geq 2$, and a positive integer $d$. For any $\eta, \delta$, where $\eta < \delta$ and $\delta < 1/(2\cdot s^d) - \varepsilon$ for $\varepsilon > 0$, there exists a randomized algorithm $\mathcal{A}$ with the following properties: Let $f: \mathbb{Z}_s^n \to G$ be a function and let $P: \mathbb{Z}_s^n \to G$ be a junta-degree $d$ function such that $\delta(f,P) \leq \delta$, and let $\mathcal{A}^{f}$ denotes that $\mathcal{A}$ has oracle access to $f$,  then
\begin{align*}
    \Pr[\delta(\mathcal{A}^{f}, P) > \eta] < 1/10,
\end{align*}
where the above probability is over the internal randomness of $\mathcal{A}$, and for every $\mathbf{x} \in \mathbb{Z}_s^n$, $\mathcal{A}^{f}$ makes $s^{k}$ queries to $f$, where $k = \poly(\frac{1}{\varepsilon},\frac{1}{\eta},s)$.
\end{restatable}

We prove the first lemma in~\Cref{subsec:sub-const} and the second lemma in~\Cref{subsec:redn-small-const}. Below, we finish the proof of the main error reduction lemma of this section using the above two lemmas. 

\begin{proof}[Proof of~\Cref{lem:error-redn}]
    The proof proceeds in a similar way to~\cite{ABPSS24}: we apply the first step of error reduction (\Cref{lem:error-reduction-main}) with $\eta =\eta_1= O_{s,d}(1)$ being smaller than the value of $\delta$ needed to apply the second step (\Cref{lem:error-reduction-small-constant-main}), i.e., $\delta \le O_{s,d}(1)$. This results in a number of queries which is the product of the number of queries from both the steps. Taking $\eta=\eta_2$ in the second error reduction step (i.e.,~\Cref{lem:error-reduction-small-constant-main}) to be equal to $\varepsilon_1 = 1/\Theta_{s,d}(\log n)^d$, we get that the total number of queries is $ O_s\paren{\frac{1}{\eta_1\varepsilon}} \cdot O_{s,d}(1)^{\log\paren{\log\paren{\frac{\log(1/\eta_2)}{\log(1/\delta)}}}} \le (\log\log n)^{O_{s,d,\varepsilon}(1)}$.
\end{proof}

\subsubsection{Reduction from Small Constant to Sub-Constant Error}\label{subsec:sub-const}

We will show that there is a randomized algorithm $\mathcal{A}^{f}$ that given oracle access to any function $f$ that is $\delta$-close to a junta-degree-$d$ function $P$ (think of $\delta$ as being a small enough constant depending on $d$), has the following property: with high probability over the internal randomness of $\mathcal{A}^{f}$, the function computed by  $\mathcal{A}^{f}$ is $\eta$-close to $P$, where $\eta$ can be much smaller than $\delta$. We restate it formally below.\\

\redsmalltosubconsterr*

In the rest of this subsection, we will prove \Cref{lem:error-reduction-small-constant-main}. The algorithm $\mathcal{A}^{f}$ in \Cref{lem:error-reduction-small-constant-main} will be a recursive algorithm. Each recursive iteration of the algorithm $\mathcal{A}^{f}$ uses the same `base algorithm' $\mathcal{B}$, which will be the core of our error reduction algorithm from small constant error. In the next lemma, we formally state the properties of the base algorithm. 

\begin{lemma}[Base Error Reduction Algorithm]\label{lem:error-reduction-subroutine}
Fix any Abelian group $G$, any integer $s\geq 2,$ and a positive integer $d$. The following holds for $K = s^{O(d)}$.  For any $0 < \gamma < 1$, there exists a randomized algorithm $\mathcal{B}$ with the following properties: Let $g: \mathbb{Z}_s^{n} \to G$ be a function and let $P: \mathbb{Z}_s^{n} \to G$ be a junta-degree-$d$ function such that $\delta(g,P) \leq \gamma$, and let $\mathcal{B}^{g}$ denote that $\mathcal{B}$ has oracle access to $g$, then
\begin{align*}
    \mathbb{E}[\delta(\mathcal{B}^{g}, P)] < O(K^2)\cdot \gamma^{1.5}
\end{align*}
where the above expectation is over the internal randomness of $\mathcal{B}$. Further, for every $\mathbf{x} \in \mathbb{Z}_s^{n}$, $\mathcal{B}^{g}$ makes $K$ queries to $g$.  
\end{lemma}

We defer the construction of the base algorithm and proof of \Cref{lem:error-reduction-subroutine} to later. For now, we assume \Cref{lem:error-reduction-subroutine} and proceed to describe the recursive construction of $\mathcal{A}^{f}$ and prove \Cref{lem:error-reduction-small-constant-main}.

\begin{proof}[Proof of \Cref{lem:error-reduction-small-constant-main}]
Let $\mathcal{B}$ be the algorithm given by \Cref{lem:error-reduction-subroutine}. We define a sequence of algorithms $\mathcal{A}_0^f, \mathcal{A}_1^f, \ldots,$ as follows. 

\begin{algobox}
The algorithm $\mathcal{A}^f_t$ computes a function mapping inputs in $\mathbb{Z}_s^{n}$ along with a uniformly random string from $\{0,1\}^{r_t}$ to a random group element in $G$.

\begin{itemize}
    \item $\mathcal{A}_0^f$ just computes the function $f.$ (In particular, $r_0 = 0$.)
    \item For each $t > 0,$ we inductively define $r_t = r_{t-1}+ r$, where $r$ is the amount of randomness required by the base error reduction algorithm $\mathcal{B}$. On input $\mathbf{x}\in \mathbb{Z}_s^{n}$ and a uniformly random string $\sigma_t$, the algorithm $\mathcal{A}_t^f$ algorithm runs the algorithm $\mc{B}$ on $\mathbf{x}$ using the first $r$ bits of $\sigma_t$ as its source of randomness, and with oracle access to $\mathcal{A}_{t-1}^f$ using the remaining $r_{t-1}$ bits of $\sigma_t$ as randomness.
\end{itemize}

The algorithm $\mathcal{A}^f$ will be $\mathcal{A}_T^f$ for $T = C\cdot \log\paren{ \dfrac{\log(1/\eta)}{\log(1/\delta)} } $ where $C$ is a large enough absolute constant chosen below. 
\end{algobox}

\textbf{Query complexity:} An easy inductive argument shows that $\mathcal{A}^f$ makes at most $K^T$ queries to $f.$

\textbf{Error probability:} We now analyze the error made by the above algorithms. We will argue inductively that for each $t\leq T$ and  $\delta_t := \delta^{(1.1)^t}$, we have 
\begin{equation}
    \label{eq:algoAindn}
    \Pr_{\sigma_t}[\, \underbrace{\delta(\mathcal{A}^{f}_t(\cdot, \sigma_t), P) > \delta_t \, }_{:= \, \mathcal{E}_t}] \; \leq \; \sum_{j=1}^t \frac{1}{100^j} \; < \; \frac{1}{10}.
\end{equation}
In the inductive proof, we will need that $\delta_0 = \delta <s^{-C_1\cdot d}$ for a large enough absolute constant $C_1.$

We now proceed with the induction. The base case ($t = 0$) is trivial as $\delta(\mathcal{A}^{f}_t, P) = \delta_0$ by definition.

Now assume that $t > 1$. We decompose the random string $\sigma_t$ into its first $r$ bits, denoted $\sigma$, and its last $r_{t-1}$ bits, denoted $\sigma_{t-1}.$ We bound the probability in \Cref{eq:algoAindn} as follows. (Note that the event $\mathcal{E}_{t-1}$ below only depends on $\sigma_{t-1}.$)
\begin{equation}
\label{eq:algoAindn-t-1}
\Pr_{\sigma_t}[\mathcal{E}_t] \, \leq \, \Pr_{\sigma_{t-1}}[\mathcal{E}_{t-1}] + \Pr_{\sigma_t}[\mathcal{E}_t\ |\ \neg\mathcal{E}_{t-1}] \, \leq \, \sum_{j=1}^{t-1} \frac{1}{100^j} + \Pr_{\sigma_t}[\mathcal{E}_t\ |\ \neg\mathcal{E}_{t-1}]
\end{equation}
where we used the induction hypothesis for the second inequality. To bound $\Pr_{\sigma_t}[\mathcal{E}_t\ |\ \neg\mathcal{E}_{t-1}]$, fix any choice of $\sigma_{t-1}$ so that $\neg\mathcal{E}_{t-1}$ holds, i.e. so that $\delta(\mathcal{A}^{f}_{t-1}, P) \leq \delta_{t-1}$. By the guarantee on $\mathcal{B}$, i.e. \Cref{lem:error-reduction-subroutine}, we know that
\[
\mathbb{E}_{\sigma}[\delta(\mathcal{A}^{f}_t(\cdot, \sigma_t), P)] < \bigO(K^2)\cdot \gamma^{1.5},
\]
where $\gamma = \delta(\mathcal{A}^f_{t-1}(\cdot, \sigma_{t-1}),P)$. Substituting it above, we get,
\begin{align*}
    \mathbb{E}_{\sigma}[\delta(\mathcal{A}^{f}_t(\cdot, \sigma_t), P)] \, \leq \, \bigO(K^2) \cdot \delta_{t-1}^{1.5} \, \leq \, \delta_{t-1}^{1.25}
\end{align*}
where for the final inequality, we use the fact that
\[
\bigO(K^2)\cdot \delta_{t-1}^{0.25} \leq \bigO(K^2)\cdot \delta_0^{0.25} \leq 1
\]
as long  as $\delta_0 = \delta \leq s^{-C_1 d}$ for a large enough constant $C_1.$ Continuing the above computation, we see that by Markov's inequality
\[
\Pr_{\sigma}[\mathcal{E}_t] \, \leq \, \frac{\delta_{t-1}^{1.25}}{\delta_t} \; = \; \delta^{\Omega((1.1)^t)} \, \leq \, \frac{1}{100^t}
\]
where the final inequality holds for all $t$ as long as $\delta \leq s^{-C_1 d}$ for a large enough constant $C_1.$ Since this inequality holds for any choice of $\sigma_{t-1}$ so that $\neg\mathcal{E}_{t-1}$ holds, we can plug this bound into \Cref{eq:algoAindn-t-1} to finish the inductive case of \Cref{eq:algoAindn}.

Setting $T = C\cdot \log\paren{ \dfrac{\log(1/\eta)}{\log(1/\delta)} } $ for a large enough constant $C$, we see that $\delta_T < \eta.$ In this case, \Cref{eq:algoAindn} implies the required bound on the error probability of $\mathcal{A}^f.$
\end{proof}
Thus we have shown so far that given the base algorithm $\mathcal{B}$, we do get an error reduction algorithm from small constant error to error $\bigO(1/\log n)$. Now it remains to describe the base error reduction algorithm.
In the next subsection, we describe the base algorithm $\mathcal{B}$ and prove \Cref{lem:error-reduction-subroutine}.

\paragraph{The base algorithm and its analysis.}

In the rest of this subsection, we prove \Cref{lem:error-reduction-subroutine}, which will then complete the proof of \Cref{lem:error-reduction-small-constant-main}. Before we describe $\mathcal{B}$, we will define an \textit{error reduction gadget}.\\

\begin{definition}[Error-reduction Gadget for $\mathcal{J}_d$]\label{defn:error-reduction-gadget}
 For $\rho\in (0,1/(s-1))$, an $(\rho,q)$-error reduction gadget for $\mathcal{J}_d$ is a distribution $\mathcal{D}$ over $(\mathbb{Z}_s^n)^{q}$ satisfying the following two properties:
\begin{enumerate}
    \item There exists $c_{1}, \ldots, c_{q}\in \mathbb{Z}$ such that for any $(\mathbf{y}^{(1)}, \ldots, \mathbf{y}^{(q)}) \in \mathrm{supp}(\mathcal{D})$, the following holds true for each $P\in \mathcal{J}_d$ and each $\mathbf{a}\in \mathbb{Z}_s^n$
    \begin{equation}
    \label{eq:error-reduction-gadget}
    P(\mathbf{a}) = c_{1}P(\mathbf{a}+ \mathbf{y}^{(1)}) + \ldots + c_{q}P(\mathbf{a}+\mathbf{y}^{(q)})
\end{equation}
where the $\mathbf{a} + \mathbf{y}^{(i)}\in \mathbb{Z}_s^{n}$ is computed via a co-ordinate-wise sum modulo $s.$

    \item For any $i \in [q]$, the co-ordinates of $\mathbf{y}^{(i)}$ are i.i.d. random variables in $\mathbb{Z}_s$ that take the value $0$ with probability $p_i$ such that
    \[
    p_i\in \left[\frac{1}{s}-\rho\cdot \left(1-\frac{1}{s}\right), \frac{1}{s}+\rho\cdot \left(1-\frac{1}{s}\right)\right]
    \]
    and each non-zero value in $\mathbb{Z}_s$ with probability $\frac{1-p_i}{s}.$ We call such distributions \emph{$\rho$-noisy distributions} over $\mathbb{Z}_s$.
\end{enumerate}
\end{definition}

To prove \Cref{lem:error-reduction-subroutine}, we need an error-reduction gadget for $\mathcal{J}_d$, the space of junta-degree-$d$ functions over a group $G.$ This is given by the following lemma.\\

\begin{lemma}[Constructing an error-reduction gadget for $\mathcal{J}_d$]
\label{lem:error-reduction-gadget}
    Fix any Abelian group $G$, $s\geq 2$ and any $\rho\in (0,1/(s-1)).$ Then $\mathcal{J}_d(\mathbb{Z}_s^n,G)$ has a $(\rho,q)$-error-reduction gadget where $q = ((1/\rho) + s)^{O(d)}.$
\end{lemma}

Assuming the above lemma, we first finish the proof of \Cref{lem:error-reduction-subroutine}.  For this, we will need the following technical claim.

\begin{claim}
    \label{clm:convolution}
    Let $y,z$ be independent random variables taking values in $\mathbb{Z}_s$ such that their distributions are $\rho_1$-noisy and $\rho_2$-noisy respectively. Then, $y-z$ is $(\rho_1\cdot \rho_2)$-noisy.
\end{claim}

\begin{proof}
    Let $\mathcal{D}_y$ and $\mathcal{D}_z$ denote the probability distributions of $y$ and $z$ respectively, which we think of as elements of $\mathbb{R}^{s}$. 
    
    We note that the condition that $y$ is $\rho_1$-noisy can be restated as 
    \[
    \mathcal{D}_y = \varepsilon_1 \cdot \delta_0 + (1-\varepsilon_1)\cdot \mathcal{U}
    \]
    where $\mathcal{U}$ denotes the uniform distribution over $\mathbb{Z}_s$, $\delta_0$ denotes the distribution that places all its mass on $0$, and $\varepsilon_1$ is a (possibly negative) number satisfying $|\varepsilon_1|\leq \rho_1.$ 
    
    A similar fact also holds for the random variable $-z\in \mathbb{Z}_s$, since $z$ being $\rho_2$-noisy implies the same for $-z.$

    Now, the distribution $\mathcal{D}$ of $y-z$ is the convolution $\mathcal{D}_y * \mathcal{D}_z$ giving us
    \[
    \mathcal{D} = (\varepsilon_1 \cdot \delta_0 + (1-\varepsilon_1)\cdot \mathcal{U}) * (\varepsilon_2 \cdot \delta_0 + (1-\varepsilon_2)\cdot \mathcal{U}) = \varepsilon_1\varepsilon_2\cdot \delta_0 + (1-\varepsilon_1\varepsilon_2)\cdot \mathcal{U}
    \]
    where the latter equality is by distributivity and the fact that the convolution of $\mathcal{U}$ with any distribution is $\mathcal{U}.$ 
    
    Since $|\varepsilon_1|\leq \rho_1$ and $|\varepsilon_2|\leq \rho_2,$ we have the claim.
\end{proof}

In the algorithm, we use the error-reduction gadget to correct the junta-sum at a \emph{random point} $\mathbf{a}\in \{0,1\}^n$. This process is likely to give the right answer except with probability $q\gamma$ since, after shifting, each query is now \emph{uniformly} distributed and hence the chance that any of the queried points is an error point of $g$ is at most $\gamma$. We reduce the error by repeating this process three times and taking a majority vote. To analyze this algorithm, we need to understand the probability that two iterations of this process both evaluate $g$ at an error point. We do this using hypercontractivity (more specifically \Cref{thm:hypercontractivity}).

\begin{proof}[Proof of \Cref{lem:error-reduction-subroutine}]
    Let $\mathcal{D}$ be a $(1/10s,q)$-error-reduction gadget as given by \Cref{lem:error-reduction-gadget}. The algorithm $\mathcal{B}$, given oracle access to $g:\mathbb{Z}_s^n\rightarrow G$ and $\mathbf{a}\in \mathbb{Z}_s^n$, does the following.

\begin{itemize}
    \item Repeat the following three times independently. Sample $(\mathbf{y}^{(1)}, \ldots, \mathbf{y}^{(q)})$ from $\mathcal{D}$ and compute 
    \[
    c_1 g(\mathbf{a} + \mathbf{y}^{(1)}) +  \cdots + c_qg(\mathbf{a} + \mathbf{y}^{(q)})
    \]
    where $c_1,\ldots, c_q$ are the coefficients corresponding to the error-reduction gadget, and the sums $\mathbf{a} + \mathbf{y}^{(i)}$ are computed in $\mathbb{Z}_s^n$.
    \item Output the plurality among the three group elements $b_1,b_2,b_3$ computed above.
\end{itemize}

The number of queries made by the algorithm is $K = O(q) = (10s + s)^{O(d)} = s^{O(d)}$ as claimed. So it only remains to analyze $\delta(\mathcal{B}^g,P)$. From now on, let $\mathbf{a}$ be a uniformly random input in $\{0,1\}^n.$

For $i\in \{1,2,3\},$ let $\mathcal{E}_i$ denote the event that $b_i \neq P(\mathbf{a}).$ We have
\[
\E[\delta(\mathcal{B}^g,P)] = \Pr[\mathcal{B}^g(\mathbf{a})\neq P(\mathbf{a})] \leq \Pr[\mathcal{E}_1 \wedge \mathcal{E}_2] + \Pr[\mathcal{E}_2 \wedge \mathcal{E}_3] + \Pr[\mathcal{E}_1 \wedge \mathcal{E}_3].
\]
It therefore suffices to show that each of the three terms in the final expression above is at most $O(q^2)\cdot \gamma^{1.5}.$

Without loss of generality, consider the event $\mathcal{E}_1\wedge \mathcal{E}_2$. Let $(\mathbf{y}^{(1)}, \ldots, \mathbf{y}^{(q)})$ and $(\mathbf{z}^{(1)}, \ldots, \mathbf{z}^{(q)})$ be the two independent samples from $\mathcal{D}$ in the two corresponding iterations.

It follows from \Cref{eq:error-reduction-gadget} that the algorithm correctly computes $P(\mathbf{a})$ in the first iteration as long as none of the queried points lie in the set $T$ of points where $g$ and $P$ differ. A similar statement also holds for the second iteration. This reasoning implies that
\begin{equation}
\label{eq:E1E2}
\Pr[\mathcal{E}_1 \wedge \mathcal{E}_2]\leq \sum_{i,j = 1}^q \Pr[\underbrace{\mathbf{a} + \mathbf{y}^{(i)}}_{\mathbf{u}^{(i)}}\in T \wedge \underbrace{\mathbf{a}+\mathbf{z}^{(j)}}_{\mathbf{v}^{(j)}}\in T].
\end{equation}

We bound the latter expression using \Cref{thm:hypercontractivity}. 

Fix $i,j\in [q]$. Note that for every fixing of $\mathbf{y}^{(i)}$, the vector $\mathbf{u}^{(i)}$ is distributed uniformly over $\mathbb{Z}_s^n$ (because $\mathbf{a}$ is uniform over $\mathbb{Z}_s^n$). In particular, this implies that $\mathbf{u}^{(i)}$ is uniformly distributed and moreover that $\mathbf{u}^{(i)}$ and $\mathbf{y}^{(i)}$ are independent random variables. 

Note, moreover, that $\mathbf{y}^{(i)}$ is independent of $\mathbf{z}^{(j)}$ and their entries are i.i.d. random variables over $\mathbb{Z}_s$ that are $\rho$-noisy. By \Cref{clm:convolution} above, we see that the entries of $\mathbf{y}^{(i)}-\mathbf{z}^{(j)}$ are i.i.d. and $\rho^2 = (1/100 s^2)$-noisy.

This means that $\mathbf{v}^{(j)} = \mathbf{u}^{(i)}+\mathbf{y}^{(i)}-\mathbf{z}^{(j)}$ is drawn from the noise distribution $\mathcal{N}_{\sigma}(\mathbf{u}^{(i)})$, where the parameter $\sigma \leq 1/100s^2$. Using \Cref{thm:hypercontractivity} with $q=4$, we have
\[
\prob{}{\mathbf{u}^{(i)}\in T \wedge \mathbf{v}^{(j)}\in T} \leq \gamma^{1.5}.
\]
Plugging this into \Cref{eq:E1E2} implies the required bound on the probability of $\mathcal{E}_1 \wedge \mathcal{E}_2.$ This concludes the analysis of $\mathcal{B}.$
\end{proof}

We now show how to construct the error-reduction gadget and prove \Cref{lem:error-reduction-gadget}. This requires the following claim (implied e.g. by M\"{o}bius inversion) that shows that any junta-degree-$d$ function over $\{0,1\}^n$ (even with group coefficients) can be interpolated from its values on a Hamming ball of radius $d.$ For completeness, we give a short proof.

\begin{lemma}\label{lem:mobius}
Fix $d \in \mathbb{N}$. For any natural number $m \geq d$ and any Hamming ball $B$ of radius $d$,
\[
P(0^m) = \sum_{\mathbf{b}\in B} \alpha_{\mathbf{b}} P(\mathbf{b})
\]
where the $\alpha_{\mathbf{b}}$ are integer coefficients.
\end{lemma}

\begin{proof}
    Assume that 
    \[
   P(\mathbf{x}) = \sum_{\substack{\mathbf{a} \in \mathbb{Z}_s^n \\ \#\mathbf{a} \le d}} g_{\mathbf{a}} \cdot \prod_{i\in [n]:~a_i \ne 0} \delta_{a_i}(x_i).
    \]
    By M\"{o}bius inversion, we know that
    \[
    g_{\mathbf{a}} = \sum_{J\subseteq I} (-1)^{|I\setminus J|} P(1_J\circ \mathbf{a})
    \]
    where $1_J\in \{0,1\}^m$ denotes the indicator vector of set $J$ and $\circ$ denots co-ordinate-wise product. Putting the above equalities together gives us 
    \[
    P(\mathbf{x}) = \sum_{\#\mathbf{b}\leq d} \alpha_{\mathbf{b},\mathbf{x}}' P(\mathbf{b})
    \]
    for suitable integer coefficients $\alpha'_{\mathbf{b},\mathbf{x}}.$ 
    
    Now, assume $B$ is the Hamming ball of radius $d$ around the point $\mathbf{c}\in \mathbb{Z}_s^m$. Replacing $\mathbf{x}$ by $\mathbf{x}+ \mathbf{c}$ in $P$ does not increase the junta-degree of the function (since each co-ordinate of $\mathbf{x}+ \mathbf{c}$ depends only on a single co-ordinate of $\mathbf{x}$). Applying this substitution above yields
    \[
    P(\mathbf{x}+ \mathbf{c}) = \sum_{\#\mathbf{b}\leq d} \alpha_{\mathbf{b},\mathbf{x}}' P(\mathbf{b}+ \mathbf{c}) = \sum_{\mathbf{b}\in B}\alpha_{\mathbf{b},\mathbf{x}} P(\mathbf{b}).
    \]
    Setting $\mathbf{x} = -\mathbf{c}$ yields the statement of the lemma.
\end{proof}

We end this section by completing the proof of \Cref{lem:error-reduction-gadget}.

\begin{proof}[Proof of \Cref{lem:error-reduction-gadget}]
    The idea is to apply \Cref{lem:mobius} on a random subcube, as defined in \Cref{defn:random-embedding}.

    More precisely, let $k,d$ be positive integers such that $k$ is divisible by $s$ and $k\geq s\cdot d$. Let $\mathbf{a}\in \mathbb{Z}_s^n$ be arbitrary. For each $i\in [n]$, let $\Pi_i\in \mathsf{Sym}[\mathbb{Z}_s]$ be chosen uniformly \emph{from among bijections that map $0$ to $a_i$,} and let $\mathbf{\Pi}$ denote $(\Pi_1,\ldots, \Pi_n)$. Also assume that $h:[n]\rightarrow [k]$ is chosen uniformly at random. Let $C = C_{\mathbf{\Pi},h}$ be the corresponding subcube of $\mathbb{Z}_s^n$ as defined in \Cref{defn:random-embedding}. Let $Q(y_1,\ldots,y_k)$ denote $P|_C$, the restriction of $P$ to this subcube.

    Fix a Hamming ball $B$ of radius $d$ in $\mathbb{Z}_s^k$ centred at a point $\mathbf{c}$ with exactly $k/s$ many occurrences of $0$. Since $Q$ is a function of junta-degree at most $d$, applying \Cref{lem:mobius} to $Q$ and the ball $B$ yields an equality
    \[
    Q(0^k) = \sum_{\mathbf{b}\in \mathbf{B}} \alpha_{\mathbf{b}} Q(\mathbf{b}).
    \]
    Since $Q$ is a restriction of $P$, the above equality can be rephrased in terms of $P$ as
    \[
    P(x(0^k)) = \sum_{\mathbf{b}\in \mathbf{B}} \alpha_{\mathbf{b}} P(x(\mathbf{b})).
    \]
    From the definition of the cube $C$, it follows that $x(0^k) = \mathbf{a}$ and thus the above gives us an equality of the type desired in an error-reduction gadget (\Cref{eq:error-reduction-gadget}). To finish the proof, we only need to argue that each $x(\mathbf{b})$ has the required distribution.

    Note that for each $\mathbf{b}\in B$, we have
    \[
    x(\mathbf{b}) =  \mathbf{a} + \mathbf{b}'
    \]
    where $\mathbf{b}'$ is the random vector in $\mathbb{Z}_s^n$ that at co-ordinate $i$ takes the random value $\Pi_i(b_{h(i)})$. Since $h$
    is chosen uniformly at random and the $\Pi_i$'s are independent and uniform subject to the constraint that $\Pi_i(0) = a_i$, it follows that the entries of $\mathbf{b}_h$ are independent and the $i$th co-ordinate is a $\mathbb{Z}_s$-valued random variable that takes the value $0$ with probability equal to the proportion of $0$'s in $\mathbf{b}$ (which we denote $\sigma$) and each non-zero value in $\mathbb{Z}_s$ with the probability $(1-\sigma)/(s-1)$. In other words, the entries of $\mathbf{b}_h$ are $\rho$-noisy as long as 
    \[
    \sigma\in \left[\frac{1}{s} - \rho\cdot \left(1-\frac{1}{s}\right), \frac{1}{s} + \rho\cdot \left(1-\frac{1}{s}\right)\right].
    \]
    
    To conclude the argument, note that $\mathbf{b}$ is at Hamming distance at most $d$ from $\mathbf{c}$, implying that $\sigma$ is in the range
    \[
    \left[\frac{1}{s}-\frac{d}{k},\frac{1}{s}+\frac{d}{k}\right].
    \]
    Setting $k$ to be the smallest multiple of $s$ larger than $2d/\rho$ gives us the desired value for the parameter of the distribution of $\mathbf{b}$.
    
    Finally, the number of queries $q$ made by the error-reduction gadget is dictated by the size of a Hamming ball in $k = O(d/\rho)$ dimensions. This can be bounded by
    \[
    \binom{k}{d}\cdot s^d \leq (k/d)^{O(d)}\cdot s^d \leq \left(1/\rho + s\right)^{O(d)}\cdot s^d = \left(1/\rho +s \right)^{O(d)}.
    \]
    
    It follows that we have a $(\rho,((1/\rho)+ s)^{O(d)})$-error-reduction gadget.
\end{proof}

\subsubsection{Reduction to Small Constant Error}\label{subsec:redn-small-const}

Now, we will show that there is a randomized algorithm $\mathcal{A}$ that given oracle access to any function $f$ that is $\delta$-close to a low junta-degree function $P$ (think of $\delta$ to be very close to half the minimum distance, i.e. $1/(2\cdot s^d) - \varepsilon$ for junta-degree $d$), has the following property: with high probability over the internal randomness of $\mathcal{A}$, $\mathcal{A}^{f}$ is $\eta$-close to $P$, where $\eta$ is much smaller than $\delta$. We recall it formally below.

\redsmallconsterr*

Now we state an algorithm $\mathcal{A}^{f}$ below and use it to prove \Cref{lem:error-reduction-main}..

\begin{algobox}
\begin{algorithm}[H]

\DontPrintSemicolon

\KwIn{$f$ and $\mathbf{a} \in \mathbb{Z}_s^n$}

Choose $k = (s/(\varepsilon \eta))^{10}$ \;
Sample a uniformly random $h: [n] \to [k]$ \tcp*{$h$ is the internal randomness of $\mathcal{A}^{f}$}
Sample $\Pi_1,\ldots, \Pi_n\in \mathsf{Sym}[\mathbb{Z}_s]$ independently and uniformly at random subject to the condition that $\Pi_i(0) = a_i$ for each $i\in [n].$\;
Construct the cube $\mathsf{C} := C_{\mathbf{\Pi},h}$ according to \Cref{defn:random-embedding} \;
Let $\Tilde{f} := f|_{\mathsf{C}}$ \tcp*{$f|_{\mathsf{C}}$ is the restriction of $f$ to the subcube $\mathsf{C}$} 
Query $\tilde{f}$ on all inputs in $\mathbb{Z}_s^k$ to find the junta-sum $\Tilde{P}$ on $\mathsf{C}$ such that $\delta(\Tilde{f}, \Tilde{P}) < 1/(2\cdot s^d)$ \tcp*{$s^{k}$ queries to $f$}
\If{such a junta-sum $\Tilde{P}$ is found}{
\Return{$\Tilde{P}(0^{k})$} \tcp*{$x(0^{k}) = \mathbf{a}$}
}
\Else{
\Return{$0$} \tcp*{An arbitrary value}
}

\caption{Error Reduction Algorithm $\mathcal{A}^{f}$}
\label{algo:error-reduction}
\end{algorithm}
\end{algobox}

\begin{proof}[Proof of \Cref{lem:error-reduction-main}]
Let $P$ be the (unique) junta-degree $d$ function such that $\delta(f,P) < 1/(2\cdot s^d)$. The junta-degree of $P$ is at most $d$ when $P$ is restricted to $\mathsf{C} = C_{\mathbf{\Pi},h}$. If $\delta(P|_{\mathsf{C}}, \Tilde{f}) < 1/(2\cdot s^d)$, then $\Tilde{P} = P|_{\mathsf{C}}$. In particular, $\Tilde{P}(x(0^{k})) = P(\mathbf{a})$, i.e. the output of the algorithm is correct.

Equivalently, $\mathcal{A}^{f}(\mathbf{a}) = P(\mathbf{a})$ unless $\delta(P|_{\mathsf{C}}, \Tilde{f}) \geq 1/(2\cdot s^d)$. In the next lemma, we will show that with high probability over random $\mathbf{a}$ and $h$, $\delta((P|_{\mathsf{C}}, \Tilde{f}) < 1/(2\cdot s^d).$
\begin{lemma}\label{lemma:error-reduction-main}
Sample $\mathbf{a}$, $\mathbf{\Pi} = (\Pi_1,\ldots, \Pi_n)$ and $h$ as in the algorithm above. Let $\mathsf{C} = C_{\mathbf{\Pi}, h}$ be the subcube of dimension $k$ as described in \Cref{defn:random-embedding}. Then,
\begin{align*}
    \Pr_{\mathbf{a},\mathbf{\Pi},h}[\delta(P|_{\mathsf{C}}, \Tilde{f}) \geq 1/(2\cdot s^d))] < \eta/10
\end{align*}
\end{lemma}
We prove \Cref{lemma:error-reduction-main} below. For now, let us assume \Cref{lemma:error-reduction-main} and finish the proof of \Cref{lem:error-reduction-main}. We have,
\begin{gather*}
    \Pr_{\mathbf{a},\mathbf{\Pi},h}[\delta(P|_{\mathsf{C}}, \Tilde{f}) \geq 1/(2\cdot s^d)] <  \eta/10 \\
    \Rightarrow \mathbb{E}_{h,\mathbf{\Pi}} \left[ \Pr_{\mathbf{a}}[\delta(P|_{\mathsf{C}}, \Tilde{f}) \geq 1/(2\cdot s^d)]\right] <  \eta/10
\end{gather*}
Note that if we fix the internal randomness of $\mathcal{A}^{f}$ (i.e. the random bits used to choose $h,\mathbf{\Pi}$), then $\delta(\mathcal{A}^{f}, P)$ is at most $\Pr_{\mathbf{a}}[\delta(P|_{\mathsf{C}}, \Tilde{f}) \geq  1/(2\cdot s^d))]$, as the algorithm always outputs $P(\mathbf{a})$ correctly when $\delta(P|_{\mathsf{C}}, \Tilde{f}) < 1/(2\cdot s^d)$ . Then from the above inequality, we have,
\begin{align*}
    \mathbb{E}_{h,\mathbf{\Pi}} \, [\delta(\mathcal{A}^{f}, f)] <  \eta/10 \\
    \Rightarrow \Pr_{h,\mathbf{\Pi}}[\delta(\mathcal{A}^{f}, f) > \eta] \leq 1/10 \tag*{(Markov's Inequality)}
\end{align*}
As commented in \Cref{algo:error-reduction}, for each $\mathbf{a} \in \mathbb{Z}_s^n$, $\mathcal{A}^{f}$ makes $s^{k}$ queries to $f$.
\end{proof}
Now we give the proof of \Cref{lemma:error-reduction-main}.
\begin{proof}[Proof of \Cref{lemma:error-reduction-main}]
Let $E$ denote the subset of points in $\mathbb{Z}_{s}^{n}$ where $P$ and $f$ disagree, i.e. $E := \setcond{\mathbf{x} \in \mathbb{Z}_s^{n}}{f(\mathbf{x}) \neq P(\mathbf{x})}$. We know that $|E|/s^{n} \leq 1/(2\cdot s^d) - \varepsilon$. 

The fractional Hamming distance between $P|_{\mathsf{C}}, \Tilde{f}$ is given by the relative size of the set $E\cap \mathsf{C}$ inside $\mathsf{C}.$ Note that since $\mathbf{a}$ is chosen at random and each $\Pi_i$ is chosen at random satisfying $\Pi_i(0) = a_i$ (for each $i\in [n]$), we see that each $\Pi_i$ is indeed a uniformly independent element of $\mathsf{Sym}[\mathbb{Z}_s]$. Hence, the subcube $\mathsf{C} = C_{\mathbf{\Pi},h}$ is a random subcube in the sense of \Cref{defn:random-embedding}.

Applying the sampling lemma from~\Cref{sec:prelims} (i.e., \Cref{lemma:sampling-subgrid}), we get that for $k = (s/(\varepsilon \eta))^{10}$ (we assume without loss of generality that $\varepsilon,\eta$ are small enough for $k$ to satisfy the hypothesis of \Cref{lemma:sampling-subgrid})
\begin{align*}
    \Pr_{\mathbf{a},\mathbf{\Pi}, h}[\delta(P|_{\mathsf{C}}, \Tilde{f}) \geq 1/(2\cdot s^d)] < \eta/10,
\end{align*}
and this completes the proof of \Cref{lemma:error-reduction-main}.
\end{proof}

\subsection{Correction in Low-Error Regime} \label{sec:low-error}

Having just shown how to reduce the error, we will now prove that there is a local correction algorithm in this ``low-error'' regime.\\

\begin{lemma}[{\bf Local correction in low-error regime}]\label{lem:low-error}
    There exists $\varepsilon_1 = 1/\Theta_{s,d}(\log n)^d$ and a $q_2 = {O}_{s,d}(1/\varepsilon_1)$ query algorithm $\A$ such that for every {randomized} oracle $f:[s]^n \to G$ satisfying $\Pr_{\x \sim [s]^n}[f(\x) \ne P(\x)] \le \varepsilon_1$ for some $P\in \J_d([s]^n, G)$, it holds for every $\x \in [s]^n$ that: 
    \begin{align*}
        \Pr[\A^f(\x) \ne P(\x)] \le 1/4.
    \end{align*}
\end{lemma}

Using the above two lemmas, we can finish the proof of the first part of~\Cref{thm:local-correction}.

\begin{proof}[Proof of~\Cref{thm:local-correction} for general Abelian groups]
    The proof follows by applying~\Cref{lem:low-error} with the randomized oracle being $\A^f$ given by~\Cref{lem:error-redn}. This yields a total number of queries of $q_1 \cdot q_2 = \widetilde{O}_{s,d,\varepsilon}(\log n)^d = \widetilde{O}_{\varepsilon}(\log n)^d$ as we can assume that $\varepsilon < \delta_\J/2 = 1/(2s^d)$. 
\end{proof}

Before we prove~\Cref{lem:low-error}, we show the following claim which reduces the problem of local correction of junta-sums to local correction over the Boolean cube but with a biased distribution.\\

\begin{lemma}[{\bf Reduction to correction over biased cube}]\label{clm:local-redn}
    Suppose there exists a $q$ query algorithm $\A$ such that for every randomized oracle $f:\{0,1\}^n \to G$ satisfying $\Pr_{{\bf y}\sim \bern(1/s)^n}[f(\y) \ne P(\y)] \le 10\varepsilon_1$ for some $P\in \J_d(\{0,1\}^n,G)$, it holds that $\Pr[\A^f({\bf 1}) \ne P({\bf 1})] \le 1/4$. 

    Then, there exists a $O(q/\varepsilon_1)$ query algorithm $\A'$ such that for every randomized oracle $f':[s]^n \to G$ satisfying $\Pr_{{\bf x}\sim [s]^n}[f'(\x) \ne P'(\x)] \le \varepsilon_1$ for some $P'\in \J_d([s]^n,G)$, it holds for every $\x \in [s]^n$ that $\Pr[\A'^{f'}(\x) \ne P'(\x)] \le 1/4$. 
\end{lemma}

\begin{proof}
    We design $\A'$ using $\A$. Fix $\x \in [s]^n$ be arbitrarily and sample $\x' \in [s]^n$ but choosing $x'_i \in [s]\setminus \{x_i\}$ uniformly and independently at random. We then define $f:\{0,1\}^n \to G$ as follows: Given $\y \in \{0,1\}^n$, let $\z=\z(\y) \in [s]^n$ be defined by $z_i= x_i$ if $y_i = 1$ and $z_i = x'_i$ otherwise -- then we define $f(\y)$ to be equal to $f'(\y(\z))$; similarly we define $P:\{0,1\}^n\to G$ by $P(\y) = P'(\y(\z))$. Since $P'$ is a $d$-junta-sum, so is $P$ (for every choice of $\x'$), i.e., $P\in \J_d(\{0,1\}^n,G)$. Furthermore, we observe that for $\y \sim \bern(1/s)^n$, the point $\z(\y)$ is uniformly distributed over $[s]^n$ (over a random choice of $\x'$ and $\y$). In particular, we have
    $$\E_{{\x'}} \bigg[\E_{\y \sim \bern(1/s)^n} [\mathbbm{1}[f(\y) \ne P(\y)]]\bigg] = \Pr_{{\z\sim [s]^n}} [f'(\z)\ne P'(\z)]\le \varepsilon_1.$$
    By Markov's inequality, therefore, $\Pr_{{\y \sim \bern(1/s)^n}}[f(\y) \ne P(\y)] \le 10\varepsilon_1$ with probability at least 0.9 over the choice of $\x'$. Now using $\A$ and oracle access to $f$ (which can be simulated using the oracle access to $f'$), we get a $q$ query algorithm that outputs $P'(\x)$ with probability at least $3/4-0.1$, which can be made at least $2/4$ by repeating this subroutine constant number of times. Finally, we have a $O(q)$ query algorithm $\A'$ such that $\Pr[\A'^{f'}(\x) \ne P'(\x)] \le 1/4$. 
\end{proof}

\begin{proof}[Proof of~\Cref{lem:low-error}]
    Using~\Cref{clm:local-redn}, we have the ability to work with a biased distribution over the Boolean cube instead of a uniform distribution over $[s]^n$ (we note that the change of error from $\varepsilon_1$  to $10\varepsilon_1$ and the queries from $q$ to $O(q/\varepsilon_1)$ are insignificant to the final asymptotic query complexity).   Hence, it suffices to show that there exists $\varepsilon_1 = 1/\Theta_{s,d}(\log n)^d$ and a $O_{s,d}(\log n)^d$ query algorithm $\A$ such that for every randomized oracle $f:\{0,1\}^n \to G$ satisfying $$\Pr_{{\x\sim \bern(1/s)^n}} [f(\x) \ne P(\x)] \le \varepsilon_1$$ for some $P\in \J_d(\{0,1\}^n,G)$, it holds that $\Pr[\A^f({\bf 1}) \ne P({\bf 1})] \le 1/4$. In other words, we want to locally correct low-junta-degree functions over the Boolean cube under a biased distribution. The high level idea is to adapt the construction for the unbiased distribution from~\cite{ABPSS25}.  In particular, we prove the following key result, and the local corrector is then described in~\Cref{algo:localc}. 

    \begin{theorem}\label{thm:interpolate}
        For a growing parameter $k$ divisible by $10s^2d$, there exists $\mathcal{S} \subseteq \{0,1\}^k$ of size at most $O_{s,d}(k^d)$ such that the following conditions hold:
        \begin{itemize}
        \item $\mathcal{S}$ is {\em weight-balanced}: i.e., there exists a probability distribution $\mathcal{D}$ over $[k]$, such that for every ${\bf b} \in \mathcal{S}$: it holds that
        \begin{align}\label{eqn:bias}\bigg|\E_{i\sim \mathcal{D}}[b_i]-\frac{1}{s}\bigg| \le \frac{1}{2^{\Omega_{s,d}(k)}}.\end{align}
        \item $\mathcal{S}$ is an {\em interpolating set}: i.e., for every Abelian group $G$ and every $Q\in \J_d(\{0,1\}^k,G)$, there exist integers $(c_{{\bf b}})_{{\bf b}\in \mathcal{S}}$ such that 
        $$Q({\bf 1}) = \sum_{{\bf b}\in \mathcal{S}} c_{{\bf b}} Q({\bf b}).$$
        \end{itemize}
    \end{theorem}

\begin{algobox}
\begin{algorithm}[H]
\caption{Local corrector in low-error regime}
\label{algo:localc}
\DontPrintSemicolon

\KwIn{Oracle access to the function $f:\{0,1\}^n\to G$}
\vspace{3mm}

Set $k=\Theta_{s,d}(\log n)$ so that the RHS term in~\eqref{eqn:bias} (i.e., $\frac{1}{2^{\Omega_{s,d}(k)}}$) is at most $\frac{1}{n^2}$ and let $\mathcal{S} \subseteq \{0,1\}^k$ be given by~\Cref{thm:interpolate}.

Let $\mathcal{D}$ be the probability distribution over $[k]$ also given by~\Cref{thm:interpolate}.

For ${\bf b}\in \mathcal{S}$, let ${\x} = {\x}({\bf b}) \in \{0,1\}^n$ be the point obtained by setting $x_i = b_{j}$, where $j\sim \mathcal{D}$ is sampled independently for all $i\in [n]$. 

Output $\sum_{{\bf b}\in \mathcal{S}} c_{{\bf b}} f(\x({\bf b}))$, where $c_{{\bf b}}$ are integers given from~\Cref{thm:interpolate}.
\end{algorithm}
\end{algobox}

We first prove the correctness of~\Cref{algo:localc} before we provide a proof of~\Cref{thm:interpolate}. We first note that the number of queries made by the local correction algorithm is equal to $|\mathcal{S}|$, which is $O_{s,d}(k^d) = O_{s,d}(\log n)^d$ as desired. It now remains to show that the probability of error $\Pr[\A^f({\bf 1})\ne P({\bf 1})]$ is at most $1/4$, where $\A^f({\bf 1}) = \sum_{{\bf b}\in \mathcal{S}} c_{{\bf b}} f(\x({\bf b}))$ is the output of~\Cref{algo:localc}. Let $Q:\{0,1\}^k\to G$ be defined by $Q(\y) = P(\x(\y))$ i.e. it depends on the choice of randomness used in Step 3 of~\Cref{algo:localc}. Since $P$ is a $d$-junta-sum, so is $Q$, so by~\Cref{thm:interpolate}, we know that $$Q({\bf 1}) = \sum_{{\bf b}\in \mathcal{S}} c_{{\bf b}} Q({\bf b}).$$ Equivalently, we thus get $$P({\bf 1}) = \sum_{{\bf b}\in \mathcal{S}} c_{{\bf b}} P(\x({\bf b})).$$
Hence, if all the queries to $f$ by $\A$ output the value of $P$, then there is no error in the algorithm. However, there are two sources of error: firstly, $f(\x({\bf b}))$ need not always be equal to $P({\x({\bf b})})$. Indeed we are only guaranteed that they are equal with high probability for an input chosen from $\bern(1/s)^n$ distribution. And secondly, the distribution of ${\x({\bf b})}$ is not exactly identical to the $\bern(1/s)^n$ distribution, but only statistically close to it. More precisely, we have
$$\Pr_{{\x}\sim \bern(1/s)^n}[f(\x) \ne P(\x)] \le \varepsilon_1,
$$ and the statistical distance between the distributions $\bern(1/s)^n$ and $\x({\bf b})$ is: $$\text{SD}(\bern(1/s)^n, \x({\bf b})) \le \frac{1}{n},$$ for every ${\bf b} \in \mathcal{S}$ by the weight-balanced property of $\mathcal{S}$ as each bit of ${\x({\bf b})}$ is $\frac{1}{n^2}$-close to $\bern(1/s)$ and the $n$ bits are all independent (see Step 3 of~\Cref{algo:localc}); here are we using the property $\text{SD}((X_1,X_2),(Y_1,Y_2)) \le \text{SD}(X_1,Y_1) + \text{SD}(X_2,Y_2)$ if $X_1,X_2$ are independent and so are $Y_1,Y_2$ (see e.g.~\cite{vadhan-pseudorandomness} Lemma 6.3). Thus, we have for each ${\bf b}\in \mathcal{S}$, $\Pr[f(\x({\bf b})) \ne P(\x({\bf b}))] \le \varepsilon_1 + \frac{1}{n}$. Now, applying a union bound over the queries made, we get
$$\Pr[\A^f({\bf 1}) \ne P({\bf 1})] \le |\mathcal{S}| \cdot \bigg(\varepsilon_1 + \frac{1}{n}\bigg) \le 1/4,$$ by taking $\varepsilon_1 = 1/\Theta_{s,d}(\log n)^d$ appropriately small. 

This finishes the proof of~\Cref{lem:low-error}.
\end{proof}

We now prove~\Cref{thm:interpolate}.

\begin{proof}[Proof of~\Cref{thm:interpolate}]
    We let $k=rm$, where $r=10s^2d$ and identify $[k]$ with $[r]\times [m]$ arbitrarily and treat ${\y} \in \{0,1\}^k$ as a tuple of points in $\{0,1\}^r$, i.e., we let $\y=(\y_1,\y_2,\dots,\y_m)$ where each $\y_i \in \{0,1\}^r$ (equivalently we treat the point $\y$ as a Boolean $r\times m$ matrix with $\y_i$ being the column vectors). Then, we define the distribution $\mathcal{D} = \mathcal{D}(m)$ over $[k] \equiv [r] \times [m]$ so that the probability mass for $(i,j)$ is proportional to $W_j^{(m)} = s^{m-j}$; in particular, we have $\Pr_{(i,j)\sim \mathcal{D}} = W_j^{(m)}/W^{(m)}$, where we denote $W=W^{(m)}=r\sum_{j=1}^{m} W_j^{(m)} = \frac{r(s^m-1)}{s-1}$.

    There exists a subset $\mathcal{S} = \mathcal{S}_{m,d} \subseteq \{0,1\}^{r\times m}$ of size at most $(4rm)^d$ such that
    \begin{itemize}
        \item $\mathcal{S}$ is weight-balanced: i.e., for every ${\bf b} \in \mathcal{S}$, we have $$\bigg| \sum_{(i,j)\in [r]\times [m]} \frac{W_j^{(m)}}{W^{(m)}} b_{j,i} - \frac{1}{s} \bigg| \le \frac{d}{W^{(m)}}.$$
        \item $\mathcal{S}$ is a {\em hitting set}: i.e., for every Abelian group $G$ and every non-zero $Q\in \J_d(\{0,1\}^k,G)$, there exists ${\bf b} \in \mathcal{S}$ such that $Q({\bf b}) \ne 0$.
    \end{itemize}
    We note that the notion of a {\em hitting set} in the second item implies the interpolating set property in the statement of~\Cref{thm:interpolate} by using Claim 3.2.4 of~\cite{ABPSS25}. Moreover, we note that the RHS of the first item is at most $\frac{O(sd)}{2^{m}} = \frac{1}{2^{{\Omega_{s,d}}(k)}}$ as required.  Thus, it remains to show the existence of the subset $\mathcal{S}_{m,d} \subseteq \{0,1\}^{r \times m}$ satisfying the above two conditions; we do this by induction on $m$. 
    
    \paragraph{Base case $m=1$.}  
    We will make use of the following claim from~\cite{ABPSS25}. 

    \begin{claim}[\cite{ABPSS25} Claim 3.2.3]\label{clm:h}
        For every interval $I \subseteq \{0,1,\dots,r\}$ of size at least $d+1$, there exists a subset $\mathcal{H}_{I,d} \subseteq \{0,1\}^r$ of size at most $(4r)^d$ such that
        \begin{itemize}
            \item $\mathcal{H}_{I,d}$ consists only of points $\z$ such that $|\z| \in I$, and
            \item For every non-zero $Q \in \J_d(\{0,1\}^k, G)$, there exists ${\z} \in \mathcal{H}_{I,d}$ such that $Q(\z) \ne 0$.
        \end{itemize}
    \end{claim}

    Using the above with $I= [\frac{r}{s} - d,\frac{r}{s}+d]$, we directly get $\mathcal{S}_{1,d} = \mathcal{H}_{I,d}$ as the desired set -- the weight-balanced property of $\mathcal{S}$ follows by the immediately as for every ${\bf b} \in \mathcal{S}_{1,d}$, we have $| |{\bf b}| -\frac{r}{s} | \le d$ by the first property of~\Cref{clm:h}.
    
    \paragraph{Induction step $m>1$.} Let $\mathcal{S}_{m-1,d'} \subseteq \{0,1\}^{r\times (m-1)}$ be given by the induction hypothesis, and similarly $\mathcal{H}_{I,d'} \subseteq \{0,1\}^r$ be given by~\Cref{clm:h} for  $0\le d'\le d$. Let ${\bf b}=({\bf b}_1,\dots,{\bf b}_{m-1}) \in \mathcal{S}_{m-1, d}$ be arbitrary. 
    This gives us that 
    \begin{align}\label{eqn:m-1} \bigg |\sum_{(i,j) \in [r]\times [m-1]} W_j^{(m-1)} b_{j,i} - \frac{W^{(m-1)}}{s} \bigg| \le d.\end{align}
    We now show that there exists an interval $I_{{\bf b}} \subseteq \{0,1,\dots,r\}$ of size at least $d+1$ such that for every ${\bf b}_m \in \{0,1\}^r$ with $|{\bf b}_m| \in I_{{\bf b}}$, it holds that ${\bf b}'=({\bf b}_1, \dots, {\bf b}_{m-1}, {\bf b}_m) \in \{0,1\}^{r \times m}$ satisfies the weight-balanced property, i.e., by letting $\tau = \sum_{(i,j)\in [r]\times [m-1]} W_j^{(m-1)} b_{j,i} - \frac{W^{(m-1)}}{s}$ and $I_{{\bf b}}$ to be the interval $\frac{r}{s} - s\tau \pm d$ (which is well-defined as $|\tau| \le d \le \frac{r}{10s^2}$ from~\eqref{eqn:m-1}), we have:
    \begin{align}\label{eqn:m}
         \bigg |\sum_{(i,j) \in [r]\times [m]} W_j^{(m)} b_{j,i} - \frac{W^{(m)}}{s} \bigg| = \bigg| |{\bf b}_m| - \frac{r}{s} + s\tau \bigg | \le d.
    \end{align}
   Now, we are ready to describe $\mathcal{S} = \mathcal{S}_{m,d}$:
   $$\mathcal{S} = \bigcup_{0\le d'\le d} \{{\bf b} \times \mathcal{H}_{I_{{\bf b}},d-d'} ~:~{\bf b} \in \mathcal{S}_{m-1,d'}\}.$$ In particular, we show the following three properties for the above definition of $\mathcal{S}$. 
   \begin{itemize}
        \item {\bf Size.} We have that 
        \begin{align*}|\mathcal{S}| & \le \sum_{d'=0}^d |\mathcal{S}_{m-1,d'}|\cdot |\mathcal{H}_{I_{{\bf b}},d-d'}| \\
        & \le \sum_{d'=0}^d  (4(m-1)r)^{d'} \cdot (4r)^{d-d'} \tag{using the induction hypothesis to upper bound $|\mathcal{S}_{m-1,d'}|$} \\
        & \le (4r)^d \sum_{d'=0}^d (m-1)^{d'}\\
        & \le (4mr)^d.\end{align*}
        \item {\bf Weight-balanced.} This follows from the discussion leading to~\eqref{eqn:m}.
        \item {\bf Hitting set.} Let $Q \in \J_d(\{0,1\}^{r\times m},G)$ be an arbitrary non-zero $d$-junta-sum. Treating it as a junta-polynomial in the last column of variables, we have for every ${\bf x} = ({\bf x}_1,\dots,{\bf x}_m) \in \{0,1\}^{r\times m}$:
        $$Q({\bf x}) = \sum_{A \subseteq [r]:|A|\le d} Q_A({\bf x}_1,\dots,{\bf x}_{m-1}) \cdot {\bf x}_m^A.$$ Since $Q$ is non-zero, let $A \subseteq [r]$ be such that $Q_A$ is a {\em non-zero} function of junta-degree $d'\le d$. By induction hypothesis, we know there exists ${\bf b} \in \mathcal{S}_{m-1,d'}$ such that $Q_A({\bf b}) \ne 0$. Letting $Q':\{0,1\}^r \to G$ denote the restriction of $Q$ obtained on setting ${\bf x}_i = {\bf b}_i$ for all $i\in [m-1]$, we note that $Q'$ is a {\em non-zero} junta-polynomial of degree at most $d-d'$. Hence, there exists ${\bf b}_m \in \mathcal{H}_{I_{{\bf b}},d'}$ such that $Q'({\bf b}) \ne 0$. Effectively, this shows that there exists ${\bf b}'\in \mathcal{S}$ such that $Q({\bf b}')\ne 0$.
   \end{itemize}
\end{proof}

\subsection{Correction for Torsion Groups} \label{sec:torsion}

We now finish the proof for the ``moreover'' part of~\Cref{thm:local-correction}, i.e., we show a constant query local correction algorithm over torsion groups of constant exponent. 

\begin{proof}[Proof of~\Cref{thm:local-correction} for torsion Abelian groups]
    Similar to the case of general Abelian groups, we first apply the error reduction step from~\Cref{sec:error-redn} (but with a different threshold $\varepsilon_1$) and the reduce the local correction problem to that over the Boolean cube but with a biased distribution (i.e.,~\Cref{clm:local-redn}). Thus, it suffices to show that there exists a $q=O_{M,s,d}(1)$ query algorithm $\A$ such that for every randomized oracle $f:\{0,1\}^n \to G$ satisfying $\Pr_{{\bf x}\sim \bern(1/s)^n}[f(\x) \ne P(\x)] \le \varepsilon_1$ for some $P\in \J(\{0,1\}^n,d)$, it holds that $\Pr[\A^f({\bf 1}) \ne P({\bf 1})] \le 1/4$. 

    In particular, we set $\varepsilon_1 = \frac{1}{10 {sk \choose k}}$ for a suitably large $k=\bigO_{M,s,d}(1)$ (so $\varepsilon_1 \ge \Omega_{M,s,d}(1)$). We state it as a lemma below:

    \begin{lemma}
    \label{lem:const-torsion-sub} For every Abelian torsion group $G$ of exponent $M$, there exists $k=\bigO_{M,s,d}(1)$ and a $q=O_{M,s,d}(1)$ query algorithm $\A$ such that for every randomized oracle $f:\{0,1\}^n \to G$ satisfying $\Pr_{{\bf x}\sim \bern(1/s)^n}[f(\x) \ne P(\x)] \le \varepsilon_1$ for some $P\in \J(\{0,1\}^n,d)$ and $\varepsilon_1 = \frac{1}{10 {sk \choose k}}$, it holds that $\Pr[\A^f({\bf 1}) \ne P({\bf 1})] \le 1/4$.
    \end{lemma}

    Now we note that by using the error-reduction lemma~\Cref{lem:error-reduction-main} with $\eta=\varepsilon_1\ge \Omega_{M,s,d}(1)$, we can convert a local corrector for error $\varepsilon_1$ to one with error up to $1/(2s^d)-\varepsilon$ with a $\bigO_{\varepsilon}(1)$ factor blow-up. Combining with the low-error local corrector of~\Cref{lem:const-torsion-sub}, we obtain a local corrector over the biased distribution $\bern(1/s)^n$ making $\bigO_{M,s,d}(1)$ queries. Therefore, by~\Cref{clm:local-redn}, we also get a $\bigO_{M,\varepsilon}(1)$ query local corrector for $d$-junta-sums over $\sgrid^n$ for error up to $1/(2s^d)-\varepsilon$.

\end{proof}

We now prove~\Cref{lem:const-torsion-sub}.

\begin{proof}[Proof of~\Cref{lem:const-torsion-sub}]

The proof proceeds in an identical manner to the analysis of~\cite{ABPSS25} by making use of Kummer's theorem which may be thought of as an analog of Lucas' theorem for prime powers. We state Kummer's theorem below, where the notation $S_p(n)$ denotes the sum of the digits of $n$ when written in base $p$.\\

\begin{theorem}[Kummer's theorem~\cite{kummer1852}] 
\label{thm:kummer}
    Let $p \in \mathbb{N}$ be a prime. Then for any integers $a \ge b \ge 0$, the largest power of $p$ that divides ${a \choose b}$ is equal to $\frac{S_p(b)+S_p(a-b)-S_p(a)}{p-1}$.
\end{theorem}

    Let $M=\prod_{j=1}^{\ell} p_j^{r_j}$ be the prime factorization of the exponent $M$ of $G$ (so $\ell \le \log M$). For each $j\in [\ell]$, let $s_j \in \mathbb{N}$ be the smallest integer such that $p_j^{r_js_j} > d$. Then, we choose $k=\prod_{j\in [\ell]} p_j^{3r_js_j}$. Note that $p_j^{r_j(s_j-1)} \le d$ and hence $k \le \prod_{j\in [\ell]}(dp_j^{r_j})^3 \le d^{3\ell}M^3=\bigO_{M,d}(1) $ as needed. We then recall that $\varepsilon_1 = \frac{1}{10{sk \choose k}}=\Omega_{M,s,d}(1)$. 
    
    We claim that the algorithm below (\Cref{algo:constant-torsion}) is the desired local corrector. It queries $f$ at a few inputs from some distribution and outputs $P({\bf 1})$ with probability at least $9/10$, where $P \in \mathcal{J}_d$ is the unique degree-$d$ junta-sum such that $\delta(f,P) \le \varepsilon_1$. We will need the following claim in order to describe the local corrector. 

    \begin{claim}
    \label{clm:interpolate-slice}
        There exist integers $c_{{\bf b}}\in \Z$ for ${\bf b} \in {[sk] \choose k}$ such that for every $d$-junta-sum  $Q({\bf y}) \in \mathcal{J}_d(\{0,1\}^{sk}, G)$, we have that 
        \begin{align}Q({\bf 1}) = \sum_{{\bf b}\in {[sk]\choose k}} c_{\bf b} \cdot Q({\bf b}).\end{align}
    \end{claim}

\begin{algobox}
\begin{algorithm}[H]
\caption{Local corrector for torsion groups}
\label{algo:constant-torsion}
\DontPrintSemicolon

\KwIn{Oracle access to a randomized function $f:\{0,1\}^n\to G$}
\vspace{3mm}

Sample a uniformly random function $h: [n] \to [sk]$.

For ${\bf b}\in {[sk]\choose k}$, let ${\x} = {\x}_h({\bf b}) \in \{0,1\}^n$ be the point obtained by setting $x_i = b_{h(i)}$ for $i\in [n]$.

Output $\sum_{{\bf b}\in {[sk]\choose k}} c_{{\bf b}} f(\x_h({\bf b}))$, where $c_{{\bf b}}$ are integers given by~\Cref{clm:interpolate-slice}.
\end{algorithm}
\end{algobox}

The above algorithm is similar to~\Cref{algo:localc} with the main difference being the choice of the interpolating set in the last step from~\Cref{clm:interpolate-slice} (as opposed to the ``weight balanced interpolating set'' of~\Cref{thm:interpolate}).

Assuming the correctness of~\Cref{clm:interpolate-slice}, we shall now finish the proof of~\Cref{lem:const-torsion-sub}. 

Firstly, we note that the local corrector makes $ {sk \choose k} = \bigO_{M,s,d}(1)$ queries as required. To prove correctness, for every ${\bf b} \in {[sk] \choose k}$, we note that the corresponding query point $x_h({\bf b}) \in \{0,1\}^n$ is distributed according to $\bern(1/s)^n$ since the map $h$ used in~\Cref{algo:constant-torsion} is uniformly random and ${\bf b}$ has $1/s$ fraction of indices as ones. Thus, we have that $f({\bf x}_h({\bf b})) \ne P({\bf x}_h({\bf b}))$ with probability at most $\varepsilon_1$, so by a union bound over ${\bf b}$, we have that with probability at least $1-\varepsilon_1 \cdot {sk \choose k} = 9/10$ (over the random choice of $h$ and the randomness of $f$), that $f({\bf x}_h({\bf b})) = P({\bf x}_h({\bf b}))$ for all ${\bf b}\in {[sk] \choose k}$. Now, letting $Q \in \mathcal{J}_d(\{0,1\}^{sk}, G)$ denote the restriction of $P$ defined as $Q({\bf y}) = P({\bf x}_h({\bf y}))$, we see that the output of~\Cref{algo:constant-torsion} is equal to $$\sum_{{\bf b}\in {[sk]\choose k}} c_{{\bf b}} P({\bf x}_h({\bf b}))= \sum_{{\bf b}\in {[sk]\choose k}} c_{{\bf b}} Q({\bf b}) = Q({\bf 1}) = P({\bf 1}),$$ where we are using~\Cref{clm:interpolate-slice} for the second equality and ${\bf x}_h({\bf 1}) = {\bf 1}$ for the last equality.  Therefore, the output of the local correction algorithm (\Cref{algo:constant-torsion}) is indeed $P({\bf 1})$ with probability at least $9/10$.  
\end{proof}

It now remains to prove~\Cref{clm:interpolate-slice}.

\begin{proof}[Proof of~\Cref{clm:interpolate-slice}] 

By replacing the variables $x_i$ with $1-x_i$, we note that the claim is equivalent to proving that there exists $c_{{\bf b}}\in \Z$ for ${\bf b}\in {[sk]\choose (s-1)k}$ such that for every $d$-junta-sum $Q\in \J_d(\{0,1\}^{sk}, G)$, it holds that 

\begin{align}\label{eqn:linear} Q({\bf 0}) = \sum_{{\bf b}\in {[sk]\choose (s-1)k}} c_{{\bf b}}\cdot Q({\bf b}).\end{align}

To show this, we proceed with the following assignments. For every ${\bf b} \in {[sk] \choose (s-1)k}$, we set $c_{\bf b} = 0$ if ${\bf b}$ contains a 1 in any of the last $k-d$ coordinates and we set $c_{\bf b} = A$ otherwise, where $A \in \Z$ will be decided later. Recall that $M=\prod_{j\in[\ell]} p_j^{r_j}$ and $k=\prod_{j\in[\ell]} p_j^{3r_js_j}$, and we have that $p_j^{r_j s_j} > d \ge p_j^{r_j(s_j-1)}$ for all $j\in [\ell]$. By linearity, it suffices to show~\eqref{eqn:linear} for $Q({\bf y})$ of the form $g\cdot \prod_{j\in I} y_j$ for all $I \in {[sk] \choose \le d}$ and $g\in G$. According to our assignment of $c_{\bf b}$, it is clear that~\eqref{eqn:linear} holds true (with LHS = RHS = 0) if $I$ contains any of the last $k-d$ coordinates. Otherwise, we have that $I \subseteq {[(s-1)k+d] \choose \le d}$. If $I=\emptyset$, we have $Q(0^{sk})=g$ and $\sum_{{\bf b}\in {[sk] \choose (s-1)k}} c_{\bf b}\cdot Q({\bf b}) = {(s-1)k+d \choose (s-1)k}A\cdot g$. On the other hand, if $|I|=i\ge 1$, we have $Q(0^{sk}) = 0$ and $\sum_{{\bf b}\in {[sk] \choose k}} c_{\bf b}\cdot Q({\bf b}) = {(s-1)k+d-i \choose (s-1)k-i}A\cdot g$ since every non-zero term must have $b_j=1$ for all $j\in I$. Hence, it suffices to find an integer $A$ satisfying the following two conditions:
\begin{align*}
     g & = {(s-1)k+d \choose (s-1)k} A \cdot g, \text{~for all~} g\in G, \text{~and~}\\
     0 & = {(s-1)k+d-i \choose (s-1)k-i} A \cdot g, \text{~for all~} g\in G \text{~and~} i\in [d].
\end{align*}

Let $k':=(s-1)k$. Since the order of every element $g$ divides the exponent $M$ of the group, for the above two conditions to hold, it suffices if for all $j\in [\ell]$ and $i\in [d]$, $p_j$ does not divide ${k'+d \choose k'}$ and that $p_j^{r_j}$ divides ${k'+d-i \choose k'-i}$ for all $i\in [d]$. Then we can take $A$ to be any integer such that $A {k'+d \choose k'} + A' M =1$ for some integer $A'$ (such $A$ and $A'$ are guaranteed to exist as $M$ and ${k'+d \choose k'}$ are coprime). The rest of the proof is dedicated to verifying these divisibility constraints hold.

\begin{itemize}
    \item \textbf{$p_j$ does not divide ${k'+d \choose k'}$:} We will represent all the numbers $k',d,i$ etc.~in base $p_j$. We note that the last $r_j s_j$ digits of $k'$ are zeroes since $p_j^{r_j}$ divides $k'$. Furthermore, since $d < p_j^{r_j s_j}$, all the digits of $d$ except the last $r_j s_j$ many are zeroes. Hence, the sum of digits of $k'+d$ is equal to the sum of the digits of $k'$ and $d$ combined. That is, $S_{p_j}(k')+S_{p_j}(d)-S_{p_j}(k'+d)=0$. Applying Kummer's theorem (\Cref{thm:kummer}) now finishes the proof. 
    
    \item \textbf{$p_j^{r_j}$ divides ${k'+d-i \choose k'-i}$:} By Kummer's theorem (\Cref{thm:kummer}), it suffices to show that 
    \begin{align}\label{eqn:kummer-eqn} \frac{S_{p_j}(d)+S_{p_j}(k'-i)-S_{p_j}(k'+d-i)}{p_j-1} \ge r_j.\end{align}
    We note that $S_{p_j}(k'+d-i) = S_{p_j}(k')+S_{p_j}(d-i)$ by the same argument as the above paragraph. In addition, we have the trivial bounds $S_{p_j}(d) \ge 1$ and $S_{p_j}(d-i) \le (p_j-1)r_j s_j$. Finally, we give a lower bound for $S_{p_j}(k'-i)$. Since $k'$ has at least $3r_j s_j$ trailing zeroes, we get that $S_{p_j}(k'-1) \ge S_{p_j}(k')+3r_js_j(p_j-1)-1$. But we observe that $S_{p_j}(k'-i) = S_{p_j}((k'-1)-(i-1)) = S_{p_j}(k'-1) - S_{p_j}(i-1)$ since the number of trailing $(p_j-1)$'s of $k'-1$ exceeds the total number of (non-zero) digits of $(i-1)$. Therefore, we get
    \begin{align*}
        {S_{p_j}(d)+S_{p_j}(k'-i)-S_{p_j}(k'+d-i)} & \ge 1+S_{p_j}(k'-1)-S_{p_j}(i-1)-S_{p_j}(k')-S_{p_j}(d-i)\\
        & \ge 1+(3r_js_j({p_j}-1)-1)-(p_j-1)r_js_j-(p_j-1)r_js_j\\
        & \ge r_js_j(p_j-1)\\
        & \ge r_j(p_j-1).
    \end{align*}
\end{itemize}

This finishes the proof of~\eqref{eqn:kummer-eqn}, and hence~\Cref{clm:interpolate-slice} and~\Cref{lem:const-torsion-sub}.

\end{proof}

\section{Combinatorial List-Decodability}\label{sec:cld}

We prove the combinatorial list-decodability bound for junta-sums (i.e., \Cref{thm:comb-bd}). \\

\combbound*

The proof can be broken into the following four steps:
\begin{itemize}
    \item First, we reduce to the setting where $G$ is finite.
    \item Second, we show the combinatorial bound for finite groups where every element has a sufficiently large order. 
    \item Third, we show the combinatorial bound for $p$-primary groups where $p$ is a sufficiently small prime.\footnote{An Abelian group is said to be {\em $p$-primary} if every element has order that is an exponent of $p$.}
    \item Finally, we combine the above bounds to get a combinatorial bound for arbitrary Abelian groups.
\end{itemize}
In particular, we prove the following two theorems.\\

\begin{theorem}[Combinatorial bound for large order]\label{thm:comb-large}
    For every $\varepsilon>0$, positive integers s,d, there exists a $p=p(s,d,\varepsilon)$ such that for every Abelian group $G$ which does not have any element of order at most $p$, the family $\J_d([s]^n,G)$ is $(1/s^d-\varepsilon, O_\varepsilon(1))$-list decodable.
\end{theorem}

And we have:\\

\begin{theorem}[Combinatorial bound for $p$-primary groups]\label{thm:comb-small}
    For every $\varepsilon>0$, positive integers $s,d$, prime $p$, and finite $p$-primary group $G$, the family $\J_d([s]^n, G)$ is $(1/s^d-\varepsilon,O_{\varepsilon,p}(1)$-list decodable. 
\end{theorem}

Using the above two theorems, we finish the proof of~\Cref{thm:comb-bd}.
\begin{proof}[Proof of~\Cref{thm:comb-bd}]
    Given~\Cref{thm:comb-large} and~\Cref{thm:comb-small}, the proof follows the same outline as the prior work~\cite{ABPSS25} on combinatorial bound for low-degree polynomials over the Boolean cube, so we defer the proof.
\end{proof}

We now show the proof for the large order case in~\Cref{sec:large-order} and the $p$-primary groups case in~\Cref{sec:p-grp}.

\subsection{Combinatorial Bound for Large Order}\label{sec:large-order}

We prove~\Cref{thm:comb-large} in this subsection. Throughout this section, we assume that all the elements of $G$ have order at least $p$ (where $p=p(\varepsilon)$ is to be determined), $s\ge 2$, and $ \varepsilon \in (0, 1/s^d)$ is arbitrary. Following along the lines of the proof for the Boolean case ($s=2$) from~\cite{ABPSS25}, we prove (and use) an anti-concentration inequality for junta-sums depending on many variables. Once the right anti-concentration lemma is in place, the rest of the proof of the combinatorial bound is more or less identical to the Boolean case, except we are able to make some simplifications as we are only aiming for a bound of $O_\varepsilon(1)$ (as opposed to $\poly(1/\varepsilon)$ from~\cite{ABPSS25}). We state the anti-concentration inequality below and defer its proof to the end of this subsection. In order to state the lemma, we need a definition -- we say that a function $f:[s]^n \to G$ {\em depends} on the $i$-th variable if there exists $\x\in [s]^n$ such that $f(\x) \ne f(\x')$ for some $\x' \in [s]^n$ that agrees with $\x$ on the coordinates $[n]\setminus \{i\}$.  

\begin{lemma}[\bf Anti-concentration lemma]\label{lem:anti-conc}
    For integers $s\ge 2$ and $d\ge 1$, and every $\varepsilon > 0$, there exists $r = r(s,d,\varepsilon)>0$ and $p=p(s,d,\varepsilon)$ such that for every Abelian group $G$ which does not contain any element of order less than $p$, and every $P \in \J_d([s]^n, G)$ that depends on at least $r$ variables, it holds that:
    $$ \Pr_{{{\bf a} \sim [s]^n}} [P({\bf a})\ne 0] \ge 1/s^{d-1}-\varepsilon.$$
\end{lemma}

Note that this improves on the trivial bound of $1/s^d$ for general non-zero junta-sums. Given the above lemma, we now prove~\Cref{thm:comb-large}. The proof proceeds in multiple stages -- in each stage, we make the junta-sums in the list (of close-by junta-sums to a fixed function) more structured, thus pruning the list at each stage.

\subsubsection{Pruning the List} \label{sec:prune}

For a function $f:[s]^n\to G$, let $L_\varepsilon(f) \subseteq \J_d([s]^n,G)$ denote the set (or rather ``list'') of junta-sums $P$ such that $\delta(f,P) \le 1/s^d - \varepsilon$. Our goal is to show that $|L_\varepsilon(f)| \le O_\varepsilon(1)$. We first reduce the problem to counting the number of junta-sums in the list that depend only on a few variables.
\paragraph{\em Reducing to counting junta-sums depending on a {few variables}.} If $P_1,P_2 \in L_\varepsilon(f)$, note that $$\delta(P_1 - P_2, {\bf 0}) = \delta(P_1,P_2) \le \delta(f,P_1) + \delta(f,P_2) \le 2/s^d - 2\varepsilon.$$  
Now applying~\Cref{lem:anti-conc} for $P=P_1-P_2$ (which is also a $d$-junta-sum), we get that $P_1-P_2$ depends on at most $r(\varepsilon)$ variables, as otherwise we get $1/s^{d-1} - \varepsilon \le \delta(P_1-P_2,{\bf 0}) \le 2/s^d - 2\varepsilon$ which would be a contradiction. Hence, if $L_\varepsilon(f) = \{P_1,P_2,\dots,P_t\}$, we observe that $P_1-P_t,P_2-P_t,\dots,P_{t-1}-P_t$ are distinct junta-sums that are in $L_\varepsilon(P_1-f)$ and depend on at most $r = O_\varepsilon(1)$ variables. Therefore, it suffices to count such junta-sums to get a final combinatorial bound. In order to do this, we first count such junta-sums depending on the same set of variables. 

\paragraph{\em Counting junta-sums depending on the {same set of few variables}.} Without loss of generality, let the variable set on which the junta-sums depend on be $[r]$. That is, let $P_1,\dots,P_t$ be the $d$-junta-sums that are at distance at most $1/s^d-\varepsilon$ from a function $f:[s]^n \to G$, and each $P_i$ only depends on the first $r$ variables. For ${\bf a} \in [s]^{n-r}$, let $f_{{\bf a}}:[s]^r \to G$ be the function obtained by setting the last $n-r$ variables of $f$ to be uniformly random independently. Since $\delta(f,P_i) \le 1/s^d-\varepsilon$ for every $i\in [t]$, we have $\E_{{\bf a}}[\delta(f_{{\bf a}},P_i)] \le 1/s^d-\varepsilon$, hence with probability at least $\varepsilon/2$ over the choice of ${\bf a}$, it holds that $\delta(f_{{\bf a}},P_i) \le 1/s^d-\varepsilon/2$ (where we are thinking of $P_i$ as being a function from $[s]^r$ to $G$). By linearity of expectation, this means that the expected number of junta-sums $P_i$ such that $P_i \in L_{\varepsilon/2}(f_{{\bf a}})$ is at least $\varepsilon t/2$. Hence, it suffices to show that $|L_{\varepsilon/2}(f')|$ for every $f':[s]^r\to G$ is $O_\varepsilon(1)$ to conclude that $t=O_\varepsilon(1)$. To do this, we note that $P_1 \ne P_2 \in L_{\varepsilon/2}(f')$ cannot agree on more than $1-1/s^d$ fraction of inputs, so for a given subset of $[s]^r$ of size $s^r-s^{r-d}$, there is at most one junta-sum in the list $L_{\varepsilon/2}(f')$ that agrees with $f'$ on that subset. Therefore $|L_{\varepsilon/2}(f')|\le {s^r \choose s^{r}-s^{r-d}} = O_\varepsilon(1)$ as $r=O_\varepsilon(1)$. 

\paragraph{\em Reducing to the case where the variable sets form a {sunflower}.} We recall that our goal is to prove that the number of $d$-junta-sums that are at distance at most $1/s^d-\varepsilon$ from a given $f:[s]^n\to G$ is $O_\varepsilon(1)$. However, from the above paragraph, we see that the number of such junta-sums depending on the same set of variables is $O_\varepsilon(1)$. Thus, it suffices to show the following:

Suppose $P_1,\dots,P_t \in L_\varepsilon(f)$ are such that they depend on {\em distinct} subsets of variables. Then, $t=O_\varepsilon(1)$.  

Now, consider the set system formed over the universe $[n]$ by the subsets of variables each $P_i$ depends on. Applying the sunflower lemma (e.g.~\cite{erdos1960intersection}, Theorem 3) to this set system, we observe that if $t>r!(m-1)^r$, then there exists $P_{i_1},\dots,P_{i_m} \in L_\varepsilon(f)$ such that the subset of variables they depend on forms a {\em sunflower}: that is, if the subset of variables that $P_i$ depends on is denoted by $V_i \subseteq [n]$, then there exists a {\em core} $C \subseteq [n]$ such that $V_{i_{j_1}} \cap V_{i_{j_2}} = C$ for every $j_1 \ne j_2 \in [m]$ and the {\em petals} $V_{i_j} \setminus C$ are non-empty. Hence, it suffices to show that $m=O_\varepsilon(1)$ to get that $t=O_\varepsilon(1)$. For the remainder of the proof, we shall assume that $i_j=j$ for $j\in [m]$, without loss of generality.

\paragraph{\em Reducing to the case where the variable sets are {pairwise disjoint}.} While the application of the sunflower lemma in the above step results in a core $C$ which can be non-empty, the goal of this step is to show that we can essentially assume that $C=\emptyset$ without loss of generality. We prove this by carefully setting the variables in $C$ (which is assumed to be non-empty) to constants. We will switch the domain of the functions from $[s]^n$ to $\Z_s^n$ as we will be using junta-polynomial representations.

Let $\x = \z \cup (\y^{(1)} \cup \y^{(2)} \dots \y^{(m)}) \cup {\bf w}$ be a partition of the variable set where $\z$ denotes the variables indexed by $C$, and $y^{(i)}$ denotes the variables that $P_i$ depends on other than $\z$ (i.e., $y^{(i)}$ corresponds to the variables indexed by $V_i \setminus C$), and ${\bf w}$ are the remaining variables. We let $n_0=|C|=|\z|$ and $n_i = |y^{(i)}|$. Then we note that we can express each $P_i$ (for $i\in [m]$) as follows: 
$$P_i(\x) = P_i(\z,\y^{(i)}) = \sum_{{\bf a}\in \Z_s^{n_i}:|{\bf a}| \le d} \delta_{{\bf a}}(\y^{(i)}) \cdot P_{i,{\bf a}}(\z),$$ 
where we use the notation $\delta_{{\bf a}}({\bf y}^{(i)}) = \prod_{j\in [n_i]} \delta_{a_j}(y^{(i)}_j) $. Let {\em $\y$-degree} of $P_i$ denote the maximum value of $|{\bf a}|$ for which $P_{i,{\bf a}}$ is non-zero; since $P_i$ depends on ${\bf y^{(i)}}$ variables, the ${\bf y}$-degree must be in $[d]$. Moreover, since $|\z| \le r=O_\varepsilon(1)$, the number of possible monomials (without considering coefficients) in $P_{i,{\bf a}}$ is $O_{s,d,\varepsilon}(1)=O_\varepsilon(1)$. Thus, assuming $m$ is a large enough function of $1/\varepsilon$ (otherwise, we are done), using the pigeon-hole principle, we can assume without loss of generality that the ${\bf y}$-degree of the $P_i$'s are all the same (say $d'\in [d]$) and that each $P_{i,{\bf a}}$ contains a non-zero coefficient for the monomial $\delta_{{\bf b}}(\z)$ for some ${\bf b} \in \Z_s^{n_0}$, and that $\delta_{{\bf b}}(\z)$ is a non-zero monomial with the maximal degree. Without loss of generality, let the first $n_0'$ coordinates of ${\bf b}$ be zero and the remaining ones be non-zero, where $0\le n_0' \le n_0$. We will first set the first $n_0'$ variables in ${\bf z}$ (if $n_0'=0$, we skip this step) uniformly at random: we note that setting these variables cannot cancel the monomial $\delta_{{\bf b}}(\z)$ as by assumption, it is a monomial with maximal degree. Denoting the restricted functions by $P_1',\dots,P_t'$ and the restriction of $f$ by $f'$, we have that these are all distinct and each $P_i'$ satisfies $\delta(f',P_i') \le 1/s^d -\varepsilon/2$ with probability at least $\varepsilon/2$. Thus, there exists a choice of assignments to the first $n_0'$  variables of $\z$ such that for at least $\varepsilon t/2$ many $P_i'$ s, it holds that $\delta(f',P_i') \le 1/s^d-\varepsilon/2$. Without loss of generality, we assume that these are the initial $t'=\varepsilon t/2$ junta-sums. We now set the remaining variables of $\z$ uniformly at random. We note that for $i\in [t']$, since $P_{i,{\bf a}}$ is still non-zero even after setting some variables of $\z$ in the earlier step, with probability at least $1/s^{n_0-n_0'}$, it holds that  $P_i$ is non-zero.  However, since the junta-degree of $P_i$ is at most $d$ and the $\y$-degree of $P_i$ is $d'$, we must have that $n_0-n_0' \le d-d'$. That is, denoting the final junta-sums after setting all the variables of $\z$ by $P_i''$ respectively and the restricted function of $f$ by $f''$, we have that $P''_i$ is non-zero with $\Omega_\varepsilon(1)$ probability. Furthermore, each $P_i''$ if non-zero has junta-degree at most $d'$. Since $\delta(f',P_i') \le 1/s^d - \varepsilon/2$ and we are only setting $n_0-n_0' \le d-d'$ variables when going from $P_i'$ to $P_i''$, we must have that $\delta(f'',P_i'') \le 1/s^{d'} - \varepsilon/2$. Thus, we have reduced to the case where the junta-sums we want to count all depend on pairwise disjoint sets of variables (although the degree changes from $d$ to $d'$, we will use $d$ for the rest of the proof for simplicity; similarly we use $\varepsilon$ instead of $\varepsilon/2$).  

\paragraph{\em Counting junta-sums depending on pairwise disjoint variables.} To recap, we are now in the following setup: We have an arbitrary function $f:[s]^n\to G$ and distinct $d$-junta-sums $P_1,\dots,P_t$ depending on pairwise disjoint subsets of variables such that $\delta(f,P_i) \le 1/s^d - \varepsilon$, and the goal is to show that $t=O_\varepsilon(1)$. The main idea is that the junta-sums behave ``independently'' as they depend on disjoint subsets of variables and so there cannot be many of them correlated with the same function $f$. More formally, we consider the following quantity:
\begin{align}\label{eqn:prob}
    \Pr_{{\x}\sim [s]^n} \bigg[\exists i\in [t]:\bigg| \{j\in [t]:P_j(\x)=P_i(\x)\}\bigg| \ge (1-1/s^d+\varepsilon/2)t-1\bigg].
\end{align}
On the one hand, since $\Pr_{{\bf x}\sim [s]^n, i\sim [t]}[f(\x)=P_i(\x)] \ge 1-1/s^d + \varepsilon$, we have that~\eqref{eqn:prob} is at least $\varepsilon/2$ (i.e., with probability $\varepsilon/2$, at least $1-1/s^d+\varepsilon/2$ fraction of the junta-sums agree with $f$ and so with each other). On the other hand, since any two distinct junta-sums agree on at most $1-1/s^d$ fraction of inputs and the events $P_j(\x) = P_i(\x)$ are independent across different $j\ne i$, we have that~\eqref{eqn:prob} is at most $t/2^{\Omega(\varepsilon^2 t)}$. Combining both, we get $t=O_\varepsilon(1)$.
    
\begin{proof}[Proof of~\Cref{thm:comb-large}]
    The above discussion finishes the proof of~\Cref{thm:comb-large}.
\end{proof}

\subsubsection{Anti-concentration Lemma}\label{sec:anti-conc}

We end with a proof of the anti-concentration lemma (\Cref{lem:anti-conc}). For this, we will need the following claim about junta-sums that have a certain matching structure. This is analogous (and extends) the corresponding result of Meka, Nguyen and Vu~\cite{MNV} used in the analysis for the Boolean case ($s=2$) in~\cite{ABPSS25}. To state the claim, we say that two monomials of a junta-polynomial: $\delta_{{\bf a}}$ and $\delta_{{\bf b}}$ (where ${\bf a},{\bf b}\in \Z_s^n$), are {\em disjoint}, if the non-zero indices of ${\bf a}$ and ${\bf b}$ are disjoint. \\

\begin{claim}\label{clm:large-matching}
    For integers $s\ge 2$ and $d\ge 1$ and every $\varepsilon>0$, there exists $u=u(s,d,\varepsilon)$ and $p=p(d,\varepsilon)$ such that for every Abelian group $G$ which does not contain any element of order less than $p$, and every $d$-junta-sum 
    $$P({\bf x}) = \sum_{{\bf a}\in \Z_s^n:|{\bf a}|\le d} g_{{\bf a}}\cdot \delta_{{\bf a}}(\x)$$
    with at least $u$ many pairwise disjoint non-zero monomials of degree $d$, it holds that:
    $$\Pr_{{\bf a}\sim \Z_s^n}[P({\bf a}) = 0] \le \varepsilon.$$
\end{claim}

\begin{proof}
    The main idea is to reduce to the Boolean case and use the following result from~\cite{ABPSS25}, which itself is derived using the anti-concentration result of Meka, Nguyen and Vu~\cite{MNV} over the reals. 
    \begin{lemma}[\cite{MNV},~\cite{ABPSS25} Theorem 4.1.6 and Claim 4.1.5]\label{lem:bool-antic}
        For every positive integer $d$ and $\varepsilon>0$, there exists $t=t(d,\varepsilon)$ and $p=p(d,\varepsilon)$ such that for every Abelian group $G$ which does not contain any element of order less than $p$, and every junta-degree-$d$ polynomial $P\in \J_d(\{0,1\}^n,G)$ with at least $t$ many pairwise disjoint non-zero monomials, it holds that:
        $$\Pr_{{{\bf a}}\sim \{0,1\}^n}[P({\bf a}) = 0] \le \varepsilon.$$
    \end{lemma}
    We now show how to use the above lemma to deduce a similar inequality for general $s$ i.e., we prove~\Cref{clm:large-matching}. Let $u$ denote the number of pairwise disjoint non-zero monomials of degree $d$ in the junta-polynomial representation of $P$. Assuming a sufficiently large lower bound on $u$, our goal is to show that $$\Pr_{{\bf a}\sim \Z_s^n}[P({\bf a})=0] \le \varepsilon.$$
    We choose a uniformly random ${\bf a}\in \Z_s^n$ as follows:
    \begin{itemize}
        \item Choose a random {\em subcube} $C \subseteq \Z^n_s$ by picking ${\bf u},{\bf v}\in \Z^n_s$, where $u_i\ne v_i \in \Z_s$ are chosen uniformly at random and independently over $i\in [n]$: more specifically, $C=\{u_1,v_1\}\times \dots \times \{u_n,v_n\}$.
        \item Choose ${\bf a} \in C$ uniformly at random. 
    \end{itemize}
    Let $t=t(d,\varepsilon/2)$ and $p=p(d,\varepsilon/2)$ be given by the functions $t(.,.)$ and $p(.,.)$ in~\Cref{lem:bool-antic}. Let ${\bf a}_1,\dots,{\bf a}_r \in \Z^n_s$ (where $u=u(s,d,\varepsilon)$ will be decided later) be such that the monomials $\delta_{{\bf a}_i}(\x)$ are pairwise disjoint monomials, with $|{\bf a}_i| = d$, and have non-zero coefficients in $P$, where $i\in [u]$. Let $S_i \subseteq [n]$ denote the indices where ${\bf a}_i$ is non-zero, so that $S_i$ are pairwise disjoint for $i\in [u]$. Now, if $u_j = 0$ and $v_j = a_j$ for all $j\in S_i$, we note that if we treat $P$ restricted to $C$ as function over the Boolean cube (with $u_j\mapsto 0$ and $v_j \mapsto 1$ for $j\in S_i$ and rest of the coordinates are mapped arbitrarily), the monomial $\prod_{j\in S_i} x_j$ has the same coefficient as that of $\delta_{{\bf a}_i}({\bf x})$ in $P$ (which is non-zero), since no other monomials can cancel this. Thus,  if we can prove that there are at least $t$ many of the ${{\bf a}_i}$'s for which it holds that $(u_j,v_j)=(0,a_j)$ for $j\in S_i$, then we have a multilinear polynomial over $C$ (or equivalently over $\{0,1\}^n$) with at least $t$ many non-zero disjoint monomials, in which case, we apply~\Cref{lem:bool-antic} to conclude that for a random point in ${\bf a}\sim C$, the probability that $P({\bf a})=0$ is at most $\varepsilon/2$. Therefore, $$\Pr_{{\bf a}\sim \Z_s^n}[P({\bf a})=0] \le \varepsilon/2 + \Pr[|\{i\in [r]:(u_j,v_j)=(0,a_j)~\forall j\in S_i\}| \ge t].$$ Now, we observe that the events $(u_j,v_j)=(0,a_j)$ are independent across $j\in S_i$ and $i$. In particular, we have $\Pr[(u_j,v_j)=(0,a_j)~\forall i\in S_i] = \paren{\frac{1}{s(s-1)}}^d \ge \Omega_{s,d}(1)$.  Hence, if $u$ is a sufficiently large enough function of $s,d,\varepsilon$, by applying the Chernoff bound, we get that $$\Pr[|\{i\in [u]:(u_j,v_j)=(0,a_j)~\forall j\in S_i\}| \ge t] \le \varepsilon/2,$$ which in turn implies that $\Pr_{{\bf a}\sim \Z_s^n}[P({\bf a})=0] \le \varepsilon$.
\end{proof}

Finally, we finish the proof of the anti-concentration lemma.

\begin{proof}[Proof of~\Cref{lem:anti-conc}]
    The proof is by induction on $d$. 
    \paragraph{Base case $d=1$.} We take $r=u(s,1,\varepsilon)$, where $u(.,.)$ is given by the function in~\Cref{clm:large-matching}, so that if $P$ depends on $r$ variables, we are guaranteed that there are at least $r$ degree 1 pairwise disjoint monomials. Then, applying~\Cref{clm:large-matching}, we get $\Pr_{{\bf a}\in [s]^n}[P({\bf a}) \ne 0] \ge 1-\varepsilon$. 

    \paragraph{Induction step $d>1$.} The analysis is based on three cases.

    \begin{itemize}
        \item {\bf Case 1:} There exists a variable (say $x_1$ w.l.o.g.) and an index $j\in [s-1]$ such that in the junta-polynomial representation $$P(\x) = P_0(x_2,\dots,x_n) + \sum_{j=1}^{s-1} \delta_j(x_1) P_j(x_2,\dots,x_n),$$ $P_j$ depends on at least $r_1 = r(s,d-1,\varepsilon)$ variables. In this case, we note that for a random choice of $a_2,\dots,a_n \in [s]$, by applying the induction hypothesis to $P_j$ (which is a $(d-1)$-junta-sum), we have that $P_j(a_2,\dots,a_n) \ne 0$ with probability at least $\frac{1}{s^{d-2}}-\varepsilon$. Thus, the restriction of $P$ unto the variable $x_1$ is a non-constant function on setting $x_i=a_i$ for $i>1$. Therefore, we have $\Pr_{{\bf a}\sim [s]^n}[P({\bf a} \ne 0)] \ge \frac{1}{s}\cdot \paren{\frac{1}{s^{d-2}}-\varepsilon} \ge \frac{1}{s^{d-1}}-\varepsilon$.
        \item {\bf Case 2:} Suppose there exists $r_2 = u(s,d,1/2)$ many pairwise disjoint non-zero monomials of degree $d$ in $P$, where $u(.)$ is given by~\Cref{clm:large-matching}. Then, we immediately get $$\Pr_{{\bf a}\sim [s]^n}[P({\bf a}) = 0] \le \frac{1}{2} \le 1-\frac{1}{s^{d-1}}.$$
        \item {\bf Case 3:} Suppose neither Case 1 nor Case 2 occur. We now consider the set system 
$\Delta$ over $[n]$, where we include $S\in \Delta$ for $S\subseteq [n]$ if there exists ${\bf b}\in \Z_s^n$ such that the coefficient of $\delta_{{\bf b}}(\x)$ is non-zero in $P$, and $S$ is the set of non-zero indices of ${\bf b}$. Since there cannot be $r_2$ many $S \in \Delta$ that are pairwise disjoint and each $S\in \Delta$ is of size at most $d$, we can guarantee that there exists a small ``cover''; i.e., there exists indices $i_1,\dots,i_\ell \in [n]$ with $\ell \le d r_2$ such that for every degree $d$ non-zero monomial $\delta_{{\bf b}}(\x)$ in $P$, there exists $j\in [\ell]$ such that $b_{i_j} \ne 0$ (i.e., $x_{i_j}$ is contained in the corresponding monomial). We now count the number of monomials in $P$ which contain the variable $x_{i_j}$ for some $j\in [\ell]$. Since Case 1 does not occur, we can bound this by $(s-1) \cdot {(s-1)r_1\choose \le d}$, where the $s-1$  factor accounts for the number of monomials where $x_{i_j}$ appears as $\delta_{j'}(x_{i_j})$ for $j'\in [s-1]$, and the second factor $(s-1)r_1\choose \le d$ bounds the number of non-zero monomials of a function depending only on at most $r_1$ variables. Now, we set the variables $\{x_{i_j}:j\in [\ell]\}$ arbitrarily and show that the restricted function of $P$ is still non-zero. We note that once we set these variables, all the degree $d$ monomials would reduce in degree as the variables begin set form a ``cover'', thus we can bound the probability of the restriction of $P$ being non-zero as being at least $\frac{1}{s^{d-1}}$.  Hence, it only remains to prove that the restriction is non-zero. To see this, we set $r=r(s,d,\varepsilon) = 1+2\ell (s-1){(s-1)r_1 \choose \le d}$; this ensures that even after setting all the variables $x_{i_j}:j\in [\ell]$, there is at least one non-zero monomial in the restricted function. 
    \end{itemize}
\end{proof}

\subsection{Combinatorial Bound for $p$-primary groups}\label{sec:p-grp}

In this section, we prove~\Cref{thm:comb-small}. The high level proof approach again follows closely as that of the Boolean case $s=2$ from~\cite{ABPSS25}. The proof consists of the following steps:

\begin{itemize}
    \item The first step is to reduce the problem from general $p$-primary groups to the case of $\Z_p$. We prove a combinatorial bound for this case and ``lift'' it to the general case. 
    \item Then we show that that we can instead count polynomials over a field $\F_q$ (for some $q=O(s,p)$) rather than junta-sums. 
    \item In order to get the bound for the $\F_q$ case, we show that we can essentially assume without loss of generality that the polynomials in the list have pairwise disjoint leading monomials.
    \item Finally, we show a tail bound for the roots of polynomials with pairwise disjoint leading monomials, which results in a list size bound. 
\end{itemize}

We divide the proof into two subsections; we prove the first two items above in~\Cref{subsec:fq-redn} and the next two items in~\Cref{subsec:fq}.

\subsubsection{Reducing to the Case of Constant-sized Field $\F_q$}\label{subsec:fq-redn}

For a field $\F$ and a subset $S \subseteq \F$ of size $|S| = s\ge 2$, we note that $\J_d(S^n, \F)$ is exactly the family of functions that can be uniquely expressed as a polynomial where each (non-zero) monomial has at most $d$ variables and individual degree at most $(s-1)$ in each variable. For the remainder of this subsection, we will use this interpretation. We show next that we can always assume that $S \subseteq \F$ (and then use the polynomial interpretation) without much loss in parameters for the combinatorial bound.  \\

\begin{lemma}[{\bf Reducing counting junta-sums to polynomials}]\label{lem:small-1}
    If $\J_d(S^n, \F_q)$ is $(1/s^d-\varepsilon, O_{q,\varepsilon}(1))$-list-decodable for every $\varepsilon>0$ and every finite field $\F_q$ and subset $S\subseteq \F_q$ of size $s$, then $\J_d([s]^n, \Z_p)$ is also $(1/s^d-\varepsilon, O_{p,\varepsilon}(1))$-list-decodable for every prime $p$ and every $\varepsilon>0$.
\end{lemma}
\begin{proof}[Proof]
    Fix an arbitrary prime $p$ and let $q$ be the smallest power of $p$ that is at least $s$. Let $f:[s]^n \to \Z_p$ be arbitrary and $P_1,\dots,P_t \in \J_d([s]^n,\Z_p)$ be distinct junta-sums such that $\delta(f,P_i) \le \frac{1}{s^d}-\varepsilon$ for $i\in [t]$. We will prove that $t=O_{q,\varepsilon}(1)$ assuming that $\J_d(S^n, \F_q)$ is $(1/s^d-\varepsilon, O_{q,\varepsilon}(1))$-list-decodable. We shall identify $\Z_p$ with a subgroup of $\F_q$ of order $p$: in particular, let $H \subseteq \F_q$ be a subgroup of $\F_q$ that is homomorphic to $\Z_p$, via a group homomorphism $\sigma : \Z_p \to H$. 
    Let $\phi:[s]\to S$ be an arbitrary bijection and let $g:S^n \to \F_q$ be defined by $g(\x) = \sigma(f(\phi^{-1}(\x)))$. Similarly, for $i\in [t]$, let $Q_i:S^n \to \F_q$ be defined by $Q_i(\x) = \sigma(P_i(\phi^{-1}(\x)))$. We claim that $Q_i \in \J_d(S^n, \F_q)$: indeed, if $P_i(\x) = \sum_{I \in {[n]\choose \le d}} P_{i,I}(\x_I)$ for functions $P_{i,I}:[s]^I \to \Z_p$, then we have $Q_i(\x) = \sigma\paren{\sum_{I\in {[n]\choose \le d}} P_{i,I}(\phi^{-1}(\x_I))} = \sum_{I\in {[n]\choose \le d}} \sigma(P_{i,I}(\phi^{-1}(\x_I)))$. Moreover, $\delta(f,P_i) = \delta(g,Q_i)$ and $Q_i$'s are pairwise distinct functions since $\sigma$ and $\phi$ are bijections. By our assumption that $\J_d(S^n, \F_q)$ is $(1/s^d-\varepsilon, O_{q,\varepsilon}(1))$-list-decodable, we get that $t\le O_{q,\varepsilon}(1)=O_{p,\varepsilon}(1)$. \\
\end{proof}

\begin{lemma}[{\bf Lifting the bound to general $p$-primary groups}]\label{lem:small-11}
    If $\J_d([s]^n, \Z_p)$ is $(1/s^d-\varepsilon, O_{p,\varepsilon}(1))$-list-decodable for every $\varepsilon>0$ and every prime $p$, then $\J_d([s]^n, G)$ is also $(1/s^d-\varepsilon, O_{p,\varepsilon}(1))$-list-decodable for every finite $p$-primary group $G$ and every $\varepsilon>0$.
\end{lemma}

\begin{proof}[Proof]
    The proof essentially follows the same outline as that from~\cite{ABPSS25} which handles $s=2$. For arbitrary fixed $\varepsilon > 0$, let $L$ be the list size; i,e., for every $g:[s]^n \to \Z_p$, there exists at most $L \le O_{p,\varepsilon}(1)$ junta-sums $Q\in \J_d([s]^n, \Z_p)$ such that $\delta(g,Q) \le 1/s^d -\varepsilon$.  We will now show that for an arbitrary $f:[s]^n \to G$, that the number of junta-sums $P\in \J_d([s]^n, G)$ such that $\delta(f,P) \le 1/s^d - \varepsilon$ is at most $L^{O(\log(1/\varepsilon))}$, thus giving the required bound. Using the notation $L_\varepsilon(f)$ to denote the set of junta-sums $P\in \J_d([s]^n, G)$ such that $\delta(f,P) \le 1/s^d - \varepsilon$, our goal now is to prove an upper bound on $|L_\varepsilon(f)|$. In order to prove this, we will need the following setup. Since $G$ is a finite $p$-primary group, there exists an element $h_0\in G$ of order $p$; let $H_0 \subseteq G$ be the subgroup generated by $h_0$. We then note that the quotient group $G/H_0$ is again a $p$-primary group. By continuing this argument, we have a sequence of groups $G=G_0,G_1,\dots,G_h$ for some $h\in \mathbb{N}$ such that $G_{i+1} = G_i/H_i$, where $H_i \subseteq G_i$ is a subgroup of order $p$ (generated by some $h_i \in G_i$) and $G_h$ is the trivial group containing just the identity element. Now, we let $f_0=f$ and for $0\le i\le h$, we define $f_i:[s]^n \to G_i$ by the recurrence $$f_{i+1}(\x) = f_i(\x)~\mod H_i.$$ 
    
    We now define a rooted tree $T$ as follows: there are $h+1$ levels of the tree, with the root being level $h$ and the leaves being level $0$. We now describe the vertices and their labels bottom-up. The vertices in level 0 are in bijection with the junta-sums $L_\varepsilon(f)$ (we treat these junta-sums as the ``labels'' of the vertices). For a vertex with label $P_0\in L_\varepsilon(f)$ in level 0, we let $P_1\in \J_d([s]^n, G_1)$  defined by $$P_1(\x) = P_0(\x) \mod H_0$$ be the label of the parent of this vertex: we note that $P_1\in L_\varepsilon(f_1)$ since if $f$ and $P_0$ agree, so do $f_1$ and $P_1$.  Proceeding in a similar way, we construct all the above levels of the tree $T$ and label its vertices. In particular, the parent of a vertex in level $i$ labeled with $P_i \in L_\varepsilon(f_i)$ is set to be the junta-sum $P_{i+1} \in L_{\varepsilon}(f_{i+1})$ defined as:
    $$P_{i+1}(\x) = P_i(\x) \mod H_i.$$
    For a vertex $v$ of $T$ at level $i\in [0..h]$ and labeled with $P_i \in L_\varepsilon(f_i)$, we let $$\rho(v) = 1/s^d-\delta(f_i,P_i).$$
    Note that $\rho(v) \ge \varepsilon$ for all vertices $v$ of $T$. We further show the following properties of $\rho(\cdot)$.
    \begin{claim}\label{clm:tree}
        For the tree $T$ and the function $\rho$ defined over the vertices of $T$ defined above, the following properties hold:
        \begin{itemize}
            \item Each non-leaf vertex of $T$ has at most $L$ children. 
            \item If $u$ is the parent of $v$, then $\rho(u) \ge \rho(v)$.
            \item If $u$ has two distinct children $v$ and $w$, then $\rho(u) \ge \rho(v) + \rho(w)$.
        \end{itemize}
    \end{claim}
    We now finish the proof of~\Cref{lem:small-11} using the above claim. We recall that the number of leaves in $T$ is exactly $|L_\varepsilon(f)|$, which is what we want to upper bound. To do this, we argue that for any non-leaf node $u$ of $T$ with children $v_1,\dots,v_t$ (for some $1\le t\le L$), it holds that 
    \begin{align}\label{eqn:tree-rn} \rho(u)^\ell \ge \sum_{i\in [t]} \rho(v_i)^\ell, \end{align}
    where $\ell = \lceil \log L\rceil$. Then, applying this inequality for all the non-leaf vertices of the tree, we get that $$\rho(\text{root})^\ell \ge \sum_{v\text{~is a leaf}} \rho(v)^\ell \ge (\text{\#leaves})\cdot \varepsilon^\ell.$$ Using $\rho(\text{root}) \le 1$, we thus get that $|L_\varepsilon(f)| = \text{(\# leaves}) \le (1/\varepsilon)^\ell = L^{O(\log (1/\varepsilon))} = O_{p,\varepsilon}(1)$ as desired. Hence, it only remains to show that~\eqref{eqn:tree-rn} holds. For this, we will assume that $t$, the number of children of $u$ is at least $2$ as otherwise, we immediately have $\rho(u)^\ell \ge \rho(v_1)^\ell$ using Item 2 of~\Cref{clm:tree}. Further, let $\rho(v_1) \ge \rho(v_2) \ge \dots \ge \rho(v_t)$ without loss of generality. Then using Item 3 of~\Cref{clm:tree} and $t\le L \le 2^\ell$, we have 
    \begin{align*}\rho(u)^\ell & \ge (\rho(v_1)+\rho(v_2))^\ell\\&  \ge \rho(v_1)^\ell + (2^\ell - 1)\rho(v_2)^\ell\\& \ge \rho(v_1)^\ell + (t-1)\rho(v_2)^\ell \\ & \ge \rho(v_1)^\ell + \rho(v_2)^\ell + \dots + \rho(v_t)^\ell.
    \end{align*}
    \end{proof}
We now prove~\Cref{clm:tree}. 

\begin{proof}[Proof of~\Cref{clm:tree}]
    Let $u$ be an arbitrary non-leaf vertex of $T$ at level $i+1$ (for some fixed $i\in [0..h-1]$), with children $v_1,\dots,v_t$. Suppose $u$ is labeled by a junta-sum $P\in L_\varepsilon(f_{i+1})$ and $v_j$ is labeled by a junta-sum $Q_j \in L_\varepsilon(f_{i})$ for $j\in [t]$. Therefore, for all $j\in [t]$, we have that \begin{align}\label{eqn:pq}P(\x) = Q_j(\x) \mod H_i.\end{align} Hence, if $f_{i}$ and $Q_j$ agree on some input, so do $f_{i+1}$ and $P$; so $\delta(f_{i+1}, P) \le \delta(f_i, Q_j)$ and $\rho(u) \ge \rho(v_j)$, thus proving Item 2. Our goal now is to show that $t\le L$ and $\rho(u) \ge \rho(v_{j_1}) + \rho(v_{j_2})$ for $j_1\ne j_2 \in [t]$. To do this, we let $c_1,c_2,\dots,c_M \in G_i$ be fixed coset representatives (where $M=|G_i|/p$ and the cosets are ordered arbitrarily) corresponding to the subgroup $H_i$ of $G_i$. Then each element $g\in G_i$ can be uniquely written as $g=g'+\widehat{g}$ with  $g'\in H_i$ and $\widehat{g} \in \{c_1,\dots,c_M\}$ being a coset representative. 
    
    Let $$Q_j(\x) = \sum_{{\bf a}\in \Z_s^n:|{\bf a}|\le d} g_{j,{\bf a}} \cdot \delta_{{\bf a}}(\x),$$ for some $g_{j,{\bf a}}\in G_i$. From~\eqref{eqn:pq}, we see that $\widehat{{g}_{{j},{\bf a}}} =\widehat{{g}_{{1},{\bf a}}}  $ for all $j\in [t]$. Now we define $\widehat{Q}:[s]^n \to G_i$ to be:
    $$\widehat{Q}(\x) = \sum_{{\bf a}\in \Z_s^n:|{\bf a}|\le d} \widehat{g_{1,{\bf a}}} \cdot \delta_{{\bf a}}(\x),$$ and $d$-junta-sums $\widetilde{Q}_j \in \J_d([s]^n,H_i)$ for $j\in [t]$, to be:
    $$\widetilde{Q}_j(\x) = \sum_{{\bf a}\in \Z_s^n:|{\bf a}|\le d} g_{j,{\bf a}}' \cdot \delta_{{\bf a}}(\x).$$
    
    Since $Q_j(\x) = \widehat{Q}(\x)+\widetilde{Q}_j(\x)$ and $Q_j$ are pairwise distinct for $j\in [t]$, we have that $\widetilde{Q}_j$ are pairwise distinct for $j\in [t]$. Moreover for the function $\widetilde{f}:[s]^n \to G_i$ defined as $\widetilde{f}(\x) = f_i(\x) - \widehat{Q}(\x)$, we have that $\delta(\widetilde{f},\widetilde{Q}_j) = \delta(f_i,Q_j) \le 1/s^d - \varepsilon$. Therefore, we get $t\le L$ as $H_i$ is isomorphic to $\Z_p$ and we have a list size bound of $L$ for junta-sums over $\Z_p$. This proves Item 1 of the claim. To prove Item 3, let $j_1 \ne j_2\in [t]$ be arbitrary and let $A_1, A_2 \subseteq [s]^n$ be the subset of points where $f_i$ agrees with $Q_{j_1}$ and $Q_{j_2}$ respectively. Let $A \subseteq [s]^n$ be the subset of points where $f_{i+1}$ agrees with $P$. From the proof of Item 2, we have that $A_1,A_2\subseteq A$. Since two distinct $d$-junta-sums cannot agree on more than $1-1/s^d$ fraction of inputs (\Cref{clm:dist-junta-sums}), we have $|A_1 \cap A_2| \le (1-1/s^d)s^n$. Hence, $|A| \ge |A_1 \cup A_2| =  |A_1| + |A_2| - |A_1 \cap A_2| \ge s^n \paren{2-\delta(f_i,Q_{j_1})-\delta(f_i,Q_{j_2})-1+1/s^d)}$. Since $|A| = s^n\paren{1-\delta(f_{i+1},P)}$, we get $\delta(f_{i+1},P) \le \delta(f_i,Q_{j_1})+\delta(f_i,Q_{j_2})-1/s^d$, or equivalently $\rho(u) \ge \rho(v_{j_1})+\rho(v_{j_2})$.
\end{proof}

Having reduced the problem to showing combinatorial bound over $\F_q$, which we state below and prove in the next subsection.\\

\begin{theorem}\label{lem:fq-case}
    For every $\varepsilon>0$, finite field $\F_q$ and subset $S\subseteq \F_q$ of size $s\ge 2$, the family $\J_d(S^n, \F_q)$ is $(1/s^d-\varepsilon, O_{q,\varepsilon}(1))$-list-decodable. 
\end{theorem}

With the above theorem, we can now finish the proof of the combinatorial bound for $p$-primary groups.

\begin{proof}[Proof of~\Cref{thm:comb-small}]
    The proof follows by combining~\Cref{lem:small-11},~\Cref{lem:small-1} and~\Cref{lem:fq-case}.
\end{proof}

\subsubsection{Combinatorial Bound for $\F_q$}\label{subsec:fq}

Throughout this subsection, we fix a finite field $\F_q$ and a subset $S\subseteq \F_q$ of size $s$ arbitrarily. We think of $q\ge 2,s\ge 2$ and $d\ge 1$ as constants. Furthermore, we fix a monomial ordering (denoted $\preceq$) over the monomials to be the {\em graded lexicographic order} (see~\cite{ABPSS25} for a definition) and denote by $\textrm{LM}(P)$ the leading monomial of a  polynomial $P$ (assuming it is non-zero). We will use the notation $m_1 \succeq m_2$ to mean that $m_2 \preceq m_1$ and $m_1 \succneq m_2$ to mean that $m_2 \preceq m_1$ and $m_1 \ne m_2$. \\

We show the following lemma which effectively reduces the list-decoding problem to bounding the number of polynomials in the list with pairwise distinct monomials.\\

\begin{lemma}[{\bf Distinct leading monomials}]\label{lem:small-2}
    If $P_1,\dots,P_t \in \J_d(S^n, \F_q)$ are such that $\delta(f,P_i) \le 1/s^d - \varepsilon$ for all $i\in [t]$ and some $f:S^n \to \F_q$, then there exists a function $f':S^n \to \F_q$ such that there are at least $\ell \ge \Omega({\log_q t})$ many polynomials $Q_1,\dots,Q_\ell \in \J_d(S^n , \F_q)$ such that $\delta(f',Q_i) \le 1/s^d - \varepsilon$ and $\textrm{LM}(Q_i)$ are pairwise distinct for $i\in [\ell]$. \\
\end{lemma}

Then, we prove the following ``tail bound'' for polynomials with {\em pairwise disjoint} leading monomials. \\

\begin{lemma}[{\bf Tail bound for disjoint leading monomials}]\label{lem:small-4}
    Let $P_1,\dots,P_t \in \J_d(S^n , \F_q)$ be such that $\textrm{LM}(P_i)$ are pairwise disjoint for $i\in [t]$. Then:
    $$\Pr_{{\bf a}\sim S^n}\bigg[\bigg|\{i\in [t]:P_i({\bf a})=0\}\bigg| \ge (1-1/s^d+\eta)t\bigg] \le \exp(-\Omega(\eta^2 t)).$$
\end{lemma}

With the above lemmas in place, we are ready to finish the proof of the main result of this subsection i.e., \Cref{lem:fq-case}. 

\begin{proof}[Proof of~\Cref{lem:fq-case}]
    Using~\Cref{lem:small-2}, it suffices to show that if $Q_1,\dots,Q_\ell \in \J_d(S^n , \F_q)$ are such that $\LM(Q_i)$ are pairwise distinct for $i\in [\ell]$ and $\delta(Q_i,f)\le 1/s^d-\varepsilon$ for some $f:S^n \to \F_q$, then that $\ell=O_{q,\varepsilon}(1)$. Applying the sunflower lemma (see e.g.~\cite{erdos1960intersection}) for the multisets determined by the leading monomials of $Q_i$'s, we can find a subset of indices $i_1,\dots,i_{\ell'}$ for some $\ell' \ge \Omega_d(\ell^{1/d})$ such that $\LM(Q_{i_j})$ form a sunflower for $j\in [\ell']$. That is, there exists variables $\{x_j:j\in C\}$ and $\{e_j\in \Z_s:j\in C\}$ where $C\in {[n]\choose \le d}$ such that $\LM(Q_{i_j}) = \prod_{j'\in C}x_{j'}^{e_{j'}} \cdot m_{j}$ and $m_j$ are monomials over the variables indexed by $[n]\setminus C$ and are pairwise disjoint for $j\in [\ell']$. Without loss of generality, we will assume that $i_j=j$ for all $j\in [\ell']$. We will now express each $Q_i$ for $i\in [\ell']$ (uniquely) as follows:
    $$Q_i(\x) = \prod_{j\in C}x_j^{e_j} \cdot Q_i^{(1)}(\x_{[n]\setminus C}) + Q_i^{(2)}(\x),$$ where $Q_i^{(1)}$ is a polynomial over variables indexed by $[n]\setminus C$ and $Q_i^{(2)}$ does not contain any monomial dividing $\prod_{j\in C} x_j^{e_j}$. We note that since the leading monomial of $Q_i$ is $\prod_{j\in C}x_j^{e_j}\cdot m_i$, by our definition of monomial ordering, we must have that $\LM(Q_i^{(1)}) = m_i$. Let ${\bf a} \sim S^n$ be sampled by first choosing ${\bf a}'\sim S^{[n]\setminus C}$ uniformly at random and then ${\bf a}''\sim S^{C}$ uniformly and independently. Letting $d'=d-|C|$, we may now apply the tail bound (\Cref{lem:small-4}) to $Q_i^{(1)}$ to get that 
    $$\Pr_{{\bf a}'}\bigg[\bigg| \{i\in [t]:Q_i^{(1)}({\bf a}') = 0\} \bigg| \ge (1-1/s^{d'}+\varepsilon/2)\ell'\bigg] \le \exp(-\Omega(\varepsilon^2 \ell')).$$ In fact, by applying it to $Q_i^{(1)}-\alpha$ for $\alpha \in \F_q$ and by a union bound, we get that
    \begin{align}\label{eqn:alpha}\Pr_{{\bf a}'}\bigg[\exists \alpha\in \F_q \text{~such that~}\bigg| \{i\in [t]:Q_i^{(1)}({\bf a}') = \alpha\} \bigg| \ge (1-1/s^{d'}+\varepsilon/2)\ell'\bigg] \le q\cdot \exp(-\Omega(\varepsilon^2 \ell')).\end{align}
    Now, let us use the notation $$Q_i'({\x}_C) = Q_i(\x_C, {\bf a}')$$ to denote the corresponding restricted functions obtained by setting the variables in ${[n]\setminus C}$ to ${\bf a}'$. Similarly, let $f':S^C\to \F_q$ be the restricted function $f'(\x_C) = f(\x_C, {\bf a}')$. For a uniformly random choice of ${\bf a}'$, let ${\mathcal{B}}$ denote the ``bad'' event that the {\em multiset} of functions $\{Q_i':i\in [\ell']\}$ has a function occurring at least $(1-1/s^{d'}+\varepsilon/2)\ell'$ many times. The bound from~\Cref{eqn:alpha} immediately implies that 
    \begin{align}\label{eqn:probb}\Pr_{{\bf a}'}[\mathcal{B}] \le q\cdot \exp(-\Omega(\varepsilon^2\ell')).\end{align}
    Conditioned on $\mathcal{B}$ not occurring, we note that there are at least $(1/s^{d'}-\varepsilon/2)\ell'$ indices $i\in [\ell']$ such that $Q_i' \ne f'$ as functions over $S^{C}$. More formally, 
    $$\Pr_{i\sim[\ell']}\bigg[Q_i'\ne f'~\mid~\overline{\mathcal{B}}\bigg] \ge 1/s^{d'}-\varepsilon/2.$$ Since two different functions over $|C|$ variables must differ on a random input with probability at least $1/s^{|C|}$, we further get:
    \begin{align}\label{eqn:bbar}
        \Pr_{i\sim [\ell'],{\bf a}''}\bigg[ Q_i'({\bf a}'')\ne f'({\bf a}'') ~\mid~\overline{\mathcal{B}}\bigg] \ge \paren{\frac{1}{s^{d'}}-\frac{\varepsilon}{2}}\frac{1}{s^{|C|}} \ge \frac{1}{s^d} - \frac{\varepsilon}{2}.
    \end{align}
    Combining~\eqref{eqn:probb} and~\eqref{eqn:bbar}, we obtain $$\Pr_{i\sim [\ell'], {\bf a}\sim S^n}\bigg[ Q_i({\bf a}) \ne f({\bf a}) \bigg] \ge \frac{1}{s^d}-\frac{\varepsilon}{2}-\frac{q}{2^{\Omega(\varepsilon^2\ell')}}.$$ However, we note that since $\delta(Q_i,f) \le 1/s^d-\varepsilon$ for all $i\in [\ell']$, the left hand side of the above inequality must be at most $\frac{1}{s^d}-\varepsilon$. Put together, they give the required bound of $\ell'=O_{q,\varepsilon}(1)$, and thus $\ell=O_{q,\varepsilon}(1)$.
\end{proof}

We now give the proofs of the above two lemmas. First, we start with the reduction to counting polynomials with distinct leading monomials, i.e.,~\Cref{lem:small-2}.

\begin{proof}[Proof of~\Cref{lem:small-2}]
    Let $\ell$ be an integer such that $t\in [q^\ell,q^{\ell+1})$ (so we have $\ell \ge \Omega(\log_q t)$). We will prove the following inductive claim. We recall that $L_\varepsilon(f)$ denotes the set of polynomials in $\J_d(S^n , \F_q)$ that are at distance at most $1/s^d-\varepsilon$ from the function $f$.

    {\bf Inductive claim.} For every $0\le i\le \ell$, there exists a function $f_i:S^n\to \F_q$, polynomials $Q_1,Q_2,\dots,Q_i \in \J_d(S^n , \F_q)$, and a set of polynomials $\Q_i \subseteq \J_d(S^n,  \F_q)$ such that:
    \begin{itemize}
        \item $Q_1,\dots,Q_i \in L_\varepsilon(f_i)$ and $Q\in L_\varepsilon(f_i)$ for all $Q\in \Q_i$,
        \item $\LM(Q_1) \succneq \LM(Q_2) \succneq \dots \succneq \LM(Q_i) \succneq \LM(Q)$ for all $Q\in \Q_i$, and
        \item $|\Q_i|\ge q^{\ell-i}$.
    \end{itemize}

    We note that the base case $i=0$ is true with $f_0=f$ and $\Q_0 = \{P_1,P_2,\dots,P_t\}$. And proving the inductive claim for $i=\ell$ finishes the proof of~\Cref{lem:small-2}. We now assume the inductive claim holds for a fixed $i<\ell$ and prove it for $i+1$. 

    Let $P \in \J_d(S^n , \F_q)$ be the ``plurality polynomial'' of $\Q_i$, i.e., we determine each coefficient of $P$ by taking a plurality vote of the corresponding coefficients from the polynomials in $\Q_i$ (by breaking ties arbitrarily). We then define $f_{i+1}:S^n \to \F_q$ to be $$f_{i+1} = f_i - P.$$ We now define $Q_1',\dots,Q_i',Q'_{i+1} \in \J_d(S^n , \F_q)$ and $\Q_{i+1}' \subseteq \J_d(S^n,   \F_q)$ such that the three items of the inductive claim hold for them. We let $\Q_i' = \{Q-P:Q\in \Q_i\}$ and set
    $Q_j' = Q_j - P$ for $j\in [i]$. We now set $Q_{i+1}'$ to be a polynomial from $\Q_i'$ with the greatest leading monomial (ignoring the zero polynomial if it exists and breaking ties arbitrarily). Then we set $\Q_{i+1}'$ to be the subset of polynomials in $\Q_i'$ with leading monomial strictly smaller than that of $Q'_{i+1}$, i.e.: $$\Q_{i+1}'=\{Q' \in \Q_i' : \LM(Q') \precneq \LM(Q'_{i+1})\}.$$ It remains to prove that the three conditions of the inductive claim actually hold for the above definitions.
    \begin{itemize}
        \item We have that $\delta(Q_j',f_{i+1}) = \delta(Q_j-P,f_i-P) = \delta(Q_j,f_i) \le 1/s^d-\varepsilon$, therefore $Q_j'\in L_\varepsilon(f_{i+1})$ for all $j\in [i]$. Similarly, we have $Q_{i+1}'\in L_{\varepsilon}(f_{i+1})$ and $Q'\in L_{\varepsilon}(f_{i+1})$ for all $Q'\in \Q'_{i+1}$.
        \item We note that $\LM(P) \precneq \LM(Q_{i})$ since the coefficients of all monomials $m \succeq \LM(Q_i)$ in all $Q\in \Q_i'$ (and hence in $P$) are zero by the induction hypothesis. Therefore, $\LM(Q_j') = \LM(Q_j)$ for $j\in [i]$ and we have $\LM(Q_1') \succneq \LM(Q_2') \succneq \dots \succneq \LM(Q_i')$. It also follows that $\LM(Q_i') \succneq \LM(Q_{i+1}')$ and $\LM(Q_{i+1}') \succneq \LM(Q)$ for all $Q\in \Q_{i+1}'$ by the definitions of $Q'_{i+1}$ and $\Q'_{i+1}$.
        \item We have that $|\Q_i'| = |\Q_i| \ge q^{\ell-i}$ by induction hypothesis. By the definition of $Q_{i+1}'$, we observe that $\LM(Q') \preceq \LM(Q_{i+1}')$ for all $Q'\in \Q_i'$ and we will show that at least $1/q$ fraction of $Q'$'s have leading monomial {\em strictly} smaller than that of $Q_{i+1}'$. By the nature of the construction of $P$ using the plurality vote, we observe that at least $1/q$ fraction of the polynomials $Q\in \Q_i'$ agree with $P$ on the coefficient $\LM(Q_{i+1}')$. The corresponding polynomials $Q'= Q-P$ have coefficient of $\LM(Q_{i+1})$ as zero. In other words there are at least $|\Q_i|/q$ polynomials $Q'\in \Q_i'$ with leading coefficient strictly smaller than $\LM(Q_{i+1}')$, and hence by our definition of $\Q_{i+1}'$, it must be of size $|\Q'_{i+1}|\ge |\Q_i'|/q \ge q^{\ell-{(i+1)}}$.
    \end{itemize}
    This finishes the proof of the inductive claim. 
\end{proof}

We now prove the tail bound for polynomials with pairwise disjoint leading monomials (\Cref{lem:small-4}).

\begin{proof}[Proof of~\Cref{lem:small-4}]
    We will use the following theorem of Panconesi and Srinivasan~\cite{PS-chernoff} which reduces the task of showing tail bounds to proving a certain ``independence'' relation among the events.

    \begin{theorem}[\cite{PS-chernoff} Theorem 3.4]\label{thm:ps-chernoff}
        Let $Z_1,\dots,Z_t$ be Boolean random variables and $\alpha \in [0,1]$ be such that for every subset $S \subseteq [t]$, we have that $\Pr[\bigwedge_{i\in S} Z_i = 1] \le \alpha^{|S|}$. Then, for every $\eta>0$, we have
        $$\Pr\bigg[\sum_{i\in [t]} Z_i \ge (\alpha + \eta) t\bigg] \le \exp(-\Omega(\eta^2 t)).$$
    \end{theorem}

    Because of the above theorem, it suffices to show the following: for every $t\in \mathbb{N}$ and non-zero polynomials $P_1,\dots,P_t \in \J_d(S^n,   \F_q)$ for which $\LM(P_i)$ are pairwise disjoint for $i\in [t]$, that:
    \begin{align}\label{eqn:com-zeroes}
    \Pr_{{\bf a}\sim S^n}\bigg[P_i({\bf a}) = 0~\text{~for all~}i\in [t]\bigg] \le \paren{1-\frac{1}{s^d}}^t.
    \end{align}
    The proof follows the same {\em footprint bound} technique as used in~\cite{ABPSS25}. In particular, letting $$Z = \{{\bf a}\in S^n:P_i({\bf a})=0\text{~for all~}i \in [t]\}$$ denote the set of common zeroes, we will prove an upper bound on the dimension of all functions from $Z$ to $\F_q$ i.e., $|Z|$. Let $f:Z\to \F_q$ be an arbitrary function. We will show that it can be expressed as a polynomial of individual degree at most $s-1$ over $\F_q$, without using any monomial divisible by any of the $\LM(P_i)$ for $i\in [t]$. That is, these are the monomials ${\bf x}^{{\bf e}}$ for ${\bf e} \in E$, where $E \subseteq \Z_s^n$ is defined below, and $\LM(P_i) = {\bf x}^{{\bf m}_i}$ for some ${\bf m}_i \in \Z_s^n$\footnote{Here we treat $\Z_s^n=\{0,1,\dots,s-1\}^n$ as a subset of $\Z^n$ and define the monomial ${\x}^{{\bf m}} = \prod_{i\in [n]}x_i^{m_i}$.}:
    $$E = \{{\bf e} \in \Z_s^n: \forall i\in [t]~\exists j\in [n]~e_j < m_{i,j}\}.$$
    Using the fact that the supports of ${\bf m}_i$ are pairwise disjoint (over $i\in [t]$) and are of size at most $d$, we get that $|E| \le \paren{1-\frac{1}{s^d}}^t \cdot s^n$.
    We will show that there exists field elements $(c_{{\bf e}})_{{\bf e}\in E}$ such that  $f({\bf x}) = \sum_{{\bf e}\in E} c_{{\bf e}}\cdot {\bf x}^{{\bf e}}$ for all ${\bf x}\in Z$. This shall finish the proof of~\eqref{eqn:com-zeroes} and thus~\Cref{lem:small-4} as it shows that $|Z| \le |E| \le \paren{1-\frac{1}{s^d}}^t \cdot s^n$. Hence, it remains to prove that $f$ can be expressed as a linear combination of monomials in $E$. Since $Z\subseteq S^n$, we know that there exists a polynomial representation for $f$ of individual degree at most $s-1$: suppose $Q(\x) = \sum_{{\bf e}\in \Z_s^n} c_{{\bf e}}\cdot {\x}^{{\bf e}}$ for some $c_{{\bf e}}\in \F_q$ is such that $f({\bf x}) = Q(\x)$ for all $\x\in Z$. If $c_{{\bf e}} = 0$ for all ${\bf e}\notin E$, we are done. Otherwise, there exists an $i\in [t]$ such that $\LM(P_i)$ divides ${\bf x}^{{\bf e}'}$ for some ${\bf e}'$ such that $c_{{\bf e}'} \ne 0$ (say that $\x^{{\bf e}'} = \LM(P_i) \cdot \x^{{\bf e}''}$); w.l.o.g.~let ${\x}^{{\bf e}'}$ be the largest monomial in the monomial ordering such that this holds. Then, we note that we can replace the monomial $\x^{{\bf e}'}$ with the polynomial $\x^{{\bf e}''}(\LM(P_i)-P_i/c)$ in the polynomial  $\sum_{{\bf e}\in \Z_s^n} c_{{\bf e}}\cdot {\x}^{\bf e}$ while still computing $f$, where $c\in\F_q^\times$ is the coefficient of $\LM(P_i)$ in $P_i$. This is due to the fact that $P_i$ (and therefore $P_i/c$) evaluates to $0$ over $Z$. Let $Q'$ denote the polynomial obtained by such a transformation. We claim that $\LM(Q') \precneq \LM(Q)$. This is because all the new non-zero monomials introduced by the transformation are of the form $\x^{{\bf e}''}\cdot \x^{{\bf e}'''}$ for some $\x^{{\bf e}'''} \precneq \LM(P_i)$, and so $\x^{{\bf e}''}\cdot \x^{{\bf e}'''} \precneq \x^{{\bf e}''}\cdot \LM(P_i) = \x^{{\bf e}'} = \LM(Q)$ using the monomial ordering property. Hence, $\LM(Q') \precneq \LM(Q)$. While the leading monomial of the polynomial computing has decreased, it may be possible that $Q'$ contains monomials with individual degree at least $s$. We now argue that we can design a new polynomial $Q''$ such that $\LM(Q'') \preceq \LM(Q')$ and $Q''(\x)=Q'(\x)=f(\x)$ for all $\x\in Z$. The idea is to use the equation $\prod_{a\in S}(x_i-a)=0$ to replace the powers of the variable $x_i$ greater than $s-1$ with smaller powers --- this only results in monomials that are smaller in the monomial order. Thus, by repeating the above two steps for a finite number of times, we will have a polynomial representing $f$ only using monomials from $E$.
\end{proof} 

\addtocontents{toc}{\protect\setcounter{tocdepth}{1}}
	
\end{document}